\documentclass[a4paper,11pt]{amsart}
\usepackage{amssymb}
\usepackage{amscd}
\usepackage{comment}
\usepackage{amsmath,amsthm}
\usepackage[colorlinks=true]{hyperref}
\usepackage{enumerate}
\usepackage{booktabs,multirow}
\usepackage{tikz}
\usepackage{rotating}
\usetikzlibrary{patterns}
\usetikzlibrary{decorations.pathreplacing}
\usetikzlibrary{calc,through}

\allowdisplaybreaks[1]
\numberwithin{figure}{section}
\numberwithin{equation}{section}

\title[Ballot tilings and increasing trees]
{Ballot tilings and increasing trees}
\author[K.~Shigechi]{Keiichi~Shigechi}
\email{k1.shigechi AT gmail.com}
\date{\today}

\newcommand\tikzpic[2]{
\raisebox{#1\totalheight}{
\begin{tikzpicture}
#2
\end{tikzpicture}
}}
\newtheorem{theorem}[figure]{Theorem}
\newtheorem{example}[figure]{Example}
\newtheorem{lemma}[figure]{Lemma}
\newtheorem{defn}[figure]{Definition}
\newtheorem{prop}[figure]{Proposition}
\newtheorem{cor}[figure]{Corollary}

\newtheorem{remark}[figure]{Remark}
\begin{document}
\begin{abstract}
We study enumerations of Dyck and ballot tilings, which are 
tilings of a region determined by two Dyck or ballot paths.
We give bijective proofs to two formulae of enumerations of Dyck tilings through 
Hermite histories. We show that one of the formulae is equal to a 
certain Kazhdan--Lusztig polynomial.
For a ballot tiling, we establish formulae which are analogues of formulae
for Dyck tilings. 
Especially, the generating functions have factorized expressions.
The key tool is a planted plane tree and its increasing labellings. 
We also introduce a generalized perfect matching which is bijective to 
an Hermite history for a ballot tiling. 
By combining these objects, we obtain various expressions of a generating 
function of ballot tilings with a fixed lower path.

\end{abstract}

\maketitle

\section{Introduction}
Cover-inclusive Dyck tilings were introduced to study Kazhdan--Lusztig 
polynomials for the Grassmannian permutations in \cite{SZJ12}. Independently, 
Dyck tilings reappeared in the study of the double-dimer model in \cite{KW11}.
One of the main purposes of this paper is to give bijective proofs of 
two conjectures by Kenyon and Wilson in \cite{KW11} (proven by Kim in \cite{K12}).
The first formula is associated with a Dyck tiling with a fixed lower path
and the second one is associated with a Dyck tiling with a fixed upper path.
In both formulae, an Hermite history which is bijective to a Dyck tiling
plays a central role.
We show that the first formula is equal to a certain Kazhdan--Lusztig polynomial 
by providing a bijection between a term in a Kazhdan--Lusztig polynomial 
and an Hermite history.

In the study of Kazhdan--Lusztig polynomials, 
one can study Hermitian symmetric pairs instead of Grassmannian permutations. 
This replacement corresponds to moving from type A to other types in the 
classification of Weyl groups.
Cover-inclusive ballot tilings appeared in the case of the Hermitian symmetric
pair $(B_{N},A_{N-1})$ in \cite{S14}.
A ballot tiling can be regarded as a generalization of a Dyck tiling 
from type A to type B.
Regarding ballot tilings, we have two types of tilings: type BI and type BIII.
These two types originate in the classification of Kazhdan--Lusztig bases 
for the Hermitian symmetric pair $(B_{N},A_{N-1})$ \cite{S14}.
The main purpose of this paper is to give proofs to generalizations 
of two formulae by Kenyon and Wilson in case of ballot tilings.
As in the case of Dyck tilings, two formulae for ballot tilings 
have factorized form (see Theorems \ref{thrm:BIII-fac}, \ref{thrm:BIII-Hhfac}
and \ref{thrm:BI-fac}).
In \cite{JVK16}, two formulae for a Dyck tiling are interpreted in terms of 
the weak order and the Bruhat order on permutations.
We give an analogue of these results for a ballot tiling (see Theorems 
\ref{thrm:BIII-2}, \ref{thrm:BIII-Hhfac} and \ref{thrm:BI-post}).
Not only factorized expressions, we obtain various expressions of 
the first formula for a ballot tiling (See Section \ref{sec:varigf}).
The key tool is a weakly increasing tree. 
We also derive several recurrence relations for a generating function 
of a ballot tiling in terms of trees and solve them explicitly.

The paper is organized as follows.
In Section \ref{sec:Dyck}, we briefly review cover-inclusive Dyck tilings 
and introduce two Theorems \ref{thrm-A1} and \ref{thrm-A2}.
Theorem \ref{thrm-A1} is a formula for Dyck tilings with a fixed lower path
and Theorem \ref{thrm-A2} is one for Dyck tilings with a fixed upper path.
We provide a simple bijective proof of Theorem \ref{thrm-A2} in 
Section \ref{sec:bijHer}.
The remaining parts of this section is devoted to the proof of 
Theorem \ref{thrm-A1}.
We start from the notion of planted plane trees in Section \ref{sec:tree}.
To compute Kazhdan--Lusztig polynomials, we introduce a labelling of 
a tree which we call type Lascoux--Sch\"utzenberger.
In Section \ref{sec:Hh}, we consider three types of Hermite histories 
which are bijective to a Dyck tiling.
To connect Kazhdan--Lusztig polynomial with an Hermite history bijectively, 
we introduce an integer array associated to an Hermite history.
This integer array plays a central role to give a bijective 
proof of Theorem \ref{thrm-A1}.
This bijection is realized by two maps constructed in Section \ref{sec:maps}.
We show that the first formula for a Dyck tiling is equal to a Kazhdan--Lusztig 
polynomial (Theorem \ref{thrm:DyckKL}).
Since this Kazhdan--Lusztig polynomial is written as a product of quantum integers,
we obtain an expression of the first formula in terms of quantum integers 
(Theorem \ref{thrm:Dyck-fac}), which is a different expression from Theorem \ref{thrm-A1}.

In Section \ref{sec:im}, we start with two types of incidence matrices $M$ and 
introduce two types of ballot tilings. 
The two types are type BI and type BIII. 
As shown in \cite{KW11}, the entries of the inverse matrix $M^{-1}$ correspond to 
a cover-inclusive ballot tiling.

In Section \ref{sec:BallotBIII}, we give bijective proofs of the first and second 
formulae for ballot tilings of type BIII.
Theorem \ref{thrm:BIII-fac} shows that the first formula for ballot tilings of 
type BIII factorizes into the product of the first formula for Dyck tilings 
and an extra factor. 
In Theorem \ref{thrm:BIII-2}, we prove an analogue of the equality between 
the second and third terms in Theorem \ref{thrm-A2}.
Here, we consider $(\underline{1}2,\underline{12})$-avoiding signed 
permutations instead of the Weyl group $\mathfrak{S}_{N}^{C}$ of type C.

Sections \ref{sec:pptgpm} to \ref{sec:varigf} are devoted to analysis 
of ballot tilings of type BI.
In section \ref{sec:ppt}, we start with the notion of planted plane trees
which is a generalization of plane trees in Section \ref{sec:tree}.
A tree introduced here appeared in \cite{Boe88,S14} as a binary tree to compute 
Kazhdan--Lusztig polynomials for the Hermitian symmetric pair $(B_{N},A_{N-1})$.
In Section \ref{sec:gpm}, we introduce a generalized perfect matching 
for ballot tilings to connect them to an Hermite history.
In Section \ref{sec:BTS}, we give a bijection from a labelling of 
a plane tree to a ballot tiling. 
We call this bijection ballot tile strip which is a generalization 
of Dyck tiling strip (DTS for short) studied in \cite{KMPW12}.
In Section \ref{sec:HhBT}, we introduce an Hermite history for a ballot 
tiling and show a bijection between an Hermite history and a generalized 
perfect matching.
paf
a fixed upper path. 
We provide a generalization of generating function of Dyck tilings 
with a fixed upper path considered in Section \ref{sec:bijHer} 
(Theorem \ref{thrm:gD-up}). 
Through a bijection between a generalized perfect matching and 
an Hermite history in Section \ref{sec:HhBT}, this generating 
function can be expressed in terms of generalized perfect matchings 
(Theorems \ref{thrm:gD-up-gpm} and \ref{thrm:bIgPM}).
In case of ballot tilings, we have a factorization of the generating 
function (Theorem \ref{thrm:BIII-Hhfac}).
This generating function is interpreted in terms of the Bruhat order
on down-up permutations.

In Section \ref{sec:BI-low}, we consider ballot tilings of type BI with 
a fixed lower path.
One of the main theorems in this paper is Theorem \ref{thrm:BI-fac}, which 
shows that a generating function is factorized.
Section \ref{sec:BalP} and \ref{sec:BI-fac} are devoted to the proof of 
Theorem \ref{thrm:BI-fac}.
In Section \ref{sec:Fac-trees}, 
we translate Theorem \ref{thrm:BI-fac} in terms of trees introduced in 
Section \ref{sec:ppt}.

In Section \ref{sec:varigf}, we provide several expressions of the generating 
function for ballot tilings with a fixed lower path.
In Section \ref{sec:ext1} and \ref{sec:ext23}, we show recurrence relations 
for the generating functions.
In Section \ref{sec:BI-gtree}, we prove three expressions for the generating 
function: Theorem \ref{thrm:treewt1}, Theorem \ref{thrm:treewt2} and Theorem \ref{thrm:GFLS}.
We consider trees without arrows in Section \ref{sec:twoarrow}.
The generating function can be expressed in a simple way: Theorem \ref{thrm:BI-post} 
and Theorem \ref{thrm:BI-tree-fac}.
These two theorems can be regarded as an analogue of Theorem \ref{thrm-A1}.
Further, we have three more different expressions: Theorems \ref{thrm:rev3}, \ref{thrm:2deg} 
and \ref{thrm:BI-hybrid}.
In Section \ref{sec:twoa-BTS}, we consider a tree $T$ without arrows such that an edge 
of $T$ does not have a parent edge with a dot (see Section \ref{sec:ppt} for detailed 
definitions). 
A generating function for a tree $T$ is expressed as a product of two
factors (Corollary \ref{cor:BTS}), one of which is the sum of weights for 
inverse pre-order words introduced in Section \ref{sec:pptgpm}.

\paragraph{Notation}
We introduce the quantum integer $[n]:=\sum_{i=0}^{n-1}q^{i}$, 
the quantum factorial $[n]!:=\prod_{i=1}^{n}[i]$, 
$[2m]!!:=\prod_{i=1}^{m}[2i]$, 
the $q$-analogue of the binomial coefficients
\begin{eqnarray*}
\genfrac{[}{]}{0pt}{}{n}{m}:=\frac{[n]!}{[n-m]![m]!}, 
\qquad
\genfrac{[}{]}{0pt}{}{n}{m}_{q^{2}}:=\frac{[2n]!!}{[2(n-m)]!!\cdot[2m]!!}, 
\end{eqnarray*}
and the multinomial coefficients
\begin{eqnarray*}
\genfrac{[}{]}{0pt}{}{n}{k_1,\ldots,k_r}:=\frac{[n]!}{\prod_{i=1}^{r}[k_{i}]!},
\qquad 
\genfrac{[}{]}{0pt}{}{n}{k_1,\ldots,k_r}_{q^2}:=\frac{[2n]!!}{\prod_{i=1}^{r}[2k_{i}]!!}.
\end{eqnarray*}
We denote $q,t$-integers by $[n]_{t}:=[n-1]+q^{n-2}t$.

\section{Dyck tiling}
\label{sec:Dyck}

We recall the definitions of Dyck tiling following \cite{JVK16,K12,SZJ12}.

A {\it Dyck path of length $2n$} is a lattice path from 
the origin $(0,0)$ to $(2n,0)$ with up (or ``U") steps $(1,1)$ 
and down (or ``D") steps $(1,-1)$, which does not go below the 
horizontal line $y=0$. 
A sequence of ``U" and ``D" corresponding to a Dyck path is called 
Dyck word.
The highest (resp. lowest) Dyck path of length $2n$ is 
$U\cdots UD\cdots U$ (resp. $UDUD\cdots UD$).

A Dyck path $\lambda$ of length $2n$ is identified with the 
Young diagram which is determined by the path $\lambda$, 
the line $y=x$ and the line $y=-x+2n$. 
Let $\lambda$ and $\mu$ be two Dyck paths of length $2n$.
If the skew shape $\lambda/\mu$ exists, we call $\lambda$ 
and $\mu$ the lower path and the upper path. 

A {\it Dyck tile} is a ribbon (a skew shape which does not 
contain a $2\times 2$ box) such that the centers of the boxes
form a Dyck path. 
A {\it Dyck tiling} is a tiling of a skew Young 
diagram $\lambda/\mu$ by Dyck tiles. 
A Dyck tiling $D$ is called {\it cover-inclusive} if 
we translate a Dyck tile of $D$ downward by $(0,-1)$, 
then it is strictly below $\lambda$ or contained in another 
Dyck tile.
We denote by $\mathcal{D}(\lambda/\mu)$ the set of 
cover-inclusive Dyck tilings of the skew shape $\lambda/\mu$.
We define 
\begin{eqnarray*}
\mathcal{D}(\lambda/\ast):=\bigcup_{\mu}\mathcal{D}(\lambda/\mu), \\
\mathcal{D}(\ast/\mu):=\bigcup_{\lambda}\mathcal{D}(\lambda/\mu).
\end{eqnarray*}

For $D\in\mathcal{D}(\lambda/\mu)$, we introduce two statistics 
called {\it area} and {\it tiles}: 
the area $\mathrm{area}(D)$ is the number of boxes in the skew 
shape $\lambda/\mu$ and the tiles $\mathrm{tiles}(D)$ is 
the number of Dyck tiles of $D$. 
The statistic {\it art} is defined as 
$\mathrm{art}(D)=(\mathrm{area}(D)+\mathrm{tiles}(D))/2$.

\begin{defn}
Let $\lambda$ be a Dyck path. 
We define 
\begin{eqnarray*}
P^{\mathrm{Dyck}}_{\lambda}:=\sum_{D\in\mathcal{D}(\lambda/\ast)}q^{\mathrm{art}(D)}
\end{eqnarray*}
\end{defn}

We introduce a {\it chord} of a Dyck path and its length and height 
(see {\it e.g.} Section 1 in \cite{K12}).
We make a pair of $U$ and $D$ which are next to each other.
Successively, we make a pair of $U$ and $D$ by ignoring paired $U$ and $D$.
A Dyck path of length $2n$ consists of $n$ pairs of $U$ and $D$.
We call this pair of $U$ and $D$ a {\it chord} of the Dyck path $\lambda$.
We denote $\mathrm{Chord}(\lambda)$ by the set of chords of a 
Dyck path $\lambda$.
We denote two statistics for a chord $c$ by the length $l(c)$ and the 
height $\mathrm{ht}(c)$.
For a chord $c\in\mathrm{Chord}(\lambda)$, $l(c)$ is one plus the number of 
chords in-between $U$ and $D$ in $c$.
The height $\mathrm{ht}(c)$ is defined to be $i$ if $c$ is between 
the lines $y=i-1$ and $y=i$.

We denote by $\mathfrak{S}_{n}$ the symmetric group of degree $n$.
The number of inversions of $\pi\in\mathfrak{S}_{n}$ is denoted 
by $\mathrm{inv}(\pi)$.

Let $\pi_0$ be a permutation associated to a Dyck path $\lambda$ 
(see \cite{JVK16} for a detailed definition).
\begin{theorem}[\cite{JVK16,KW11,K12}]
\label{thrm-A1}
We have 
\begin{eqnarray*}
P^{\mathrm{Dyck}}_{\lambda}
&=&\frac{[n]!}{\prod_{c\in\mathrm{Chord}(\lambda)}[l(c)]} \\
&=&\sum_{\pi\ge_{L}\pi_{0}}q^{\mathrm{inv}(\pi)-\mathrm{inv}(\pi_0)},
\end{eqnarray*}
where $\ge_{L}$ is the weak left order on $\mathfrak{S}_n$.
\end{theorem}

Let $\rho(\mu)$ be a $132$-avoiding permutation associated to 
a Dyck path $\mu$ (see \cite{JVK16} for a detailed definition).
\begin{theorem}[\cite{JVK16,KW11,K12}]
\label{thrm-A2}
We have 
\begin{eqnarray*}
\sum_{D\in\mathcal{D}(\ast/\mu)}q^{\mathrm{tiles}(D)}
&=&\prod_{c\in\mathrm{Chord}(\mu)}[\mathrm{ht}(c)] \\
&=&\sum_{\pi\ge\rho(\mu)}q^{\mathrm{inv}(\pi)-\mathrm{inv}(\rho(\mu))}
\end{eqnarray*}
where $\ge$ is the Bruhat order on $\mathfrak{S}_{n}$.
\end{theorem}

In this section, we give a simple bijective proof of Theorem \ref{thrm-A2} 
and a bijective proof of an another expression of Theorem \ref{thrm-A1}.

\subsection{A bijective proof of Theorem \ref{thrm-A2}}
\label{sec:bijHer}

Recall that a chord $c$ of Dyck path $\mu$ is a matching 
pair of an up step $u$ and a down step.
The statistic $\mathrm{ht}(c)$ for a chord of $\mu$
is equal to the number of boxes in the southeast direction 
between $u$ and the lowest path $UDUD\ldots UD$ plus 1.
We denote by $n(u)$ this number of boxes, that is,  
$\mathrm{ht}(c)=n(u)+1$.
Thus, the right hand side of Theorem \ref{thrm-A2} is rewritten as 
$\prod_{u:\mathrm{up\ step\ of\ } \mu}[n(u)+1]$.

We denote by $\mu_{0}$ the lowest path for $\mathcal{D}(\ast/\mu)$, 
which is the zig-zag path (``$UD\cdots UD$").
Fix a path $\mu$ and an up step $u$ of $\mu$.
The number of boxes in $\mu_0/\mu$ in the southeast direction at  
$u$ is $n(u)$. 
The $q$-integer is expanded as $[\mathrm{ht}(c)]=1+q+\ldots+q^{n(u)}$.
When we take a term $q^{i}$ from $[\mathrm{ht}(c)]$, we place 
$i$ successive boxes in the region between $u$ and the zig-zag path such 
that the northwest edge of the northwest box is attached to the up 
step $u$.
Since a term in $\prod_{u}[n(u)+1]$ is a product of $q^{i}$ from $[n(u)+1]$,
we place boxes as above for each up step (see the left figure in Figure \ref{fig-1}).
We denote by $D'$ the obtained diagram consisting of boxes.

We will construct a Dyck tiling from $D'$. 
We enumerate the up steps of $\mu$ from right to left by $u_1, u_2, \ldots u_n$. 
We also enumerate the boxes attached to the up step $u_j$ from northwest to southeast by 
$1,2,\ldots,r_{j}$ where $r_{j}$ is the number of boxes. 
Fix a pair $(u_j,i)$ with $1\le i\le r_{j}$ and call the corresponding box 
$(u_j,i)$ box. 

If there is a box to the northeast of $(u_j,i)$ box and it forms a Dyck tile, 
move to the next pair $(u_j,i+1)$.
Otherwise, there is no box to the northeast of $(u_j,i)$ box.
Let $b$ be a box of $\mu_0/\mu$ right to the $(u_j,i)$ box such that 
the translation of $b$ by $(1,1)$ is either outside of the region $\mu_0/\mu$
or contained in a Dyck tile.
Then, there exists a unique Dyck tile $D$ such that it starts from $(u_j,i)$ box 
and ends at the box $b$ and a box to the north of a box of $D$ is either 
outside of the region $\mu_0/\mu$ or contained in a Dyck tile.
We move the boxes $(u_j,k)$ with $i+1\le k\le r_{j}$ to the southeast of the box $b$.
Then, we move to the next pair $(u_j,i+1)$ for $1\le i\le r_{j}-1$ or $(u_{j+1},1)$ 
for $i=r_{j}$ (see Figure \ref{fig-1}).

It is obvious that the obtained Dyck tile has the weight $q^{\mathrm{tiles}(D)}$. 
The above operation has an inverse, that is, we can construct a diagram $D'$ from 
a Dyck tiling. 
Thus, we obtain the first equality in Theorem \ref{thrm-A2}. 

\begin{figure}[ht]
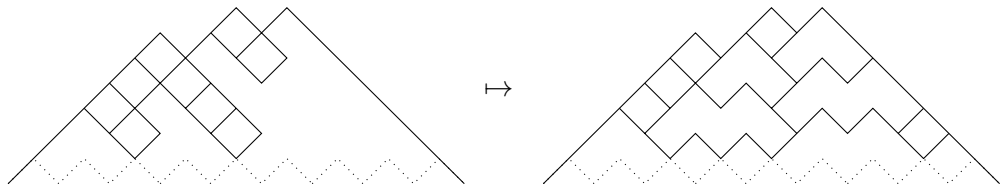

\tikzpic{-0.5}{
\draw (0,0)--(2,2)--(2+1/3,2-1/3)--(3,2+1/3)--(3+1/3,2)--(3+2/3,2+1/3)--(6,0);
\draw[dotted] (0,0)--(1/3,1/3)--(2/3,0)--(1,1/3)--(4/3,0)--(5/3,1/3)--(2,0)
(2,0)--(2+1/3,1/3)--(2+2/3,0)--(2+1,1/3)--(2+4/3,0)--(2+5/3,1/3)--(4,0)
(4,0)--(4+1/3,1/3)--(4+2/3,0)--(5,1/3)--(4+4/3,0)--(4+5/3,1/3)--(6,0);
\draw (2+2/3,2)--(3+1/3,2-2/3)--(3+2/3,2-1/3)--(3+1/3,2)--(3,2-1/3);
\draw (2-1/3,2-1/3)--(3,1/3)--(3+1/3,2/3)--(2+1/3,2-1/3)
      (2,2-2/3)--(2+1/3,2-1/3)(2+1/3,1)--(2+2/3,1+1/3)(2+2/3,2/3)--(3,1);
\draw (4/3,4/3)--(5/3,1)--(2,4/3);
\draw (1,1)--(5/3,1/3)--(2,2/3)--(5/3,1)--(4/3,2/3);
}
$\mapsto$
\tikzpic{-0.5}{
\draw (0,0)--(2,2)--(2+1/3,2-1/3)--(3,2+1/3)--(3+1/3,2)--(3+2/3,2+1/3)--(6,0);
\draw[dotted] (0,0)--(1/3,1/3)--(2/3,0)--(1,1/3)--(4/3,0)--(5/3,1/3)--(2,0)
(2,0)--(2+1/3,1/3)--(2+2/3,0)--(2+1,1/3)--(2+4/3,0)--(2+5/3,1/3)--(4,0)
(4,0)--(4+1/3,1/3)--(4+2/3,0)--(5,1/3)--(4+4/3,0)--(4+5/3,1/3)--(6,0);
\draw (1,1)--(5/3,1/3)--(2,2/3)--(7/3,1/3)--(8/3,2/3)--(3,1/3)--(11/3,1)
--(4,2/3)--(13/3,1)--(15/3,1/3)
(4/3,2/3)--(2,4/3)--(7/3,1)--(8/3,4/3)--(10/3,2/3)
(3,1)--(11/3,5/3)--(4,4/3)--(13/3,5/3)
(5,1/3)--(16/3,2/3)(14/3,2/3)--(5,1)
(4/3,4/3)--(5/3,3/3)(5/3,5/3)--(2,4/3)--(7/3,5/3)
(8/3,6/3)--(10/3,4/3)(9/3,5/3)--(10/3,6/3);
}
\caption{The bijection from a diagram $D'$ to a Dyck tiling. 
The statistics tiles of the Dyck tiling (right figure) is 
$\mathrm{tiles}(D)=9$.}
\label{fig-1}
\end{figure}

\subsection{Planted plane tree}
\label{sec:tree}
Let $\mathcal{Z}$ be a set of words consisting of $U$ and $D$ such that 
$\emptyset\in\mathcal{Z}$, if $z\in\mathcal{Z}$ then $UzD\in\mathcal{Z}$ 
and if $z_1,z_{2}\in\mathcal{Z}$ then the concatenation $z_1z_2\in\mathcal{Z}$.
In other words, the set $\mathcal{Z}$ is a set of Dyck words.

Let $\lambda_1$, $\lambda_2$ be two words of length $2n$ consisting of $U$ and $D$.
If $\lambda_1$ is above $\lambda_2$, we denote it by $\lambda_{1}\ge\lambda_{2}$.  
For a word $\lambda$, we denote by $||\lambda||$ the length of the word and 
by $||\lambda||_{\alpha}$, $\alpha=U$ or $D$, the number of $\alpha$ 
in the word $\lambda$.
Let $\lambda$ be a Dyck word of length $2n$ and $\lambda_0\ge\lambda$ be a word 
of length $2n$ consisting of $U$ and $D$. 
Here, $\lambda_0$ is not necessarily a Dyck word.
Suppose that $\lambda_{0}=\lambda'_{0}vw\lambda''_{0}$ and 
$\lambda=\lambda'UD\lambda''$
where $v,w\in\{U,D\}$ and $||\lambda'_{0}||=||\lambda'||$.
A {\it capacity} of the partial word corresponding to $UD$ in $\lambda$ 
is defined by 
\begin{eqnarray*}
\mathrm{cap}(UD):=||\lambda'_{0}v||_{U}-||\lambda'U||_{U}.
\end{eqnarray*}
The condition $\lambda\le\lambda_{0}$ implies a capacity of $\lambda$ 
is non-negative.

We define a tree $A(\lambda)$ for a Dyck word $\lambda$. 
A tree $A(\lambda)$ satisfies
\begin{enumerate}[($\Diamond$1)]
\item $A(\emptyset)$ is the empty tree.
\item $A(D\lambda')=A(\lambda')$.
\item $A(z\lambda'), z\in\mathcal{D}$, is obtained by attaching the tree 
for $A(z)$ and $A(\lambda')$ at their roots. 
\item $A(UzD)$, $z\in\mathcal{Z}$, is obtained by attaching an edge just 
above the tree $A(z)$.
\end{enumerate}

We put the capacities of $\lambda$ with respect to $\lambda_0$ on 
leaves of the plane tree $A(\lambda)$. 
We denote by $A(\lambda/\lambda_{0})$ a tree $A(\lambda)$ with 
capacities. 
A {\it labelling of Lascoux--Sch\"utzenberger type} (labelling 
of LS type for short) is a set of non-negative integers on the 
edges of $A(\lambda)$ satisfying 
\begin{enumerate}[(LS1)]
\item Integers on edges are non-increasing from leaves to the root.
\item An integer on an edge connecting to a leaf is less than or equal 
to its capacity.
\end{enumerate}
Examples of $A(\lambda/\lambda_{0})$ and its labelling of LS type 
are shown in Figure \ref{fig-2}.

For a labelling $\nu$ of LS type, we denote by $\sigma(\nu)$ the 
sum of integers in $\nu$. 
Suppose that $\lambda$ and $\lambda_{0}$ be a Dyck word.
Define a generating function 
\begin{eqnarray*}
P_{\lambda,\lambda_{0}}:=\sum_{\nu}q^{\sigma(\nu)},
\end{eqnarray*}
where the sum runs over all possible labellings of LS 
type of $A(\lambda/\lambda_{0})$.
Then, Lascoux and Sch\"utzenberger proved that 
\begin{theorem}[Lascoux and Sch\"utzenberger \cite{LS81}]
\label{thrm:KL}
The generating function $P_{\lambda,\lambda_{0}}$ is the Kazhdan--Lusztig 
polynomial for a Grassmannian permutation.
\end{theorem}

Besides a labelling of LS type, we will consider several labellings  
of a tree $A(\lambda)$. 
For this purpose, we introduce an array of integers whose shape is 
determined by a tree $A(\lambda)$.
If two edges $e_1$ and $e_2$ have the same parent edge and $e_2$ is right to 
$e_1$, we put an integer on $e_2$ right to an integer on $e_1$. 
If $e_1$ is a parent of $e_2$, then we put an integer on $e_2$ below  
an integer on $e_1$.
We call an array of integers associated to a labelling of LS type 
an array of integers of LS type for short.
An example of an array of integers of LS type
is shown in Figure \ref{fig-2}.

\begin{figure}[ht]
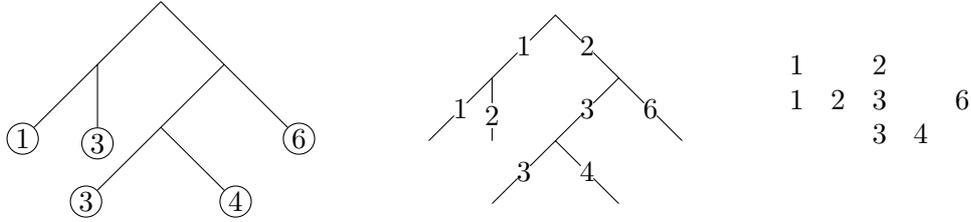

\tikzpic{-0.5}{
\draw (0,0)--(-2/1.2,-2/1.2)node[circle,inner sep=1pt,draw,anchor=north east]{$1$}
(-1/1.2,-1/1.2)--(-1/1.2,-2/1.2)node[circle,inner sep=1pt,draw,anchor=north]{$3$}
(0,0)--(2/1.2,-2/1.2)node[circle,inner sep=1pt,draw,anchor=north west]{$6$}
(1/1.2,-1/1.2)--(-1/1.2,-3/1.2)node[circle,inner sep=1pt,draw,anchor=north east]{$3$}
(0,-2/1.2)--(1/1.2,-3/1.2)node[circle,inner sep=1pt,draw,anchor=north west]{$4$}
;
}
\hspace{24pt}
\tikzpic{-0.5}{
\draw(0,0)--(-0.4/1.2,-0.4/1.2)(-0.5/1.2,-0.5/1.2)node{$1$}
(-0.6/1.2,-0.6/1.2)--(-1.4/1.2,-1.4/1.2)(-1.5/1.2,-1.5/1.2)node{$1$}
(-1.6/1.2,-1.6/1.2)--(-2/1.2,-2/1.2)
(-1/1.2,-1/1.2)--(-1/1.2,-1.4/1.2)(-1/1.2,-1.6/1.2)node{$2$}
(-1/1.2,-1.8/1.2)--(-1/1.2,-2/1.2)
(0,0)--(0.4/1.2,-0.4/1.2)(0.5/1.2,-0.5/1.2)node{$2$}
(0.6/1.2,-0.6/1.2)--(1.4/1.2,-1.4/1.2)(1.5/1.2,-1.5/1.2)node{$6$}
(1.6/1.2,-1.6/1.2)--(2/1.2,-2/1.2)
(1/1.2,-1/1.2)--(0.6/1.2,-1.4/1.2)(0.5/1.2,-1.5/1.2)node{$3$}
(0.4/1.2,-1.6/1.2)--(-0.4/1.2,-2.4/1.2)
(-0.5/1.2,-2.5/1.2)node{$3$}(-0.6/1.2,-2.6/1.2)--(-1/1.2,-3/1.2)
(0/1.2,-2/1.2)--(0.4/1.2,-2.4/1.2)(0.5/1.2,-2.5/1.2)node{$4$}
(0.6/1.2,-2.6/1.2)--(1/1.2,-3/1.2);
}
\hspace{24pt}
$
\begin{array}{ccccc}
1 &   & 2 &   & \\
1 & 2 & 3 &   & 6 \\
  &   & 3 & 4 &  \\ 
\end{array}
$
\caption{
The left picture is an example of a plane tree $A(\lambda/\lambda_0)$ with capacities $(1,3,3,4,6)$.
The middle picture shows an example of a labelling of LS type associated to the plane tree. 
The right picture is an array of integers of LS type associated to the middle picture.
}
\label{fig-2}
\end{figure}

\subsection{Hermite history}
\label{sec:Hh}
We introduce an Hermite history for a Dyck tiling following \cite[Section 6]{K12} 
and generalize it.
Fix a Dyck tiling $D$ in $\mathcal{D}(\lambda/\mu)$.
For an up step $u$ of a Dyck word $\lambda$, we denote by $n(u)$
the number of down steps left to the up step $u$ and by $h$ the 
height of the up step $u$. 
For example, the height of the third up step from left in a path $UUDUDD$ 
is two.
A {\it Hermite history of length $2n$} of type I (resp. type II) is 
a pair $(\mu,H)$ (resp. $(H,\lambda)$) of a Dyck path 
$\mu$ (resp. $\lambda$) and a labelling $H$ (resp. $H'$) 
of $\mu$ (resp. $\lambda$). 
A labelling $H$ (resp. $H'$) is a set of integers on the up steps 
of $\mu$ (resp. $\lambda$) such that a label of $u$ is an integer 
in $[0,h-1]$ (resp. $[0,n(u)]$).
An Hermite history of type III is a pair $(\mu,H'')$ of a Dyck path 
$\mu$ and a labelling $H''$ of $\mu$.
A labelling $H''$ is a set of integers on the down steps 
of $\mu$ such that if the height of an down step $d$ is $h'$ then
a label of $d$ is an integer in $[0,h'-1]$.

A bijection from a Dyck tiling $D$ to $(\mu,H)$ of type I is defined as follows.
For a Dyck tile $D$, we define the {\it entry} (resp. {\it exit}) 
northwest (resp. southeast) edge of the leftmost (resp. rightmost)
cell of $D$. 
We consider a path starting from an up step $u$ of $\mu$. 
If $u$ is not a entry of a Dyck tile, then the label of $u$ is zero. 
If $u$ is an entry of a Dyck tile $D$ , then extend a path from the entry to
the exit of $D$ until the path arrives at an edge which is not an entry.
The labelling of $u$ is a number of tiles which the path starting from $u$ 
travels.
We denote by $|H|$ the sum of labels on up steps of $\mu$. 
Then, we have $|H|=\mathrm{tiles}(H)$. 

A bijection from a Dyck tiling to $(\mu,H'')$ of type III is defined by 
replacing northwest with northeast, southeast with southwest in the 
definition of the entry and the exit, and replacing an up step with 
a down step in the definitions for type I.

A bijection from a Dyck tiling $D$ to $(H',\lambda)$ of type II is defined as follows.
We consider the same paths on $D$ as in the case of $(\mu,H)$.
A label of an up step $u$ of $\lambda$ is zero if the translation of $u$ by $(0,2k)$ with 
$k\ge 0$ is not an exit of a Dyck tile. If the translation of $u$ by $(0,2k)$ is an exit 
and its path passes through $r$ Dyck tiles of length $2n_1, 2n_2,\ldots, 2n_r$, 
the label of $u$ is $\sum_{1\le i\le r}(n_i+1)$.
We denote by $|H'|$ the sum of labels on up steps of $\lambda$.
Then, we have $|H'|=\mathrm{art}(D)$.
The left figure in Figure \ref{fig-3} is an example of Hermite histories of type I
and type II.

We associate an array of integers to the pair $(H',\lambda)$. 
The shape of the array is determined by the tree $A(\lambda)$
as in Section \ref{sec:tree}.
We put a label of $H'$ on the array in the up-right order: 
starting from the bottom row of the leftmost column to the first row, 
move to the bottom row of the second-leftmost column 
and ending at the first row of the rightmost column.
We call this array of integers the one of type Hermite 
history (type Hh for short). 
The right figure in Figure \ref{fig-3} is an example of integer array 
of type Hh.

\begin{figure}[ht]
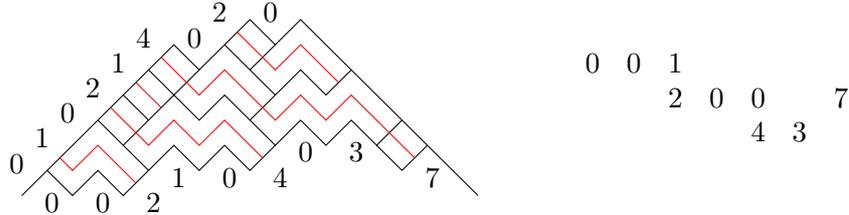

\tikzpic{-0.5}{
\draw (0,0)--(2,2)--(2+1/3,2-1/3)--(3,2+1/3)--(3+1/3,2)--(3+2/3,2+1/3)--(6,0);
\draw (1/3,1/3)--(2/3,0)--(1,1/3)--(4/3,0)--(5/3,1/3);
\draw (1,1)--(5/3,1/3)--(2,2/3)--(7/3,1/3)--(8/3,2/3)--(3,1/3)--(11/3,1)
--(4,2/3)--(13/3,1)--(15/3,1/3)
(4/3,2/3)--(2,4/3)--(7/3,1)--(8/3,4/3)--(10/3,2/3)
(3,1)--(11/3,5/3)--(4,4/3)--(13/3,5/3)
(5,1/3)--(16/3,2/3)(14/3,2/3)--(5,1)
(4/3,4/3)--(5/3,3/3)(5/3,5/3)--(2,4/3)--(7/3,5/3)
(8/3,6/3)--(10/3,4/3)(9/3,5/3)--(10/3,6/3);
\draw[red] (1.5/3,1.5/3)--(2/3,1/3)--(3/3,2/3)--(4.5/3,0.5/3)
(3.5/3,3.5/3)--(5/3,2/3)--(6/3,3/3)--(7/3,2/3)--(8/3,3/3)--(9.5/3,1.5/3)
(4.5/3,4.5/3)--(5.5/3,3.5/3)
(5.5/3,5.5/3)--(7/3,4/3)--(8/3,5/3)--(10/3,3/3)--(11/3,4/3)--(12/3,3/3)
--(13/3,4/3)--(15.5/3,1.5/3)
(8.5/3,6.5/3)--(10/3,5/3)--(11/3,6/3)--(12.5/3,4.5/3);
\draw(0.5/3,0.5/3)node[anchor=south east]{$0$}
(1.5/3,1.5/3)node[anchor=south east]{$1$}
(2.5/3,2.5/3)node[anchor=south east]{$0$}
(3.5/3,3.5/3)node[anchor=south east]{$2$}
(4.5/3,4.5/3)node[anchor=south east]{$1$}
(5.5/3,5.5/3)node[anchor=south east]{$4$}
(7.5/3,5.5/3)node[anchor=south east]{$0$}
(8.5/3,6.5/3)node[anchor=south east]{$2$}
(10.5/3,6.5/3)node[anchor=south east]{$0$};
\draw (0.5/3,0.5/3)node[anchor=north west]{$0$}
(2.5/3,0.5/3)node[anchor=north west]{$0$}
(4.5/3,0.5/3)node[anchor=north west]{$2$}
(5.5/3,1.5/3)node[anchor=north west]{$1$}
(7.5/3,1.5/3)node[anchor=north west]{$0$}
(9.5/3,1.5/3)node[anchor=north west]{$4$}
(10.5/3,2.5/3)node[anchor=north west]{$0$}
(12.5/3,2.5/3)node[anchor=north west]{$3$}
(15.5/3,1.5/3)node[anchor=north west]{$7$};
}
\hspace{24pt}
$
\begin{array}{ccccccc}
0 & 0 & 1 &   &   &   &  \\
  &   & 2 & 0 & 0 &   & 7 \\
  &   &   &   & 4 & 3 &    
\end{array}
$
\caption{The left picture is an example of Hermite histories 
of type I and II. 
A red path is a non-intersecting path connecting an entry 
and an exit of a Dyck tile. 
The upper path is $\mu$ and the lower path is $\lambda$. 
The labelling $H$ of $\mu$ is $(0,1,0,2,1,4,0,2,0)$ and 
the labelling $H'$ of $\lambda$ is $(0,0,2,1,0,4,0,3,7)$.
The right is an array of integers of type Hh associated to $(H',\lambda)$.
}
\label{fig-3}
\end{figure}

Given a Dyck path $\lambda$, recall that the shape of an array of integers 
$M$ is determined by $\lambda$. 
We denote by $\lambda_0$ a path which consists of only "U". 
Then, a capacity of $A(\lambda/\lambda_0)$ is equal to $n(u)$ where
$u$ is an up step which corresponds to the bottom row of $M$. 
We denote by $\mathcal{M}(\lambda)$ a set of arrays of integers 
such that the shape is determined by $\lambda$, integers in a column 
are weakly increasing from top to bottom and an integer of the bottom 
row in the column is equal to or less than the corresponding capacity 
of $A(\lambda/\lambda_0)$.
We denote by $\mathcal{M}_{\mathrm{LS}}(\lambda)$ the subset of 
$\mathcal{M}(\lambda)$ satisfying the conditions (LS1) and (LS2).

Let $M$ be the following partial array of integers:
\begin{eqnarray}
\label{ArrayM}
\begin{array}{ccc}
x_1 &  & \\
x_2 &  & y_1\\
\vdots & M' & \vdots \\
x_n & & y_m
\end{array},
\end{eqnarray}
where the $i$-th column is $(x_1,\ldots,x_n)^{T}$, 
the $j$-th ($i<j$) column is $(y_1,\ldots,y_m)^{T}$ 
and $M'$ is a partial array such that the non-empty entry 
of $M'$ is at the same height to or lower than $y_1$.  
We have no constraint on the height of $x_1$. 
Let $c_i$ be $i$-th the capacity of $\lambda$ with respect to $\lambda_0$.  
Let $p$ be the number of $x_k$ which is at the same height to or lower than
$y_1$ and $q$ be the number of non-empty entries in $M'$.

Given $x_k$ in the $i$-th column, we define a map
 $\varphi_{i}:\mathbb{N}_{\ge0}\rightarrow\mathbb{N}_{\ge0}$ by 
\begin{eqnarray*}
\varphi_{i}(x_k):=c_i-x_{k},
\end{eqnarray*}
where $c_i$ is the capacity for the $i$-th column.
We consider the following three conditions on $M$:
\begin{enumerate}[(Hh1)]
\item There exists an entry $z$ of $M$ such that 
$z\in M'\cup\{x_1,\ldots,x_{p}\}$ and $\varphi_{l}(z)\le\varphi_{j}(y_m)$ with 
$i\le l\le j-1$.
\item  $\varphi_{l}(z)>\varphi_{j}(y_m)$ with $i<l<j$ and 
$c_{i-1}\le \varphi_{j}(y_m)<c_{i}$, when $n=p$ and $z\in M'$.
\item $x_{p+1}+p+q\ge y_m$, when $n>p$, $\varphi_i(x_{p+1})\le\varphi_{j}(y_m)$ and 
$\varphi_{i}(x_k),\varphi_{l}(z)>\varphi_{j}(y_m)$ where 
$1\le k\le p$ and $i<l<j$.
\end{enumerate}
Let $\mathcal{M}_{\mathrm{Hh}}(\lambda)$ be the subset of $\mathcal{M}(\lambda)$ satisfying 
either (Hh1), (Hh2) or (Hh3).

\begin{prop}
\label{prop:Hh}
There exists a bijection between the pair $(H',\lambda)$ 
and $\mathcal{M}_{\mathrm{Hh}}(\lambda)$.
\end{prop}
\begin{proof}
Fix a lower Dyck path $\lambda$, a Dyck tiling $\mathcal{D}$.
In an integer array $M$ associated to $\lambda$, integers in a column 
are weakly increasing from top to bottom since we have the following three
reasons: 1) the upper path of $\mathcal{D}$ is lower than $U\ldots UD\ldots D$,
2) a Dyck tile of length $2n$ contributes $n+1$ in the label $H'$ and 
3) a tiling $\mathcal{D}$ is cover-inclusive.

We enumerate all steps in $\lambda$ from left to right by $1,2\ldots$.
Suppose that the rightmost box of a Dyck tile of length $2n$ is
on the $i$-th up step in $\lambda$.
Then, the leftmost box $b$ of the Dyck tile is on the $(i-2n+2)$-th 
up step.
To have a Dyck tiling, the existence of a box $b'$ which is at south west of 
the box $b$ is required.
It is obvious that the conditions (Hh1), (Hh2) and (Hh3) are equivalent to the conditions for 
the existence of $b'$.

\end{proof}

\subsection{\texorpdfstring{Maps $\omega$ and $\sigma$}{Maps omega and sigma}}
\label{sec:maps}
Fix $\lambda_0$ be a path of length $2n$ consisting of only $U$. 
We will give a bijection between an array of integers of type LS  
associated to a plane tree $A(\lambda/\lambda_0)$ 
and an array of integers of type Hh.

Let $M$ be a partial array of integers in $\mathcal{M}(\lambda)$.  
The array $M$ of integers looks partially an array of integers shown 
in Eqn.(\ref{ArrayM}).
We consider the case where $x_1$ is strictly above $y_1$ and 
the entries of $M'$ are the same height to or lower than $y_1$.
Suppose that $x_k$ is an entry which corresponds to a parent of $y_1$ 
in the tree $A(\lambda)$, that is, $x_k$ is the lowest entry in
$i$-th column which is strictly above $y_1$.
An array of integers is said to have a {\it discrepancy of type LS} at $j$-th column 
when there exists the pair $(x_k,y_1)$ with $x_k>y_1$.
An array $M$ of integers is said to have a {\it discrepancy of type Hh} at $j$-th column
when $M$ violates the conditions from (Hh1) to (Hh3).
From the definitions of $\mathcal{M}(\lambda)$ and 
$\mathcal{M}_{\mathrm{LS}}(\lambda)$, if $M$ does not have a discrepancy of type LS, 
then $M$ satisfies $x_{k}<y_{1}$.
Similarly, a discrepancy of type Hh means the condition such that $x_{p+1}+p+q<y_{m}$.
Here, $p$ is the number of $x_l$ which is at the same height to or lower than
$y_1$ and $q$ is the number of non-empty entries in $M'$.
Let $a$ be the number such that $x_{a-1}\le y_1$ and $x_{a}> y_1$ and 
$b$ be the number such that $x_b<y_m-p-q$ and $x_{b+1}\ge y_{m}-p-q$.
We define two maps $\omega,\sigma:\mathcal{M}(\lambda)\rightarrow\mathcal{M}(\lambda)$ 
as follows.
The image $\omega(M)$ is obtained by replacing $(x_1,\ldots,x_{n})^{T}$
with 
$(x_1,\ldots,x_{a-1},y_1,x_{a}-1,\ldots,x_{a+p-1}-1,x_{a+p+1},\ldots,x_{n})^{T}$,
the entry $z$ with $z-1$ when $z\in M'$, and $(y_1,\ldots,y_m)^{T}$
with $(y_2,\ldots,y_m,x_{a+p}+p+q)^{T}$.
In the matrix notation, $\omega(M)$ is described as  
\begin{eqnarray}
\omega(M):=
\begin{array}{ccc}
x_1 & & \\
\vdots& & \\ 
x_{a-1}& & y_2 \\ 
y_1  & & \vdots\\
x_{a}-1 & M'-1 & y_{m} \\
\vdots & & x_{a+p}+p+q \\
x_{a+p-1}-1 & & \\
x_{a+p+1} & & \\
\vdots & & \\
x_{n} & & \\
\end{array}.
\end{eqnarray}
The image $\sigma(M)$ is obtained from $M$ by replacing $(x_{1},\ldots,x_n)^{T}$ 
with $(x_1,\ldots,x_{b-p-1},x_{b-p+1}+1,\ldots,x_{b}+1,y_{m}-p-q,x_{b+1},\ldots,x_n)^{T}$, 
the entry $z$ with $z+1$ when $z\in\mathcal{M'}$, and 
$(y_1,\ldots,y_m)$ with $(x_{b-p},y_1,y_2,\ldots ,y_{m-1})$. 
In the matrix notation, $\sigma(M)$ is described as 
\begin{eqnarray}
\sigma(M):=
\begin{array}{ccc}
x_1 & & \\
\vdots& & \\ 
x_{b-p-1}& & x_{b-p} \\ 
x_{b-p+1}+1  & & y_1 \\
\vdots & M'+1 & \vdots \\
x_{b}+1 & & y_{m-1} \\
y_{m}-p-q & & \\
x_{b+1} & & \\
\vdots & & \\
x_{n} & & \\
\end{array}.
\end{eqnarray}

Given an integer array $N$, we denote by $\mathrm{Dis}^{\mathrm{LS}}(N)$ 
(resp. $\mathrm{Dis}^{\mathrm{Hh}}(N)$) be the set of labels of columns in $N$ which have a 
discrepancy of type LS (resp. Hh).
Let $(y_{j,1},\ldots,y_{j,n_{j}})^{T}$ be the $j$-th column in $N$. 
Let $j_1$ be the biggest integer such that 
$j_1\in\mathrm{Dis}^{\mathrm{LS}}(N)$ and 
$\varphi_{j_1}(y_1)=\max\{\varphi_{i}(y_{i,1}) |\  i\in\mathrm{Dis}^{\mathrm{LS}}(N)\}$.
Similarly, let $j_2$ be the smallest integer such that 
$j_2\in\mathrm{Dis}^{\mathrm{Hh}}(N)$ and 
$\varphi_{j_2}(y_1)=\min\{\varphi_{i}(y_{i,n_{i}}) |\  i\in\mathrm{Dis}^{\mathrm{Hh}}(N)\}$.
We denote by $M$ the partial array of $N$ involving the $j_1$-th (resp. $j_2$-th) column.
The actions of $\omega$ (resp. $\sigma$) on $N$ is defined 
as $\omega(N):=(N\backslash M)\cup\omega(M)$ (resp. $\sigma(N):=(N\backslash M)\cup\sigma(M)$), {\it i.e.}, 
we change a partial matrix $M$ to $\omega(M)$ or $\sigma(M)$ and other 
entries are not changed.

\begin{lemma}
\label{lemma:omega1}
Suppose $M\in\mathcal{M}_{\mathrm{Hh}}(\lambda)\cap
(\mathcal{M}(\lambda)\setminus\mathcal{M}_{\mathrm{LS}}(\lambda))$. 
Then, $\omega(M)\not\in\mathcal{M}_{\mathrm{Hh}}(\lambda)$. 
\end{lemma}
\begin{proof}
First we show that the non-empty entries of the $i$-th and $j$-th columns
of $\omega(M)$ are weakly increasing.
Since $x_a>y_1$ implies $x_{a}-1\ge y_1$, the $i$-th column 
of $\omega(M)$ is weakly increasing. 
Since the array $M$ is in $\mathcal{M}_{\mathrm{Hh}}(\lambda)$, 
we have a constraint: $y_{n}\le x_{p+1}+p+q$. 
The column $(x_1,\ldots,x_{m})^{T}$ is weakly increasing, 
which implies $y_n\le x_{p+1}+p+q\le x_{a+p}+p+q$. 
Therefore, the $j$-th column is weakly increasing.

Since we have a discrepancy of type LS at the $j$-th column, 
we do not have a discrepancy of type LS at the $l$-th column
with $i<l<j$. 
Thus, the integer $z\in M'$ satisfies $z\ge1$ and $M'-1$ is well-defined 
under the action of $\omega$. 
A column of the partial array $M'-1$ is weakly increasing since $M'$ is so.

To prove Lemma, it is enough to show that $\omega(M)$ is not in 
$\mathcal{M}_{\mathrm{Hh}}(\lambda)$.
The $(p+1)$-th entry of the $i$-th column of $\omega(M)$ is 
$x_{p}-1$.
The bottom entry of the $j$-th column of $\omega(M)$ is 
$x_{a+p}+p+q$. 
Since $x_{p}\le x_{a+p}$, we have 
$x_{p}+p+q-1<x_{a+p}+p+q$, which is a contradiction 
against (Hh1), (Hh2) and (Hh3). 
Therefore, $\omega(M)\not\in\mathcal{M}_{\mathrm{Hh}}(\lambda)$.
\end{proof}

Let $c$ be an integer such that the $(i,c)$-th entry of $M$ is
a parent of the $(j,1)$-th entry in the corresponding tree.

\begin{lemma}
\label{lemma:omega2}
Suppose $M\in\mathcal{M}_{\mathrm{Hh}}(\lambda)\cap
(\mathcal{M}(\lambda)\setminus\mathcal{M}_{\mathrm{LS}}(\lambda))$.
Then, we have $\omega^{r}(M)\not\in\mathcal{M}_{Hh}(\lambda)$ 
with a positive integer $r\le\min(x_c,c)$. 
\end{lemma}
\begin{proof}
A proof is similar to the one of Lemma \ref{lemma:omega1}. 
The bottom entry of the $j$-th column of $\omega(M)$ 
is $x_{a+p+s(r)}+p+q$ with $s(r)\ge r-1$.
The integers $s(r')$ with $0\le r'\le r$ is a weakly increasing. 
Thus the entries of the $j$-th column is weakly increasing.
 
The $(p+1)$-th entry of the $i$-th column of $\omega(M)$
is $x_{p-s'}-1$ with $0\le s'\le r-1$.  
We have $x_{p-s'}+p+q-1<x_{p}+p+q\le x_{a+p+s}+p+q$, which 
implies the violation of the conditions (Hh1) to (Hh3).
Therefore, we have 
$\omega^{r}(M)\in\mathcal{M}(\lambda)\backslash\mathcal{M}_{\mathrm{Hh}}(\lambda)$.	

Since the $k$-th column with $i<k<j$ is not a discrepancy of type LS, a non-empty element 
in $M'$ is equal to or bigger than $x_c$. 
We have $x_c-r\ge0$.
By a similar argument to Lemma \ref{lemma:omega1}, we have $M'-r\ge0$. 
Thus the image of $\omega^{r}$ is in $\mathcal{M}(\lambda)$ but not in 
$\mathcal{M}_{Hh}(\lambda)$. This completes the proof.
\end{proof}

\begin{lemma}
\label{lemma:omega3}
Suppose that $M\in\mathcal{M}_{\mathrm{Hh}}(\lambda)\cap
(\mathcal{M}(\lambda)\setminus\mathcal{M}_{\mathrm{LS}}(\lambda))$.
Then, there exists a positive integer $r\le\min(x_c,c)$ such that 
$\omega^{r}(M)\in\mathcal{M}_{\mathrm{LS}}(\lambda)$.
\end{lemma}
\begin{proof}
From Lemma \ref{lemma:omega2}, $\omega^{r}(M)$ is in $\mathcal{M}(\lambda)$ with 
a positive integer $r\le\min(x_c,c)$.
This implies that the first entry of the $j$-th column of $\omega^{r}(M)$ is 
weakly increasing. 
When $r=\min(x_c,c)$, the first entry of the $j$-th column is $y_{r+1}$ or 
$x_{a+p+s(r)}+p+q$ with some integer $s(r)$. 
The $c$-th entry of the $i$-th column is equal to or smaller than $y':=\min(y_r,x_{c-r}-r)$.
Further, the value $y'$ is equal to or smaller than $y_{r}$ or $x_{a+p+s(r)}+p+q$. 
The image $\omega^{r}(M)$ satisfies the condition (LS1) at the $j$-th column.
The $k$-th column with $i<k<j$ also satisfies the condition (LS1).
It is obvious that $\omega^{r}(M)$ satisfies the condition (LS2).
Therefore, $\omega^{r}(M)$ is in $\mathcal{M}_{\mathrm{LS}}(\lambda)$.
\end{proof}

\begin{lemma}
Suppose $M\in\mathcal{M}_{\mathrm{LS}}(\lambda)\cap(\mathcal{M}(\lambda)\setminus\mathcal{M}_{\mathrm{Hh}}(\lambda))$.
Then, $\sigma^{r}(M)\in(\mathcal{M}(\lambda)\setminus\mathcal{M}_{\mathrm{LS}}(\lambda))$ with 
a positive integer $r$.
\end{lemma}
\begin{proof}
Since $M\in\mathcal{M}_{\mathrm{LS}}\cap(\mathcal{M}\setminus\mathcal{M}_{Hh})$, 
we have $x_{b-p}\le y_1$.
From the definition of $\sigma$, we have $x_{c+1}+1>x_{c+1}>x_{c}>x_{b-p}$, 
which implies that $\sigma(M)\notin\mathcal{M}_{\mathrm{LS}}$.
By a similar argument, the image $\sigma^{r}(M)$ has a discrepancy of type LS
at the $j$-th column. This completes the proof.
\end{proof}

\begin{lemma}
Suppose $M\in\mathcal{M}_{\mathrm{LS}}(\lambda)\cap(\mathcal{M}(\lambda)\setminus\mathcal{M}_{\mathrm{Hh}}(\lambda))$.
Then, there exists a positive integer $r$ such that  
$\sigma^{r}(M)\in\mathcal{M}_{\mathrm{Hh}}(\lambda)$.
\end{lemma}
\begin{proof}
The integer array $M$ has a discrepancy of type Hh at the $j$-th column.
By the definition of $\sigma$, an application of $\sigma$ on a partial integer 
array $M''$ from the $i$-th column to the $(j-1)$-th column satisfies one of 
conditions from (Hh1) to (Hh3).
The map $\sigma$ maps $M'$ to $M'+1$ and some entries in the $i$-th column 
are increased. 
Thus, the application of $\sigma^{r}$ on
$M''$ also satisfies the conditions for $\mathcal{M}_{\mathrm{Hh}}$.
The bottom entry $e$ of the $j$-th column is weakly decreasing if 
we apply $\sigma$ on $M$.
Let $y(n,i)$ be the integer on $e$ after we apply $\sigma$ on $M$
$i$ times.
If $y(n,i)\le p+q$, then $\sigma^{i}(M)$ is in $\mathcal{M}_{Hh}$.
If $y(n,i)>p+q$, we have $y(n,i)\ge y(n,i+1)\ldots$. In this case,
there exits an integer $r$ such that $\sigma^{r}(M)$ satisfies 
the condition (Hh3).
This completes the proof.
\end{proof}

For $M\in\mathcal{M}(\lambda)$, we denote by $|M|$ the sum of integers 
in $M$. 
From the definitions of $\omega$ and $\sigma$, it is obvious that 
\begin{lemma}
\label{lemma:degomega}
We have $|M|=|\omega(M)|=|\sigma(M)|$.
\end{lemma}

Suppose that $M\in\mathcal{M}_{\mathrm{Hh}}(\lambda)$ and 
$M\not\in\mathcal{M}_{\mathrm{LS}}(\lambda)$. 
From Lemma \ref{lemma:omega3}, there exits a unique smallest positive
integer $r_0$ such that $\omega^{r_0}(M)\in\mathcal{M}_{\mathrm{LS}}(\lambda)$.

\begin{theorem}
\label{thrm:omega-sigma}
A map $\omega^{r_0}$ is a bijection between 
$\mathcal{M}_{\mathrm{Hh}}\cap(\mathcal{M}\setminus\mathcal{M}_{\mathrm{LS}})$
and $\mathcal{M}_{\mathrm{LS}}\cap(\mathcal{M}\setminus\mathcal{M}_{\mathrm{Hh}})$.
The inverse is $\sigma^{r_0}$. 
\end{theorem}
\begin{proof}
Suppose that 
$M\in\mathcal{M}_{Hh}\cap(\mathcal{M}(\lambda)\setminus\mathcal{M}_{\mathrm{LS}}(\lambda))$. 
It is enough to show that $\sigma\circ\omega^{r}(M)=\omega^{r-1}(M)$ with $1\le r\le r_{0}$.
This follows from the explicit definitions of $\omega$ and $\sigma$.
\end{proof}

\begin{example}
Let $\lambda=UUDDUUUDDUDUUDDD$ be a Dyck path.
The shape of an integer array is 
\begin{eqnarray*}
\begin{array}{ccc}
\ast & & \\
\ast & \ast & \ast \\
\ast & & \ast \\ \hline
2 & 4 & 5 
\end{array}
\end{eqnarray*}
where the integers $2,4$ and $5$ below a horizontal line are the capacities with respect 
to $\lambda_0$ and 
we omit the first column whose capacity and entries are zero.
We have $\#\{\nu \ |\  \nu\in\mathcal{M}_{\mathrm{LS}}(\lambda) \text{\ and\ } |\nu|=5\}=38$. 
We have 6 labellings out of 38 labellings which have a discrepancy of type Hh and they are 
\begin{eqnarray*}
\begin{array}{ccc}
0 & & \\
0 & 0 & 0 \\
0 & & 5
\end{array},\quad
\begin{array}{ccc}
0 & & \\
0 & 0 & 1 \\
0 & & 4
\end{array},\quad
\begin{array}{ccc}
0 & & \\
0 & 1 & 0 \\
0 & & 4
\end{array},\quad
\begin{array}{ccc}
0 & & \\
0 & 3 & 0 \\
0 & & 2
\end{array},\quad
\begin{array}{ccc}
0 & & \\
0 & 3 & 1 \\
0 & & 1
\end{array},\quad
\begin{array}{ccc}
0 & & \\
0 & 4 & 0 \\
0 & & 1
\end{array}.
\end{eqnarray*}
The actions of $\sigma$ on these 6 labellings are 
\begin{eqnarray*}
\begin{array}{ccc}
1 & & \\
1 & 1 & 0 \\
2 & & 0
\end{array},\quad
\begin{array}{ccc}
1 & & \\
1 & 1 & 0 \\
1 & & 1
\end{array},\quad
\begin{array}{ccc}
1 & & \\
1 & 2 & 0 \\
1 & & 0
\end{array},\quad
\begin{array}{ccc}
1 & & \\
1 & 0 & 0 \\
1 & & 2
\end{array},\quad
\begin{array}{ccc}
1 & & \\
1 & 0 & 1 \\
1 & & 1
\end{array},\quad
\begin{array}{ccc}
1 & & \\
1 & 0 & 0 \\
2 & & 1
\end{array},
\end{eqnarray*}
and they are in $\mathcal{M}_{\mathrm{Hh}}(\lambda)$ but have a 
discrepancy of type LS.
\end{example}

\subsection{A bijective proof of another expression of Theorem \ref{thrm-A1}}

Let $\lambda_{0}$ be a path consisting of only $U$'s. 
\begin{theorem}
\label{thrm:DyckKL}
For a Dyck path $\lambda$, we have 
\begin{eqnarray}
\label{art1}
\sum_{D\in\mathcal{D}(\lambda/\ast)}q^{\mathrm{art}(D)}
=P_{\lambda,\lambda_{0}}
\end{eqnarray}
\end{theorem}
\begin{proof}
From the bijection between a Dyck tiling and $(H',\lambda)$, 
the left hand side of Eqn.(\ref{art1}) is equal to 
\begin{eqnarray*}
\sum_{H'}q^{|H'|},
\end{eqnarray*}
where sum runs over all possible $H'$ of the pair $(H',\lambda)$.
From Proposition \ref{prop:Hh}, the sum is expressed as 
\begin{eqnarray*}
\sum_{M\in\mathcal{M}_{\mathrm{Hh}}(\lambda)}q^{|M|}. 
\end{eqnarray*}
The set $\mathcal{M}_{\mathrm{Hh}}(\lambda)$ is equal to 
$(\mathcal{M}_{\mathrm{LS}}(\lambda)\cap\mathcal{M}_{\mathrm{Hh}}(\lambda))\cup
(\mathcal{M}_{\mathrm{Hh}}(\lambda)\cap(\mathcal{M}(\lambda)
\setminus\mathcal{M}_{\mathrm{LS}}(\lambda)))$.
We apply Theorem \ref{thrm:omega-sigma} and Lemma \ref{lemma:degomega}
to $\mathcal{M}_{\mathrm{Hh}}(\lambda)\cap(\mathcal{M}(\lambda)
\setminus\mathcal{M}_{\mathrm{LS}}(\lambda))$. 
Then the sum is equivalent to 
\begin{eqnarray*}
\sum_{M\in\mathcal{M}_{\mathrm{LS}}(\lambda)}q^{|M|},
\end{eqnarray*}
which implies the sum is equal to the Kazhdan--Lusztig polynomial
$P_{\lambda,\lambda_{0}}$ from Theorem \ref{thrm:KL}.
\end{proof}

Let $A(\lambda/\lambda_0)$ be a plane tree associated to a path $\lambda$.
We depict the generating function $P_{\lambda,\lambda_{0}}$ 
as this tree $A(\lambda/\lambda_0)$ with the capacities.
Recall that the tree satisfies the conditions (LS1) and (LS2).
\begin{lemma}
\label{lemma:tritree}
A trivalent tree satisfies the relation:
\begin{eqnarray}
\label{tri-tree}
\tikzpic{-0.5}{
\draw (0,0)--(0,0.5) (0,1)--(0,2)
      (-0.1,0.4)--(0.1,0.4)(-0.1,1.6)--(0.1,1.6)(-0.1,1.2)--(0.1,1.2)
;
\draw[dotted] (0,0.5)--(0,1.5);
\draw[decoration={brace,mirror,raise=5pt},decorate]
  (0,0) -- node[right=6pt] {$M-m_{1}-n_{1}$} (0,2);
\draw[decoration={brace,mirror,raise=5pt},decorate]
  (0,0) -- node[anchor=south east]{$n_{1}$} (-1.2,-1.2);

\draw(0,0)node[anchor=north]{A}--(-0.6,-0.6)
(-0.8,-0.8)--(-1.2,-1.2)node[circle,inner sep=1pt,draw,anchor=north east]{$c_1$}
(-0.5,-0.3)--(-0.3,-0.5)(-0.95,-0.79)--(-0.79,-0.95);
\draw[dotted](-0.6,-0.6)--(-0.8,-0.8);

\draw(0,0)--(0.5,-0.5)
(0.8,-0.8)--(1.2,-1.2)node[circle,inner sep=1pt,draw,anchor=north west]{$c_2$}
(0.3,-0.5)--(0.5,-0.3)(0.7+0.05,-0.9-0.05)--(0.9+0.05,-0.7-0.05);
\draw[dotted](0.5,-0.5)--(1.2,-1.2);
\draw[decoration={brace,mirror,raise=5pt},decorate]
  (1.2,-1.2) -- node[right=6pt] {$m_{1}$} (0,0);
}
=\genfrac{[}{]}{0pt}{}{n_1+m_1}{n_1}\cdot
\tikzpic{-0.5}{
\draw(0,0)--(0,-0.5)(0,-1.5)--(0,-2)node[circle,inner sep=1pt,draw,anchor=north]{$c_1$}
(-0.1,-0.4)--(0.1,-0.4)(-0.1,-1.6)--(0.1,-1.6);
\draw[dotted](0,-0.5)--(0,-2);
\draw[decoration={brace,mirror,raise=5pt},decorate]
  (0,-2) -- node[right=6pt] {$M$} (0,0);
}
\end{eqnarray}
where $c_2:=c_1+n_1$.
\end{lemma}

\begin{proof}
Suppose that a tree $T$ consisting of $N$ edges has a single capacity $c$.
The generating function corresponding to the tree $T$ is $\genfrac{[}{]}{0pt}{}{N+c}{c}$.
Similarly, if the bottom edge of $T$ is equal to $c$, the generating function 
is 
\begin{eqnarray}
\label{singletree}
\genfrac{[}{]}{0pt}{}{N+c}{c}-\genfrac{[}{]}{0pt}{}{N+c-1}{c-1}
=q^{c}\genfrac{[}{]}{0pt}{}{N+c-1}{c}.
\end{eqnarray}

Let $A$ be the ramification point of the trivalent tree in the left hand 
side of Eqn.(\ref{tri-tree}). 
We denote by $f'(M,c_1,m_1,n_1)$ the left hand side of Eqn.(\ref{tri-tree}).
We sum over all possible labellings of the trivalent tree where the edge
just above the point $A$ has the integer $i$.
From Eqn.(\ref{singletree}), we have 
\begin{eqnarray*}
f'(M,c_1,m_1,n_1)
=\sum_{i=0}^{c_1}q^{i(m_1+n_1+1)}\genfrac{[}{]}{0pt}{}{M-m_1-n_1+i-1}{i}
\genfrac{[}{]}{0pt}{}{n_1+c_1-i}{n_1}\genfrac{[}{]}{0pt}{}{n_1+m_1+c_1-i}{m_1}.
\end{eqnarray*}
From a straight forward calculation, we have 
\begin{eqnarray}
\label{tri-tree-binomial}
\genfrac{[}{]}{0pt}{}{n_1+c_1-i}{n_1}\genfrac{[}{]}{0pt}{}{n_1+m_1+c_1-i}{m_1}
=\genfrac{[}{]}{0pt}{}{n_1+m_1+c_1-i}{n_1+m_1}\genfrac{[}{]}{0pt}{}{n_1+m_1}{n_1}.
\end{eqnarray}
Substituting Eqn.(\ref{tri-tree-binomial}) into $f'(M,c_1,m_1,n_1)$, 
Eqn.(\ref{tri-tree}) is equivalent to the following identity:
\begin{eqnarray}
\label{tri-tree-binomial2}
\sum_{i=0}^{c_1}q^{i(x+1)}\genfrac{[}{]}{0pt}{}{M-x+i-1}{i}
\genfrac{[}{]}{0pt}{}{M+x-i}{x}=\genfrac{[}{]}{0pt}{}{M+c_1}{c_1}, 
\end{eqnarray}
where $x=n_1+m_1$.
Let $f(M,c_1)$ be the left hand side of Eqn.(\ref{tri-tree-binomial2}). 
By using $\genfrac{[}{]}{0pt}{}{n+1}{k}=q^{k}\genfrac{[}{]}{0pt}{}{n}{k}
+\genfrac{[}{]}{0pt}{}{n}{k-1}$, 
we have the recurrence relation:
\begin{eqnarray*}
f(M,c_1)&=&\sum_{i=0}^{c_1}q^{i(x+1)}\left\{
q^{M-x}\genfrac{[}{]}{0pt}{}{M+i-x-1}{M-x}+\genfrac{[}{]}{0pt}{}{M+i-x-1}{M-x-1}
\right\}\genfrac{[}{]}{0pt}{}{c_1+x-i}{x} \\
&=&q^{M+1}f(M+1,c_1-1)+f(M,c_1),
\end{eqnarray*}
with $f(0,0)=1$. 
This recurrence relation is equivalent to the one for $q$-binomial coefficient,
which implies that Eqn.(\ref{tri-tree}) holds true.
\end{proof}

Let $\lambda$ (resp. $\lambda'$) be a Dyck path of length $2n$ (resp. $2(n-1)$).
\begin{lemma}
\label{lemma:treedel}
Suppose a Dyck path $\lambda$ is written as $\lambda=U\lambda'D$. 
Then, we have 
\begin{eqnarray*}
P_{\lambda,\lambda_{0}}=P_{\lambda',\lambda_{0}}
\end{eqnarray*}
\end{lemma}
\begin{proof}
The tree $A(\lambda/\lambda_{0})$ has a unique edge $e$ connecting to the root. 
Since the leftmost capacity of $A(\lambda/\lambda_0)$ is zero, the integer on the 
edge $e$ has to be zero.
The tree obtained from $A(\lambda/\lambda_0)$ by deleting the edge $e$ is nothing 
but the tree $A(\lambda'/\lambda_0)$.
The sum of labellings in $A(\lambda/\lambda_0)$ is equal to $A(\lambda'/\lambda_0)$.
Thus, we have $P_{\lambda,\lambda_{0}}=P_{\lambda',\lambda_{0}}$.
\end{proof}

Given two Dyck paths $\lambda_1$ 
of length $2n_1$ and $\lambda_2$ of length $2m_1$,
we denote by $\lambda:=\lambda_1\circ\lambda_2$ the 
concatenation of two Dyck paths.  
\begin{lemma}
\label{lemma:treefac}
Suppose $\lambda=\lambda_1\circ\lambda_2$.
Then, we have 
\begin{eqnarray*}
P_{\lambda,\lambda_{0}}=\genfrac{[}{]}{0pt}{}{n_1+m_1}{n_1}
P_{\lambda_1,\lambda_0}P_{\lambda_{2},\lambda_{0}}
\end{eqnarray*}
\end{lemma}
\begin{proof}

For a Dyck path $\lambda$ of length $2n$, we denote by $A(\lambda/\lambda_0;c_1)$ 
the tree $A(\lambda/\lambda_0)$ with capacities shifted by $c_1$. 
We denote by $P_{\lambda,\lambda_0}^{c_1}$ the generating function
corresponding to $A(\lambda/\lambda_0;c_1)$.
By successive use of Lemma \ref{lemma:tritree}, we have 
\begin{eqnarray*}
P^{c_1}_{\lambda,\lambda_0}=\genfrac{[}{]}{0pt}{}{n+c_1}{c_1}P_{\lambda,\lambda_0}.
\end{eqnarray*}

Let $\lambda_1$ (resp. $\lambda_2$) be a Dyck path of length $2n_1$ (resp. $2m_1$).
We consider the case where $\lambda=\lambda_{1}\circ\lambda_{2}$.
The tree $A(\lambda/\lambda_{0})$ is obtained by connecting two trees 
$A(\lambda_1/\lambda_0)$ and $A(\lambda_2/\lambda_{0};n_1)$ at the root. 
Then, we have
\begin{eqnarray*}
P_{\lambda,\lambda_0}&=&P_{\lambda_{1},\lambda_{0}}P^{n_1}_{\lambda_2,\lambda_0} \\
&=&\genfrac{[}{]}{0pt}{}{n_1+m_1}{n_1}P_{\lambda_{1},\lambda_{0}}P_{\lambda_2,\lambda_0},
\end{eqnarray*}
which implies Lemma holds true.
\end{proof}

\begin{remark}
If $\lambda=\lambda_1\circ\lambda_{2}\circ\lambda_{3}$, $\lambda$ can be 
regarded as $\lambda=(\lambda_1\circ\lambda_2)\circ\lambda_3$ or 
$\lambda=\lambda_1\circ(\lambda_2\circ\lambda_3)$, which 
leads to a different expression in Lemma \ref{lemma:treefac}.
\end{remark}

Given a tree $A(\lambda)$, we enumerate the ramification points
of the tree $A(\lambda)$ by $1,\ldots, R(\lambda)$ where $R(\lambda)$ is 
the number of ramification points. 
Here, a ramification point means that an edge $e$ has two or more edges connecting 
below $e$. 
Fix the $i$-th ramification point $A$. 
At a point $A$, there exists a unique edge $e$ above $A$ and there are $k\ge2$ 
edges $e_1,\ldots, e_{k}$ below $A$.
Let $M_{i}$ be the number of edges of a partial tree connected the point $A$.
Let $N_{i,j}$, $1\le j\le k$, be the number of a partial tree connected to 
the point $A$ such that $e_j$ is the unique edge connected to the root.
Thus, we have $M_i=\sum_{1\le j\le k}N_{i,j}$.
\begin{theorem}
\label{thrm:Dyck-fac}
Let $\lambda$ be a Dyck path and the integers $n_i,m_i$, $1\le i\le R(\lambda)$, 
defined as above.
We have 
\begin{eqnarray}
\label{Dyck-fac}
\sum_{D\in\mathcal{D}(\lambda/\ast)}q^{\mathrm{art}(D)}
=\prod_{i=0}^{R(\lambda)}\genfrac{[}{]}{0pt}{}{M_{i}}{N_{i,1},N_{i,2},\ldots,N_{i,k}}.
\end{eqnarray}
\end{theorem}
\begin{proof}
From Theorem \ref{thrm:DyckKL}, the left hand side of Eqn.(\ref{Dyck-fac}) 
is equal to the Kazhdan--Lusztig polynomial $P_{\lambda,\lambda_0}$. 
We successively apply Lemma \ref{lemma:treedel} and Lemma \ref{lemma:treefac} to 
$P_{\lambda,\lambda_0}$.   
\end{proof}

\section{Incidence matrix and ballot tilings}
\label{sec:im}
\subsection{\texorpdfstring{$q,t$-deformed incidence matrix}{q,t-deformed incidence matrix}}
\label{subsec:im}
Let $\lambda$ and $\mu$ be two Dyck paths and 
$\mathcal{D}(\lambda/\mu)$ be the set of cover-inclusive 
Dyck tilings as in Section \ref{sec:Dyck}.
In \cite{KW11}, Kenyon and Wilson showed that the inverse 
matrix $M^{-1}$ of the incidence matrix $M$ is expressed 
in terms of $|\mathcal{D}_{\lambda,\mu}|$. 
In this subsection, we consider the $q,t$-deformed incidence 
matrix $M$ whose inverse $M^{-1}$ is expressed in terms of 
ballot tilings.

Let $\mathcal{P}_{n}$ be the set of paths of length $n$ which 
consists of $U$ and $D$. The cardinality of $\mathcal{P}_{n}$ 
is $2^{n}$.
We consider operations on a path $\lambda\in\mathcal{P}_n$:
\begin{enumerate}[(A)]
\item We make a pair between adjacent $U$ and $D$ in this order. 
Then, connect the pair into a simple arc.
\item Repeat the procedure $(A)$ until all the $U$'s are to the left 
of all the $D$'s.
\end{enumerate}
Suppose that $U$ and $D$ are connected by a simple arc and 
the positions of $U$ and $D$ from left are $i$ and $j$ with ($i<j$).
The {\it size} of an arc is defined as $(j-i+1)/2$.  

We define three operations $UD$-flipping, $UU$-flipping and $U$-flipping 
on $\mathcal{P}_{n}$ as follows.
The $UD$-flipping is an operation to reverse a pair of $U$ and $D$ ($U$ and $D$ are
not necessarily adjacent) into $D$ and $U$. 
Similarly, $UU$-flipping (resp. $U$-flipping) is an operation to reverse 
a pair of two $U$'s (resp. a single $U$) into two $D$'s (resp. a single $D$).

We consider the following two cases: type BI and type BIII.
\begin{remark}
We call two cases type BI and BIII since an underlying diagram is 
equivalent to the Kazhdan--Lusztig basis of type BI and BIII studied 
in \cite{S14}.
\end{remark}

\paragraph{\bf Type BI}
In addition to the operations (A) and (B), 
we have two more operations on a path:
\begin{enumerate}
\item[(C)] Put a star $(\bigstar)$ on the rightmost $U$ if it exists.
\item[(D)] For remaining $U$'s, we make a pair of adjacent $U$'s from 
right to left. Then, we connect this pair into a simple dashed arc.
\end{enumerate}

We define a relation $\leftarrow_{I}$ on paths in $\mathcal{P}_{n}$. 
We say $\lambda_1\leftarrow_{I}\lambda_{2}$ if $\lambda_1$  can be 
obtained from $\lambda_2$ by $UD$-flippings of the paired $UD$ 
connected by a simple arc, a $U$-flipping of the $U$ with a star 
or $UU$-flippings of the $UU$ pairs connected by a simple dashed arc.

Suppose $\lambda_1$ is obtained from $\lambda_{2}$ by a $U$-flipping 
on the $U$ with a star and the flipped $U$ is at $(2r+1)$-th position 
from right. 
We define the weight $\mathrm{wt}(\lambda_1\leftarrow_{I}\lambda_2):=-q^{r+1}$.

Suppose that $\lambda_1$ is obtained from $\lambda_{2}$ by a $UD$-flipping 
on $U$ and $D$ connected by an arc of size $m$.
We define the weight 
$\mathrm{wt}(\lambda_1\leftarrow_{I}\lambda_2):=-q^{m}$

Suppose that $\lambda_1$ is obtained from $\lambda_{2}$ by a $UU$-flipping 
on two $U$'s connected by a dashed arc of size $m$ and the position of the 
right $U$ is $2r$ from right. 
We define the weight $\mathrm{wt}(\lambda_1\leftarrow_{I}\lambda_2):=-q^{2r+m}$.

Suppose $\lambda_1\leftarrow_{I}\lambda_2$. 
The weight $\mathrm{wt}(\lambda_1\leftarrow_{I}\lambda_2)$ is defined 
as the product of weights of flippings.

\begin{example}
Let $\lambda_1=DDUDU$ and $\lambda_2=UUDUU$. 
The path $\lambda_1$ can be obtained from $\lambda_2$ by a $UU$-flipping and 
a $UD$-flipping.
The weight is $\mathrm{wt}(\lambda_1\leftarrow_{I}\lambda_2)=q^{5}$.
\end{example}

\paragraph{\bf Type BIII}
We apply operations (A) and (B) to a path. We have unpaired $U$'s and 
unpaired $D$'s. 

We define a relation $\leftarrow_{III}$ on paths in $\mathcal{P}_{n}$. 
We say $\lambda_1\leftarrow_{III}\lambda_2$ if $\lambda_1$ can be obtained 
from $\lambda_2$ by $UD$-flippings on $U$ and $D$ connected by a simple arc,
and by $U$-flippings on unpaired $U$'s.

Suppose that $\lambda_1$ is obtained from $\lambda_2$ by a $U$-flipping 
on an unpaired $U$ and there are $s$ arcs right to the $U$.
We define $\mathrm{wt}(\lambda_1\leftarrow_{III}\lambda_2):=-q^{s}t$.

The weight for $UD$-flipping is the same as the one of type BIII. 
The weight $\mathrm{wt}(\lambda_1\leftarrow_{III}\lambda_2)$ is the product 
of the weights of flippings.

\paragraph{Definition of an incidence matrix $M$}
We denote $\leftarrow_{I}$ or $\leftarrow_{III}$ by $\leftarrow$.
\begin{defn}
The incidence matrix $M:=(M_{\lambda,\mu})_{\lambda,\mu\in\mathcal{P}_{n}}$ 
is defined as 
\begin{eqnarray*}
M_{\lambda,\mu}:=\mathrm{wt}(\lambda\leftarrow\mu)\cdot
\delta_{\{\lambda\leftarrow\mu\}}
\end{eqnarray*}
where $\delta_{S}=1$ in case of $S$ is true and $\delta_{S}=0$ otherwise.
\end{defn}
We order the rows and columns according the reversed lexicographic
order on paths, {\it i.e.}, we have $U<D$.
Then, the incidence matrix $M$ is lower triangular.

\begin{example}
\label{example-invM}
The incidence matrix $M_{I}, M_{III}$ on $\mathcal{P}_2$ and their
inverses are
\begin{eqnarray*}
&&M_{I}=
\left(
\begin{array}{cccc}
1 & 0 & 0 & 0 \\
-q & 1 & 0 & 0 \\
0 & -q & 1 & 0 \\
0 & 0 & -q & 1 \\
\end{array}
\right), \qquad
M^{-1}_{I}=
\left(
\begin{array}{cccc}
1 & 0 & 0 & 0 \\
q & 1 & 0 & 0 \\
q^2 & q & 1 & 0 \\
q^3 & q^2 &  q & 1 \\
\end{array}
\right), \\ 
&&M_{III}=
\left(
\begin{array}{cccc}
1 & 0 & 0 & 0 \\
-t & 1 & 0 & 0 \\
-t & -q & 1 & 0 \\
t^2 & 0 & -t & 1 \\
\end{array}
\right), \qquad
M^{-1}_{III}=
\left(
\begin{array}{cccc}
1 & 0 & 0 & 0 \\
t & 1 & 0 & 0 \\
(1+q)t & q & 1 & 0 \\
qt^2 & qt &  t & 1 \\
\end{array}
\right),
\end{eqnarray*}
where the order of bases is $(UU,UD,DU,DD)$.
\end{example}

\subsection{Ballot tiling}
A {\it ballot path of length $(2n,n')$} is a lattice path from 
the origin $(0,0)$ to $(2n+n',n')$ with $U$ steps and 
$D$ steps, which does not go below the horizontal line $y=0$.
We denote by $\mathcal{B}_{N}$ the set of ballot paths of length 
$N=2n+n'$. 
A Dyck path in Section \ref{sec:Dyck} is a ballot path of length $(2n,0)$.
The highest path in $\mathcal{B}_{N}$ is a path consisting of 
only $U$'s and the lowest path is the zig-zag path $\underbrace{UDUD\ldots}_{N}$.

A ballot path $\lambda$ of length $N$ is identified with the shifted Young diagram 
which is determined by the path $\lambda$, the line $y=x$ and boxes whose center 
is on the line $x=N$. 
We call a box whose center is on the line $x=N$ an {\it anchor} box.
See Figure \ref{fig:sYoungD} for an example.
If the shifted skew shape $\lambda/\mu$ exists, we call $\lambda$ and $\mu$ 
the lower path and the upper path.

We denote by $\mathfrak{S}_{N}^{C}$ the Weyl group associated to the Dynkin 
diagram of type $C$. 
We denote by $s_i$ with $1\le i\le N$ the elementary transposition of 
$\mathfrak{S}_{N}^{C}$.
The Weyl group $\mathfrak{S}_{N}^{C}$ is isomorphic to the signed permutation
group and its order is $2^{N}\cdot N!$. 
The elementary transposition $s_i$ is a transposition of $(i,i+1)$ for 
$1\le i\le N-1$ and $s_N$ is a transposition $(N,\underline{N})$.
Here, we use the bar-notation instead of the minus sign. 
For a reduced word $w:=s_{i_1}\ldots s_{i_p}$, we define 
inversions as follows: 
\begin{eqnarray*}
\mathrm{inv}_{1}(w)&:=&\#\{i_j \ |\ 1\le j\le p, 1\le i_j\le N-1\}, \\
\mathrm{inv}_{2}(w)&:=&\#\{i_j \ |\  1\le j\le p, i_j=N\}, 
\end{eqnarray*}
and $\mathrm{Inv}(w):=\mathrm{inv}_{1}(w)+\mathrm{inv}_{2}(w)$.

The shifted Young diagram $\lambda$ gives a 
word $\omega(\lambda)\in\mathfrak{S}_{N}^{C}$ as follows.
Fix a path $\lambda\in\mathcal{B}_{N}$. 
We enumerate the all $U$'s and $D$'s from right by $1,2,\ldots N$.
We denote by $\mathrm{Pos}(\lambda)$ the set of the positions of $D$ in $\lambda$.
For the $i$-th $D$, we assign a word $\omega_{i}(\lambda):=s_{N-i+1}\ldots s_{N}$.
Suppose that the set is written as $\mathrm{Pos}(\lambda)=\{i_1<i_2<\cdots<i_{N_{D}}\}$ 
where $N_D$ is the number of $D$ in $\lambda$.
Then, the word $\omega(\lambda)$ is defined as 
\begin{eqnarray*}
\omega(\lambda)&:=&\overset{\rightarrow}{\prod_{i\in\mathrm{Pos}(\lambda)}}\omega_{i}(\lambda) \\
&=&\omega_{i_1}(\lambda)\omega_{i_2}(\lambda)\cdots\omega_{i_{N_{D}}}(\lambda),
\end{eqnarray*}
where the product $\overset{\rightarrow}{\prod}$ is an ordered product.
See Figure \ref{fig:sYoungD} for an example.

\begin{figure}[ht]
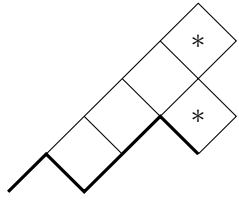

\tikzpic{-0.5}{
\draw(0,0)--(5/2,5/2)--(6/2,4/2)--(5/2,3/2)--(6/2,2/2)--(5/2,1/2);
\draw(2/2,2/2)--(3/2,1/2)(3/2,3/2)--(4/2,2/2)--(5/2,3/2)--(4/2,4/2);
\draw[very thick](0,0)--(1/2,1/2)--(2/2,0)--(4/2,2/2)--(5/2,1/2);
\draw(5/2,2/2)node{$\ast$}(5/2,4/2)node{$\ast$};
}
\caption{The shifted Young diagram associated to the path $\lambda=UDUUD$.
The boxes with an asterisk are anchor boxes.
The word is $\omega(\lambda)=s_{5}s_{2}s_{3}s_{4}s_{5}$.
}
\label{fig:sYoungD}
\end{figure}

A {\it ballot tile} is a ribbon such that the centers of the boxes form a ballot path.
A {\it ballot tiling} is a tiling of a shifted skew Young diagram $\lambda/\mu$ 
by ballot tiles. 
We consider a cover-inclusive tiling with a condition: the rightmost box of a ballot 
tile of length $(2n,n')$ with $n'\ge1$ is on an anchor box of the skew shape $\lambda/\mu$.
We have two types of a ballot tiling: type BI and type BIII.

\paragraph{\bf Type BI}
We put a constraint on a ballot tiling. The number of ballot tiles of length $(2n,n')$ 
is even for $n'\in 2\mathbb{Z}+1$ and zero for $n'\in2\mathbb{N}_{+}$.
The statistics area and tiles for a ballot tile $B$ of length $(2n,n')$ with $n'\ge1$
are defined as 
$\mathrm{area}(B):=2n+1+n'$ and $\mathrm{tiles}(B):=0$.
The statistics on Dyck tile $D$ (that is a ballot tile of length $(2n,0)$) are given 
by the same as in the case of Dyck tilings, {\it i.e.}, 
$\mathrm{area}(D)=2n+1$ and $\mathrm{tiles}(D)=1$.
The statistics art for a ballot tile 
is defined as $\mathrm{art}(B):=(\mathrm{area}(B)+\mathrm{tiles}(B))/2$.

\paragraph{\bf Type BIII}
Let $B$ be a ballot tile of length $(2n,n')$.
We denote by (S1) the following statement:
\begin{enumerate}
\item[(S1)] 
The rightmost box of a ballot tile (including a Dyck tile) is on an anchor box.
\end{enumerate}
We define two statistics area and tiles by $\mathrm{area}(B):=2n+1$ and 
$\mathrm{tiles}(B):=1$.
Then, statistics art is define by 
\begin{eqnarray*}
\mathrm{art}(B):=
\begin{cases}
(\mathrm{area}(B)-\mathrm{tiles}(B))/2, & \text{(S1) is true}, \\
(\mathrm{area}(B)+\mathrm{tiles}(B))/2, & \text{otherwise}.
\end{cases}
\end{eqnarray*}
Given a ballot tiling $T$, we denote by $\mathrm{tiles}_2(T)$
the number of ballot tiles satisfying the statement (S1).

We denote by $\mathcal{B}^{X}(\lambda/\mu)$ with $X=I$ or $III$ 
the set of cover-inclusive ballot tilings of type $X$  
in the skew shape $\lambda/\mu$.  
We define 
\begin{eqnarray*}
\mathcal{B}^{X}(\lambda/\ast):=\bigcup_{\mu}\mathcal{B}^{X}(\lambda/\mu), \\
\mathcal{B}^{X}(\ast/\mu):=\bigcup_{\lambda}\mathcal{B}^{X}(\lambda/\mu).
\end{eqnarray*}

\begin{defn}
We define generating functions of ballot tilings:
\begin{eqnarray*}
P^{I}_{\lambda,\mu}&:=&\sum_{B\in\mathcal{B}(\lambda/\mu)}q^{\mathrm{art}(B)}, \\
P^{III}_{\lambda,\mu}&:=&\sum_{B\in\mathcal{B}(\lambda/\mu)}
q^{\mathrm{art}(B)}t^{\mathrm{tiles}_{2}(B)},
\end{eqnarray*}
and 
\begin{eqnarray*}
P^{X}_{\lambda,\ast}&:=&\sum_{\mu}P^{X}_{\lambda,\mu}, \\
P^{I}_{\ast,\mu}&:=&
\sum_{B\in\mathcal{B}^{I}(\ast/\mu)}q^{\mathrm{tiles}(B)},
\end{eqnarray*}
where $X=I$ or $III$.
\end{defn}

\begin{example}
We consider the set $\mathcal{B}^{I}(\lambda/\ast)$ where $\lambda=DDUDU$.
We have eight ballot tilings in $\mathcal{B}^{I}(\lambda/\ast)$ whose 
statistic art is equal to five:
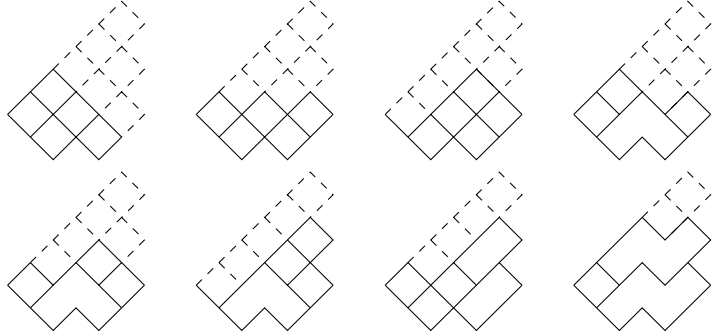
\begin{eqnarray*}
\tikzpic{-0.5}{
\draw(0,0)--(0.6,-0.6)--(0.9,-0.3)--(1.2,-0.6)--(1.5,-0.3);
\draw(0,0)--(0.6,0.6)--(1.5,-0.3);
\draw(0.3,0.3)--(0.9,-0.3)(0.3,-0.3)--(0.9,0.3)(0.9,-0.3)--(1.2,0);
\draw[dashed](0.6,0.6)--(1.5,1.5)--(1.8,1.2)--(0.9,0.3)
(1.2,1.2)--(1.8,0.6)--(1.2,0)(0.9,0.9)--(1.8,0)--(1.5,-0.3);
}\quad
\tikzpic{-0.5}{
\draw(0,0)--(0.6,-0.6)--(0.9,-0.3)--(1.2,-0.6)--(1.8,0);
\draw(0,0)--(0.3,0.3)--(0.9,-0.3)--(1.5,0.3)--(1.8,0);
\draw(0.3,-0.3)--(0.9,0.3)--(1.5,-0.3);
\draw[dashed](0.3,0.3)--(1.5,1.5)--(1.8,1.2)--(0.9,0.3)
(0.6,0.6)--(0.9,0.3)(0.9,0.9)--(1.5,0.3)(1.2,1.2)--(1.8,0.6)--(1.5,0.3);
}\quad
\tikzpic{-0.5}{
\draw(0,0)--(0.6,-0.6)--(0.9,-0.3)--(1.2,-0.6)--(1.8,0);
\draw(0.3,-0.3)--(1.2,0.6)--(1.8,0)(0.6,0)--(0.9,-0.3)--(1.5,0.3)
(0.9,0.3)--(1.5,-0.3);
\draw[dashed](0,0)--(1.5,1.5)--(1.8,1.2)--(1.2,0.6)
(0.3,0.3)--(0.6,0)(0.6,0.6)--(0.9,0.3)(0.9,0.9)--(1.5,0.3)
(1.2,1.2)--(1.8,0.6)--(1.5,0.3);
}\quad
\tikzpic{-0.5}{
\draw(0,0)--(0.6,-0.6)--(0.9,-0.3)--(1.2,-0.6)--(1.8,0);
\draw(0,0)--(0.6,0.6)--(1.5,-0.3)(1.2,0)--(1.5,0.3)--(1.8,0);
\draw(0.3,0.3)--(0.6,0)(0.3,-0.3)--(0.9,0.3);
\draw[dashed](0.6,0.6)--(1.5,1.5)--(1.8,1.2)--(0.9,0.3)
(0.9,0.9)--(1.5,0.3)(1.2,1.2)--(1.8,0.6)--(1.2,0);
} \\
\tikzpic{-0.5}{
\draw(0,0)--(0.6,-0.6)--(0.9,-0.3)--(1.2,-0.6)--(1.8,0);
\draw(0,0)--(0.3,0.3)--(0.6,0)(0.3,-0.3)--(1.2,0.6)--(1.8,0)
(0.9,0.3)--(1.5,-0.3)(1.2,0)--(1.5,0.3);
\draw[dashed](0.3,0.3)--(1.5,1.5)--(1.8,1.2)--(1.2,0.6)
(0.6,0.6)--(0.9,0.3)(0.9,0.9)--(1.2,0.6)(1.2,1.2)--(1.8,0.6)--(1.5,0.3);
}\quad
\tikzpic{-0.5}{
\draw(0,0)--(0.6,-0.6)--(0.9,-0.3)--(1.2,-0.6)--(1.8,0);
\draw(0.3,-0.3)--(1.5,0.9)--(1.8,0.6)--(1.5,0.3)--(1.8,0);
\draw(0.9,0.3)--(1.5,-0.3)(1.2,0)--(1.5,0.3)--(1.2,0.6);
\draw[dashed](0,0)--(1.5,1.5)--(1.8,1.2)--(1.5,0.9)
(0.3,0.3)--(0.6,0)(0.6,0.6)--(0.9,0.3)(0.9,0.9)--(1.2,0.6)
(1.2,1.2)--(1.5,0.9);
}\quad
\tikzpic{-0.5}{
\draw(0,0)--(0.6,-0.6)--(0.9,-0.3)--(1.2,-0.6)--(1.8,0);
\draw(0,0)--(0.3,0.3)--(0.9,-0.3)--(1.5,0.3)--(1.8,0);
\draw(0.3,-0.3)--(1.5,0.9)--(1.8,0.6)--(1.5,0.3)(0.9,0.3)--(1.2,0);
\draw[dashed](0.3,0.3)--(1.5,1.5)--(1.8,1.2)--(1.5,0.9)
(0.6,0.6)--(0.9,0.3)(0.9,0.9)--(1.2,0.6)(1.2,1.2)--(1.5,0.9);
}\quad
\tikzpic{-0.5}{
\draw(0,0)--(0.6,-0.6)--(0.9,-0.3)--(1.2,-0.6)--(1.8,0);
\draw(0,0)--(0.9,0.9)--(1.2,0.6)--(1.5,0.9)--(1.8,0.6)--(1.2,0)
--(0.9,0.3)--(0.3,-0.3);
\draw(1.8,0)--(1.5,0.3)(0.3,0.3)--(0.6,0);
\draw[dashed](0.9,0.9)--(1.5,1.5)--(1.8,1.2)--(1.2,0.6)
(1.2,1.2)--(1.5,0.9);
}
\end{eqnarray*}
\end{example}

By a similar argument to Section 1 in \cite{KW11}, we have the following 
theorem:
\begin{theorem}
The inverse of the incidence matrix $M_{X}$ ($X=I$ or $III$) is given by
$(M^{-1}_{X})_{\lambda,\mu}=P^{X}_{\lambda,\mu}$.
\end{theorem}

\section{Ballot tiling of type BIII}
\label{sec:BallotBIII}
Let $\lambda$ be a ballot path. 
The number of $U$ in $\lambda$ may not be equal to that of 
$D$ in $\lambda$. 
A Dyck path can be obtained by a concatenation of $\lambda$,  
a sequence of $U$'s from left and a sequence of $D$'s from right.  
By abuse of notation, we also denote the obtained Dyck path 
by $\lambda$.

We denote by $N_{D}(\lambda)$ the number of $D$ in $\lambda$.
\begin{theorem}
\label{thrm:BIII-fac}
We have 
\begin{eqnarray}
\label{BIII-fac}
P_{\lambda,\ast}^{III}=P^{\mathrm{Dyck}}_{\lambda,\ast}\cdot
\prod_{i=1}^{N_{D}(\lambda)}(1+q^{i-1}t).
\end{eqnarray}
\end{theorem}
\begin{proof}
We give a bijective proof of Eqn.(\ref{BIII-fac}).
A term from the right hand side of Eqn.(\ref{BIII-fac}) 
is a product of $q^{\mathrm{art}(T)}$ for a Dyck tiling $T$ 
and a factor $q^{i_1+\ldots+i_{r}}t^{r}$ with 
$\{1\le i_1<i_2<\ldots<i_r\le N_{D}(\lambda)\}$. 

We divide the shifted shape $\lambda$ into two domains:
the first one is the domain determined by two Dyck paths $\lambda$ 
and $\lambda_{H}:=\underbrace{U\ldots U}_{N-N_{D}(\lambda)}\underbrace{D\ldots D}_{N_{D}(\lambda)}$,
and the second one is the domain determined by two ballot paths 
$\lambda_{H}$ and $\lambda_0$.
The second domain is the staircase of $N_{D}(\lambda)$ rows. 
We put a Dyck tiling $T$ in the first domain. 
For a factor $q^{i_1+\ldots i_r}t^{r}$, we put $i_{r+1-i}$ boxes and a box with $\ast$ 
on the $i$-th row from the bottom in the second domain such that 
the box $b$ with $\ast$ is on an anchor box and other boxes are adjacent to the box $b$
(See the left picture in Figure \ref{fig:BI} as an example).

\begin{figure}[ht]
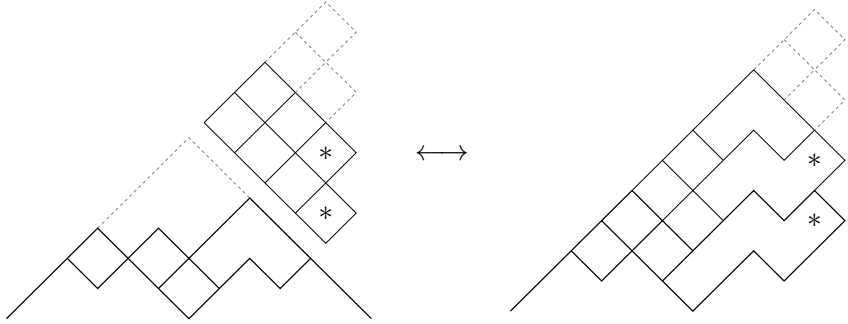

\scalebox{0.4}{
\tikzpic{-0.5}{
\draw[very thick](0,0)--(2,2)--(3,1)--(4,2)--(6,0)--(8,2)--(9,1)--(10,2)--(12,0);
\draw[gray,dashed](3,3)--(6,6)--(10,2);
\draw[very thick](2,2)--(3,3)--(4,2)--(5,3)--(7,1)
     (5,1)--(6,2)--(8,4)--(10,2);

\draw (6.5,6.5)--(10.5,2.5)--(11.5,3.5)--(7.5,7.5)--(6.5,6.5);
\draw(7.5,7.5)--(8.5,8.5)--(11.5,5.5)--(9.5,3.5);
\draw(7.5,5.5)--(9.5,7.5)(8.5,4.5)--(10.5,6.5);
\draw[gray,dashed](8.5,8.5)--(10.5,10.5)--(11.5,9.5)--(10.5,8.5)--(11.5,7.5)--(10.5,6.5)
                   (9.5,9.5)--(10.5,8.5)--(9.5,7.5);
\draw(10.5,3.5)node{\scalebox{2.5}{$\ast$}}(10.5,5.5)node{\scalebox{2.5}{$\ast$}};
}}
\quad$\longleftrightarrow$\quad
\scalebox{0.4}{
\tikzpic{-0.5}{
\draw[very thick](0,0)--(2,2)--(3,1)--(4,2)--(6,0)--(8,2)--(9,1)--(10,2);
\draw[very thick](2,2)--(3,3)--(4,2)--(5,3)--(6,2)
     (5,1)--(6,2)(3,3)--(4,4)--(5,3)
     (6,2)--(8,4)--(9,3)--(10,4)--(11,3)--(10,2);
\draw(10,3)node{\scalebox{2.5}{$\ast$}};
\draw(4,4)--(5,5)--(6,4)--(5,3)(6,4)--(7,3)
     (6,4)--(8,6)--(9,5)--(10,6)--(11,5)--(10,4);
\draw(10,5)node{\scalebox{2.5}{$\ast$}};
\draw(5,5)--(6,6)--(7,5);
\draw(6,6)--(8,8)--(10,6);
\draw[gray,dashed](8,8)--(10,10)--(11,9)--(10,8)--(11,7)--(10,6)
                  (9,9)--(10,8)--(9,7);
}
}
\caption{The bijection proving Theorem \ref{thrm:BIII-fac}. The lower path is $\lambda=UUDUDDUUDU$.
The left picture is a term $q^{5}\cdot q^{3}t\cdot q^{2}t$ from the right hand side 
of Eqn.(\ref{BIII-fac}) and the right picture is a ballot tiling from the left hand 
side of Eqn.(\ref{BIII-fac}).
}
\label{fig:BI}
\end{figure}

Let $\mu$ be the upper path of the Dyck tiling $T$ and 
$\mathrm{UP}(\mu)$ be the set of $U$ in $\mu$.
We define $\varphi:\mathrm{UP}(\mu)\rightarrow \mathbb{N}$ 
by $\varphi(U)$ as the number of boxes between $U$ and the Dyck 
path $\lambda_H$ in the $(1,-1)$-direction.

We define an operation on $T$ in the first domain and $i_r+1$ boxes
on the bottom row in the second domain as follows.

Suppose that $\varphi(U)>N_{D}(\lambda)-i_r-1$ with $U\in\mathrm{UP}(\mu)$.
We translate by $(1,0)$ the boxes (including Dyck tiles) 
$T'$ between $U\in\mathrm{UP}(\mu)$ and $\lambda$ in the 
$(1,-1)$-direction. Thus, $T$ is divide into two tilings, {\it i.e.}, $T=T'\sqcup T''$ 
where a tiling $T''$ is attached to the path $\lambda$.
We denote by $\lambda'$ the upper path of the tiling $T''$.

Let $U_x$ be the up step of $\mu$ such that $\varphi(U_x)\le N_{D}(\lambda)-i_r-1$
and $\varphi(U_{x+1})>N_{D}-i_r-1$.
We denote by $D_x$ the down step $D$ in $\mu$ (or equivalently $\lambda'$) 
such that $D_x$ is $N_{D}(\lambda)-i_r-\varphi(U_x)$ steps right to $U_x$.

We put ballot tiles along the shape $\lambda'$ from the lowest anchor box $a$ to 
the down step $D_{x}$ of $\lambda'$. 
Let $L$ be the locally lowest point of $\lambda'$ between $a$ and $D_{x}$.
We put a ballot tile from $L$ to the anchor box $a$ along $\lambda'$.
We put Dyck tiles from $L$ to $D_{x}$ along $\lambda'$ such that we have 
maximal Dyck tiles.
If $L$ and $D_x$ are at the same height, we merge the ballot tile from $L$ 
to $a$ and a Dyck tile from $L$ to $D_x$ into a ballot tile from $D_x$ 
to $a$.
Together with $T'$, this operation gives a ballot tiling above $\lambda'$.

We repeat the above procedure until boxes corresponding to a term $q^{i_1}t$ 
are transformed into ballot tiles (see the right picture in Figure \ref{fig:BI}).
It is easy to see that this operation is invertible and the exponent of 
a ballot tiling is preserved.
Thus, we have a bijection regarding to Eqn.(\ref{BIII-fac}).

\end{proof}

\begin{example}
From the inverse of incidence matrix in Example \ref{example-invM}, 
we have 
\begin{eqnarray*}
P^{III}_{UU,\ast}=1, \quad P^{III}_{UD,\ast}=1+t, \quad P^{III}_{DU,\ast}=(1+q)(1+t), 
\quad P^{III}_{DD,\ast}=(1+t)(1+qt).
\end{eqnarray*}

\end{example}

Let $\lambda_0$ be a path consisting of only $U$'s. 
Recall that two paths $\lambda_0$ and $\lambda$ determine the shifted 
shape $\lambda$.
For each $U$ in $\lambda_0$, we define $\mathrm{ht}(U)$ as one plus the number 
of boxes between $U$ and a path $\lambda$ or an anchor box in the $(1,-1)$ direction.
Let $\mathrm{UP}_{\ast}(\lambda)$ (resp. $\mathrm{UP}_{\circ}(\lambda)$) 
be the set of $U$'s in $\lambda_0$ such that there is an anchor box  (resp. no anchor box) 
in the $(1,-1)$-direction of $U$.

We denote by $\mathfrak{S}_{N}^{C}(\underline{1}2,\underline{12})$ the set of 
$(\underline{1}2,\underline{12})$-avoiding signed permutations.
\begin{theorem}
\label{thrm:BIII-2}
Let $\pi_0=\omega(\mu)$. Then, 
\begin{eqnarray}
\label{BIII-2}
\sum_{\substack{\pi\le\pi_0 \\ \pi\in\mathfrak{S}_{N}^{C}(\underline{1}2,\underline{12})}}
q^{\mathrm{inv}_{1}(\pi)}t^{\mathrm{inv}_{2}(\pi)}
=\prod_{U\in\mathrm{UP}_{\circ}(\mu)}[\mathrm{ht}(U)]
\prod_{U\in\mathrm{UP}_{\ast}(\mu)}[\mathrm{ht}(U)]_{t},
\end{eqnarray}
where $\pi\le\pi_0$ is the Bruhat order on $\mathfrak{S}_{N}^{C}$.
\end{theorem}

\begin{remark}
Theorem \ref{thrm:BIII-2} is an analogue of the second and the third 
terms in Theorem \ref{thrm-A2}.
\end{remark}

\begin{proof}[Proof of Theorem \ref{thrm:BIII-2}]
We rotate the shifted shape $\mu$ $45$ degrees clockwise.
For a term $q^{i}$, $0\le i\le \mathrm{ht}(U)-1$, in 
$[\mathrm{ht}(U)]$ or $[\mathrm{ht}(U)]_{t}$, we put $i$ boxes from top to bottom 
in the same column as $U$.
Similarly, for a term $q^{\mathrm{ht}(U)-2}t$ in $[\mathrm{ht}(U)]_{t}$, 
we put $\mathrm{ht}(U)-1$ boxes from top to bottom in the same column as $U$.
We denote by $T$ a configuration of boxes in the shifted shape $\mu$.
Recall that each box in the shape $\mu$ corresponds to an 
elementary transposition $s_i$ by $\omega(\mu)$. 
We assign a word in $\mathfrak{S}_{N}^{C}$ to boxes in $\mu$ 
by reading transpositions in boxes from top to bottom and from right 
to left, {\it i.e.}, starting with the rightmost column, reading 
down transpositions from top to bottom, moving to the column left to 
the rightmost column, working down from top to bottom and continue 
the procedure until it ends at the bottom row of the leftmost column.
We denote by $\omega_{T}(\mu)$ the word obtained in this way.

It is obvious that the word $\pi:=\omega_{T}(\mu)$ satisfies 
$\pi\le\pi_{0}$ in the Bruhat order.
The existence of two patterns $\underline{1}2$ and $\underline{12}$ 
means that $\pi$ contains the partial sequence $s_{N-1}s_{N}s_{N-1}$. 
Suppose $\pi$ contains the partial sequence $s_{N-1}s_{N}s_{N-1}$. 
Then, the existence of a box $b$ corresponding the right transposition $s_{N-1}$ 
means that there are boxes left to $b$. 
Such boxes correspond to a transposition $s_j$ with $j\le N-2$.
Since we read boxes from top to bottom and from right to left, 
$\pi$ partially contains a sequence $s_{N-1}ws_{N}w's_{N-1}$ where 
$w$ or $w'$ contains $s_{N-2}$. 
There is no braid relations such that a reduced word $\pi$ 
contains $s_{N-1}s_{N}s_{N-1}$.
This contradicts the assumption.
Similarly, if $\pi$ contains $s_{N-2}$ right to the right $s_{N-1}$,
there is $s_{N-3}$ right to the left $s_{N-2}$ in $w_{T}(\mu)$.
A reduced word for $\pi$ does not contain a sequence 
$s_{N-1}s_{N}s_{N-2}s_{N-1}s_{N-2}=s_{N-1}s_{N}s_{N-1}s_{N-2}s_{N-1}$. 
From these observations, $\pi$ is a 
$(\underline{1}2,\underline{12})$-avoiding signed permutation.
The inversion $\mathrm{inv}_1(\pi)$ is equal to the number of 
boxes except anchor boxes and the inversion $\mathrm{inv}_2(\pi)$ 
is equal to the number of anchor boxes.

Suppose that a $(\underline{1}2,\underline{12})$-avoiding signed permutation 
$\pi$ satisfies $\pi\le\pi_0$.
We show that there exists a unique expression of $\pi$ such that 
$\pi$ is obtained as $\omega_{T}(\mu)$ with a certain $T$.
Different expressions of a word $\pi\in\mathfrak{S}_{N}^{C}$ are 
obtained by using the relations $s_is_{i+1}s_{i}=s_{i+1}s_{i}s_{i+1}$ 
for $1\le i\le N-1$, $s_is_j=s_js_i$ for $|i-j|>1$ and 
$s_{N-1}s_{N}s_{N-1}s_{N}=s_{N}s_{N-1}s_{N}s_{N-1}$.
Since $\pi$ is a $(\underline{1}2,\underline{12})$-avoiding signed permutation,
$\pi$ contains neither $s_{N-1}s_{N}s_{N-1}s_{N}$, $s_{N}s_{N-1}s_{N}s_{N-1}$
nor $s_{N-1}s_{N}s_{N-1}$.
By a construction of $T$, $\omega_{T}(\mu)$ do not contain partial 
sequences $s_{i}s_{j}$ with $i<j$ and $s_{i}s_{i+1}s_{i}$ with $1\le i\le N-1$.
Thus, given an expression of $\pi$, we have a unique expression of $\pi$ such 
that $\pi=\omega_{T}(\mu)$ with some $T$.

We have a correspondence a term in the right hand side of Eqn.(\ref{BIII-2})
with $\pi=\omega_{T}(\mu)$ where $\pi$ is in 
$\mathfrak{S}_{N}^{C}(\underline{1}2,\underline{12})$ and satisfies $\pi\le\pi_0$.
This completes the proof.
\end{proof}

\section{Planted plane tree and Generalized perfect matching}
\label{sec:pptgpm}
\subsection{Planted plane tree}
\label{sec:ppt}
We introduce the notion of a tree for a ballot path following \cite{Boe88}.
Let $\lambda$ be a ballot path.
Recall the definition of $\mathcal{Z}$ in Section \ref{sec:tree}.
If a path $\lambda=\lambda' \underline{U}z$ with $z\in\mathcal{Z}$, 
we call the underlined $U$ a {\it terminal} $U$.
If a path $\lambda=\lambda'\underline{U}z_{2r}\underline{U}z_{2r-1}U\cdots
z_3Uz_2Uz_1Uz_{0}$ with $r\ge1$ and $z_i\in\mathcal{Z}$, we call 
underlined two $U$'s a $UU$-pair.
A $U$ which is not classified above is called an {\it extra} $U$.

We define the capacity for a pair $UD$ in the same way as in the case of a Dyck word.
If $\lambda=\lambda'\underline{U}$ and $\lambda_{0}=\lambda'_0\alpha$ with $\alpha\in\{U,D\}$, 
the capacity of the underlined $U$ is 
\begin{eqnarray*}
\mathrm{cap}(U):=||\lambda_0||_{U}-||\lambda'||_{U}.
\end{eqnarray*}

We define a plane tree $A(\lambda)$ for a ballot path $\lambda$. 
A tree $A(\lambda)$ satisfies $(\Diamond1)$ to $(\Diamond4)$ (see Section \ref{sec:tree}) 
and the following four conditions $(\Diamond5)$ to $(\Diamond8)$:
\begin{enumerate}
\item[($\Diamond5$)]
If the underlined $U$ in $\underline{U}\lambda'$ is the terminal $U$, 
the tree $A(U\lambda')$ is obtained by putting an edge just above 
the tree $A(\lambda')$.
We mark the edge with a dot ($\bullet$).
\item[($\Diamond6$)]Suppose underlined $U$'s in $\underline{U}z\underline{U}\lambda'$ with 
$z\in\mathcal{Z}$ are a $UU$-pair. 
Then the tree $A(UzU\lambda')$ is obtained by attaching an edge above the root of  
$A(z\lambda')$. We mark the edge with a dot ($\bullet$).
\item[($\Diamond7$)]If the underlined $U$ in $\underline{U}\lambda'$ is an extra $U$, 
we have $A(U\lambda')=A(\lambda')$.
\end{enumerate}
We need to encode an additional information on a tree \cite{Boe88}.
Suppose that $\lambda=\lambda'z_{2r+1}\lambda''$ with 
$\lambda''=Uz_{2r}Uz_{2r-1}\ldots z_{1}Uz_{0}$ and $z_{2r+1}=x_{s}x_{s-1}\cdots x_1$, 
$x_i\in\mathcal{Z}$.
Here, all $x_i$'s can not be expressed as a concatenation of two non-empty Dyck words.
Thus, the tree $A(x_i)$ contains a unique maximal edge (the edge connecting to the root)
corresponding to a $UD$ pair. 
The tree $A(\lambda'')$ also has a unique maximal edge corresponding to a $UU$-pair 
or a terminal $U$.
The maximal edge of $A(x_i)$ (resp. $A(\lambda'')$) is said to {\it immediately precede} 
the maximal edge of $A(x_{i+1})$ (resp. $A(x_1)$) for $1\le i\le s$. 
Then we put an information on a tree:
\begin{enumerate}
\item[($\Diamond8$)] When an edge $e$ immediately precedes an edge $e'$ in a binary tree
$A(\lambda)$, we put a dashed arrow from $e$ to $e'$.
\end{enumerate}

\begin{remark}
We can embed a ballot path $\lambda$ into a set of longer ballot paths by 
adding $U$'s left to $\lambda$.
Then, we have a ballot path without $D$'s satisfying ($\Diamond4$). 
In other words, all $D$'s are paired with $U$'s. 
Especially, when a ballot path $\lambda$ consists of only $D$'s, 
that is, $\lambda=\underbrace{D\ldots D}_{N}$, 
we embed $\lambda$ in the set of ballot paths of length $2N$ 
by adding $N$ $U$'s left to $\lambda$.
\end{remark}

\begin{example}
When $\lambda=UUDDUUUUUDDU$, the tree $A(\lambda)$ is depicted as 
\begin{eqnarray*}
\scalebox{0.7}{
\tikzpic{-0.5}{
\coordinate
	child{coordinate (c0)   
	child{coordinate (c01)}
	child[missing]
	}
	child {coordinate (c2)}{
        child {coordinate (c3) child {coordinate (c4)} child[missing]}
	child {coordinate (c5)}
	};
\draw[very thick](c2)node{$-$} (c3)node{$-$} (c4)node{$-$} (c5)node{$-$}(c0)node{$-$}(c01)node{$-$};
\draw[very thick,latex-,dashed] ($(c2)!.7!(c3)$)--($(c2)!.7!(c5)$);
\node at ($(c2)!.5!(c5)$){\scalebox{1.5}{$\bullet$}};
\node at ($(0,0)!.5!(c2)$){\scalebox{1.5}{$\bullet$}};
}
}.
\end{eqnarray*}
\end{example}

As in the case of Section \ref{sec:tree}, a {\it labelling of Lascoux--Sch\"utzenberger type}
(type LS for short) is a set of non-negative integers on the edges of $A(\lambda/\lambda_0)$ 
satisfying (LS1) and (LS2).

\begin{remark}
One can compute a Kazhdan--Lusztig polynomial for the Hermitian symmetric pair 
$(B_{N},A_{N-1})$ by a tree $A(\lambda/\mu)$. 
It requires two more additional information on a tree 
in addition to (LS1) and (LS2). 
The extra information comes from the existence of arrows and dotted edges.
See~\cite{Boe88,S14} for a detailed explanation. 
\end{remark}

Let $E$ be the number of edges in $A(\lambda)$.
A {\it natural labelling} of $A(\lambda)$ is a labelling such that 
integers on edges are increasing from the root to a leaf and an integer 
in $\{1,\ldots,E\}$ appears exactly once. 
A {\it reference} natural labelling of $A(\lambda)$ is a natural labelling 
such that integers on edges are increasing in left-to-right depth-first search. 
We denote by $\mathrm{pre}'(L)$ (resp. $\mathrm{post}'(L)$) the permutation 
obtained by reading a natural labelling $L$ from left to right using 
the pre-order (resp. modified post-order).  
The pre-order (resp. post-order) means that we visit a node before (resp. after)
both of its left and right subtrees.
The modified post-order means as follows. We visit edges on a tree one-by-one by
following the post-order. We visit edges $e$'s with $\bullet$ from top to bottom 
up to a ramification point soon after all the edges on the subtree which is left 
to edges $e$'s are visited. 
Then, we continue visiting remaining edges by following the modified post-order.  
Let $\pi=\mathrm{post}'(L):=\pi_1\ldots\pi_{E}$ be a permutation in the one-line notation. 
If $\pi_{i}$ with $1\le i\le E$ is on an edge with a dot, we put an underbar on $\pi_{i}$.
The obtained barred permutation is denote by $\mathrm{post}(L)$.
Similarly, we obtain $\mathrm{pre}(L)$ by adding underbars to $\mathrm{pre}'(L)$.
We identify $\mathrm{post}(L)$ and $\mathrm{pre}(L)$ with signed permutations 
in $\mathfrak{S}_{E}^{C}$ by regarding an unbarred integer as a negative integer.
Note that the pre-order word for the reference natural labelling is the identity.

We define the {\it inversion} of a signed permutation $\sigma$ in the signed 
alphabet $\{\pm1,\ldots,\pm E\}$ as 
\begin{eqnarray*}
\mathrm{Inv}(\sigma)&:=&\#\{(i,j)\ |\  i<j, |\sigma(i)|>|\sigma(j)|\} \\
&&\quad+2\#\{(i,j) \ |\ i<j, 0<-\sigma(i)<|\sigma(j)|\}
+\#\{i \ |\  \sigma(i)<0 \}.
\end{eqnarray*}
We define the inversion of an inverse pre-order word $\sigma$ as 
\begin{eqnarray*}
\mathrm{inv}(\sigma)
&:=&
\#\{(i,j) \ |\  i<j, \sigma(i)<\sigma(j), \text{ and } 
i,j \text{ are not underlined} \} \\
&&+2\#\{(i,j) \ |\  i<j, \sigma(i)>\sigma(j), \text{ and }
i \text{ or } j \text{ is underlined} \}.
\end{eqnarray*}

\begin{example}
\label{example:L}
A natural labelling $L$ associated to a path $UUDDUUUUDDUU$ is 
\begin{eqnarray*}
L=
\scalebox{0.7}{
\tikzpic{-0.5}{
\coordinate
	child{coordinate (c0)   
	child{coordinate (c01)}
	child[missing]
	}
	child {coordinate (c2)}{
        child {coordinate (c3) child {coordinate (c4)} child[missing]}
	child {coordinate (c5)}
	};
\draw[very thick](c2)node{$-$} (c3)node{$-$} (c4)node{$-$} (c5)node{$-$}(c0)node{$-$}(c01)node{$-$};
\node at ($(c2)!.5!(c5)$){\scalebox{1.5}{$\bullet$}};
\node at ($(0,0)!.5!(c2)$){\scalebox{1.5}{$\bullet$}};
\node[fill=white,anchor=south east] at ($(0,0)!.6!(c0)$){\scalebox{1.425}{1}};
\node[fill=white,anchor=south east] at ($(c0)!.6!(c01)$){\scalebox{1.425}{3}};
\node[fill=white,anchor=south west] at ($(0,0)!.6!(c2)$){\scalebox{1.425}{2}};
\node[fill=white,anchor=south east] at ($(c2)!.6!(c3)$){\scalebox{1.425}{4}};
\node[fill=white,anchor=south east] at ($(c3)!.6!(c4)$){\scalebox{1.425}{5}};
\node[fill=white,anchor=south west] at ($(c2)!.6!(c5)$){\scalebox{1.425}{6}};
}
}.
\end{eqnarray*}
Then, we have $\mathrm{pre}(L)=13\underline{2}45\underline{6}$ 
and $\mathrm{post}(L)=31\underline{2}54\underline{6}$.
\end{example}

\subsection{Generalized perfect matching}
\label{sec:gpm}
Given a ballot path $\gamma$,  let $\gamma_{0}$ be the lowest ballot paths such that  
the skew shape $\gamma_{0}/\gamma$ exists and $|\gamma|_{X}=|\gamma_0|_{X}$ 
with $X\in\{U,D\}$.
The path $\gamma_0$ is a concatenation of a zig-zag path and a path consisting 
of only $U$'s.

A generalized {\it perfect matching} is a diagram consisting 
of arcs, dashed arcs and a diamond ($\diamond$).
We consider two types of generalized perfect matching: type I and type II.

A generalized perfect matching on a path $\gamma$ of type I 
is defined as follows.
An arc connects a $U$ and a $D$ in $\gamma$, a dashed 
arc connects two $U$'s in $\gamma$ and a diamond is on a $U$ in $\gamma$.
Given $\gamma$, we start with making pairs of $U$ and $D$ by connecting them 
into an arc. 
Then, we make a pairs of two $U$'s next to each other by connecting them 
into a dashed arc.
We put a diamond $\diamond$ in a $U$ of a remaining $U$ (if it exists).
We denote by $\mathrm{PM}_{I}(\gamma)$ the set of generalized perfect matchings 
of type I for a path $\gamma$. 
We consider two conditions:
\begin{enumerate}
\item[(C1)]
Two configurations of two dashed arcs  
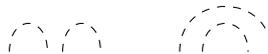
\begin{eqnarray*}
\tikzpic{0}{
\draw[dashed](0,0)..controls(0,0.5)and(0.5,0.5)..(0.5,0);
\draw[dashed](0.7,0)..controls(0.7,0.5)and(1.2,0.5)..(1.2,0);
}\qquad
\tikzpic{0}{
\draw[dashed](0,0)..controls(0,0.8)and(1.2,0.8)..(1.2,0);
\draw[dashed](0.3,0)..controls(0.3,0.5)and(0.9,0.5)..(0.9,0);
}
\end{eqnarray*}
are not admissible. 
\item[(C2)]There exists no $U$ right to the diamond ($\diamond$) such 
that it forms a dashed arc.
\end{enumerate}
We denote by $\mathrm{PM}'_{I}(\gamma)\subseteq\mathrm{PM}_{I}(\gamma)$ the set of generalized 
perfect matchings of $\gamma$ satisfying (C1) and (C2).

We define a generalized perfect matching of type II as follows.
Let $\lambda$ be a ballot path and $L$ be a natural labelling of 
the tree $A(\lambda)$ and $\pi=\mathrm{pre}(L)$. 
Here, we consider only a plane tree without ``arrows".
We associate a sequence of (underlined) integers $\mathbf{p}':=(p'_1,\ldots,p'_E)$ 
to $L$ which has $E$ edges. 
Given an edge $e$ with an integer $n_{e}$, there is a unique sequence 
of edges from $e$ to the root. 
We denote by $S_{e}$ the set of such edges.
We denote by $T_{e}$ the set of edges such that an integer $n_{e'}$
on the edge $e'\in T_{e}$ satisfies $n_{e'}< n_{e}$ and $n_{e'}$ appears 
to the left to $n_e$ in $\pi$.
We define $\mathbf{p}'$ as 
\begin{eqnarray*}
p'_{n_e}:=2|T_{e}|-|S_{e}|+1,
\end{eqnarray*}
where $1\le n_e\le E$.
If an edge $e$ in $A(\lambda)$ has a dot, we put an underline on $p'_{n_e}$.

We define a {\it mini-word} $w$ from $\mathbf{p}'$ as follows.
Let $S_{E}:=[1,2E]$. 
For $1\le i\le E$, we make a pair between $(p'_{i}+1)$-th smallest 
element $a'_{i}$ and $b'_{i}:=\max(S_{i})$ in $S_{i}$ and define 
$S_{i-1}:=S_{i}\setminus\{a'_{i},b'_{i}\}$.
When $p'_{i}$ is underlined, we put an underline on $a'_{i}$.
Then, we obtain a sequence of (underlined) integers $\mathbf{a}':=(a'_{1},\ldots,a'_{E})$.
Let $e_{i}$ be an edge with the integer $i$, and 
$m_{i}$, $1\le i\le E$, be 
\begin{eqnarray*}
m_{i}:=\#\{a'_{j}\ |\  a'_{j}<a'_{i}, \text{$a'_{j}$ is underlined, and } e_{j}\in S_{e_{i}}\}.
\end{eqnarray*}
We define $\mathbf{a}:=(a_{1},\ldots,a_{E})=(a'_{1}+m_1,a'_{2}+m_2,\ldots,a'_{E}+m_{E})$ 
and $\mathbf{p}:=(p_1,\ldots,p_{E})=(p'_{1}+m_1,p'_{2}+m_2,\ldots,p'_{E}+m_{E})$.
A sequence of integers $\mathbf{b}:=(b_1,\ldots,b_{E})$ is defined 
as an increasing sequence in $S_{2E}\setminus\{a_{i}\ |\ 1\le i\le E\}$.
Then, $w$ is a $2\times E$ array of integer where the first row 
is $\mathbf{a}$ and the second row is $\mathbf{b}$.
By construction, $b_1,\ldots,b_{E}$ is an increasing sequence.
Further, when $a_{i}$ is underlined, we put an underline on $b_{i}$.
We have a generalized perfect matching from a mini-word $w$ in the following way.
We consider a graph with $2E$ nodes which lie in a line from left to right.
A pair $(a_i,b_i)$ without underlines corresponds to an arc from the $a_i$-th 
node and the $b_i$-th node. If a pair $(a_j,b_j)$ is underlined, then we connect 
the $a_j$-th node and $b_{j}$-th node by a dashed arc.
We call a generalized perfect matching obtained in this way type II.
We denote by $\mathrm{PM}_{II}(\lambda)$ the set of generalized perfect matchings 
of type II corresponding to pre-order words of $A(\lambda)$.

\begin{example}
We consider the same path and the same natural labelling as in Example \ref{example:L}.
We have $\mathbf{p}'=(0,\underline{2},1,5,6,\underline{9})$ and 
$\mathbf{a}'=(1,\underline{4},2,6,7,\underline{10})$.
Thus we have $\mathbf{p}=(0,\underline{2},1,6,7,\underline{10})$ and 
$w=
\begin{array}{cccccc}
1 & \underline{4} & 2 & 7 & 8 & \underline{11} \\
3 & \underline{5} & 6 & 9 & 10 & \underline{12}
\end{array}
$.
The generalized perfect matching corresponding to $w$ 
is depicted as 
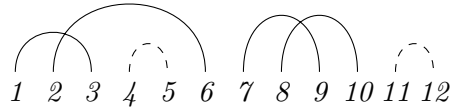
\begin{eqnarray*}
\tikzpic{-0.5}{
\draw(0,0)..controls(0,0.7)and(1,0.7)..(1,0)
     (0.5,0)..controls(0.5,1.2)and(2.5,1.2)..(2.5,0)
     (3,0)..controls(3,1)and(4,1)..(4,0)
     (3.5,0)..controls(3.5,1)and(4.5,1)..(4.5,0);
\draw[dashed](1.5,0)..controls(1.5,0.5)and(2,0.5)..(2,0) 
            (5,0)..controls(5,0.5)and(5.5,0.5)..(5.5,0);
\foreach \x in {1,2,...,12} {
	\draw(0.5*\x-0.5,0)node[anchor=north]{\x};
	}
}
\end{eqnarray*}
\end{example}

Below, we define two statistics {\it crossing} and {\it nesting}
for a generalized perfect matching in $\mathrm{PM}'_{I}(\mu)$ and 
$\mathrm{PM}_{I}(\mu)$, and {\it nesting} for $\mathrm{PM}_{II}(\mu)$.
If the length of a generalized perfect matching is $l$,  
crossing and nesting are sequences of integers of length $l$.

We start with the definitions of two statistics {\it crossing} and 
{\it nesting} for a generalized perfect matching in $\mathrm{PM}'_{I}(\mu)$.
Suppose that the length of $\pi\in\mathrm{PM}'_{I}(\mu)$ is $l$.
We assign sequences of integers $\mathbf{p}_{\mathrm{cr}}$ and 
$\mathbf{p}_{\mathrm{nes}}$ of length $l$ to $\pi$.

For crossing, we define a partial integer sequence for a configuration of 
arcs, dashed arcs and the diamond as follows.
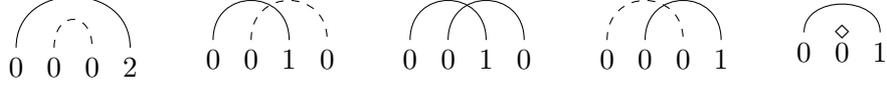
\begin{eqnarray*}
\tikzpic{-0.5}{
\draw(0,0)..controls(0,0.9)and(1.5,0.9)..(1.5,0);
\draw[dashed](0.5,0)..controls(0.5,0.5)and(1,0.5)..(1,0);
\draw(0,0)node[anchor=north]{$0$}(0.5,0)node[anchor=north]{$0$}
     (1,0)node[anchor=north]{$0$}(1.5,0)node[anchor=north]{$2$};
}
\quad
\tikzpic{-0.5}{
\draw(0,0)..controls(0,0.7)and(1,0.7)..(1,0);
\draw[dashed](0.5,0)..controls(0.5,0.7)and(1.5,0.7)..(1.5,0);
\draw(0,0)node[anchor=north]{$0$}(0.5,0)node[anchor=north]{$0$}
     (1,0)node[anchor=north]{$1$}(1.5,0)node[anchor=north]{$0$};
}
\quad
\tikzpic{-0.5}{
\draw(0,0)..controls(0,0.7)and(1,0.7)..(1,0);
\draw(0.5,0)..controls(0.5,0.7)and(1.5,0.7)..(1.5,0);
\draw(0,0)node[anchor=north]{$0$}(0.5,0)node[anchor=north]{$0$}
     (1,0)node[anchor=north]{$1$}(1.5,0)node[anchor=north]{$0$};
}
\quad
\tikzpic{-0.5}{
\draw[dashed](0,0)..controls(0,0.7)and(1,0.7)..(1,0);
\draw(0.5,0)..controls(0.5,0.7)and(1.5,0.7)..(1.5,0);
\draw(0,0)node[anchor=north]{$0$}(0.5,0)node[anchor=north]{$0$}
     (1,0)node[anchor=north]{$0$}(1.5,0)node[anchor=north]{$1$};
}\quad
\tikzpic{-0.5}{
\draw(0,0)..controls(0,0.5)and(1,0.5)..(1,0);
\draw(0.5,0)node{$\diamond$};
\draw(0.5,0)node[anchor=north]{$0$}
     (0,0)node[anchor=north]{$0$}(1,0)node[anchor=north]{$1$};
}
\end{eqnarray*}
Partial integer sequences for other configurations are zero.
Then, a sequence $\mathbf{p}_{\mathrm{cr}}$ is defined by 
the sum of partial sequences corresponding to diagrams depicted above.
We define $\mathrm{cr}(\pi):=\sum_{i=1}^{l}p_i$ for 
$\mathbf{p}_\mathrm{cr}=p_1\ldots p_l$.

For nesting, we define a partial integer sequence for a configuration
of arcs, dashed arcs and the diamond as follows.
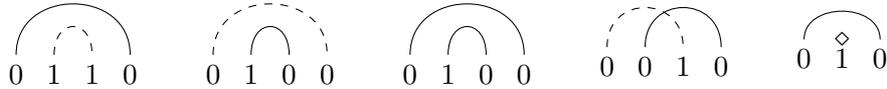
\begin{eqnarray*}
\tikzpic{-0.5}{
\draw(0,0)..controls(0,0.9)and(1.5,0.9)..(1.5,0);
\draw[dashed](0.5,0)..controls(0.5,0.5)and(1,0.5)..(1,0);
\draw(0,0)node[anchor=north]{$0$}(0.5,0)node[anchor=north]{$1$}
     (1,0)node[anchor=north]{$1$}(1.5,0)node[anchor=north]{$0$};
}
\quad
\tikzpic{-0.5}{
\draw[dashed](0,0)..controls(0,0.9)and(1.5,0.9)..(1.5,0);
\draw(0.5,0)..controls(0.5,0.5)and(1,0.5)..(1,0);
\draw(0,0)node[anchor=north]{$0$}(0.5,0)node[anchor=north]{$1$}
     (1,0)node[anchor=north]{$0$}(1.5,0)node[anchor=north]{$0$};
}
\quad
\tikzpic{-0.5}{
\draw(0,0)..controls(0,0.9)and(1.5,0.9)..(1.5,0);
\draw(0.5,0)..controls(0.5,0.5)and(1,0.5)..(1,0);
\draw(0,0)node[anchor=north]{$0$}(0.5,0)node[anchor=north]{$1$}
     (1,0)node[anchor=north]{$0$}(1.5,0)node[anchor=north]{$0$};
}
\quad
\tikzpic{-0.5}{
\draw[dashed](0,0)..controls(0,0.7)and(1,0.7)..(1,0);
\draw(0.5,0)..controls(0.5,0.7)and(1.5,0.7)..(1.5,0);
\draw(0,0)node[anchor=north]{$0$}(0.5,0)node[anchor=north]{$0$}
     (1,0)node[anchor=north]{$1$}(1.5,0)node[anchor=north]{$0$};
}
\quad
\tikzpic{-0.5}{
\draw(0,0)..controls(0,0.5)and(1,0.5)..(1,0);
\draw(0.5,0)node{$\diamond$};
\draw(0.5,0)node[anchor=north]{$1$}
     (0,0)node[anchor=north]{$0$}(1,0)node[anchor=north]{$0$};
}
\end{eqnarray*}
Partial integer sequences for other configurations are zero.
The sequence $\mathbf{p}_{\mathrm{nes}}$ is defined by adding the partial 
sequences corresponding to diagrams depicted above.
We define $\mathrm{ne(\pi)}:=\sum_{i=1}^{l}p_i$ for 
$\mathbf{p}_{\mathrm{nes}}=p_1\ldots p_l$.

\begin{example}
\label{ex:GLP}
Fix a path $\mu:=UUDUUUDDUUUUD$. 
Let $\pi_1$ be a generalized link pattern in $\mathrm{PM}'_{I}(\mu)$ depicted 
as below: 
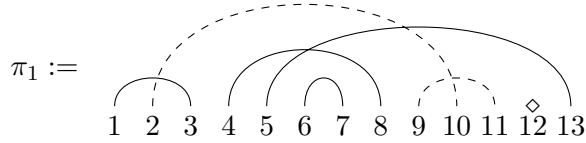
\begin{eqnarray*}
\pi_1:=
\tikzpic{-0.4}{
\draw(0,0)..controls(0,0.5)and(1,0.5)..(1,0)
     (1.5,0)..controls(1.5,1)and(3.5,1)..(3.5,0)
     (2,0)..controls(2,1.4)and(6,1.4)..(6,0) 
     (2.5,0)..controls(2.5,0.5)and(3,0.5)..(3,0);
\draw[dashed](0.5,0)..controls(0.5,1.8)and(4.5,1.8)..(4.5,0)
             (4,0)..controls(4,0.5)and(5,0.5)..(5,0);
\draw(5.5,0)node{$\diamond$};
\foreach \x/\y in{0/1,0.5/2,1/3,1.5/4,2/5,2.5/6,3/7,3.5/8,4/9,4.5/10,5/11,5.5/12,6/13}{
	\draw(\x,0)node[anchor=north]{$\y$};
			 }
}
\end{eqnarray*}
The crossing of $\pi_{1}$ is $\mathbf{p}_{\mathrm{cr}}(\pi_1):=(p_{3},p_{7},p_{8},p_{13})=(1,0,1,4)$ and 
other $p_i$'s are zero.
The nesting of $\pi_1$ is 
$\mathbf{p}_{\mathrm{nes}}(\pi_{1}):=(p_{1},p_{2},p_{4},p_{5},p_{6},p_{9},p_{10},p_{11},p_{12})
=(0,0,1,0,3,1,1,1,1)$ and other $p_{i}$'s are zero.
\end{example}

For a perfect matching in $\mathrm{PM}_{I}(\mu)$, 
we define a partial sequence of integers for a configuration
of arcs, dashed arcs and the diamond as follows:
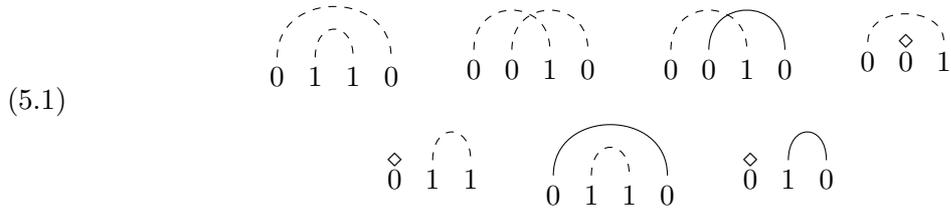
\begin{eqnarray}
\label{weight-PMI}
\begin{gathered}
\tikzpic{-0.5}{
\draw[dashed](0,0)..controls(0,0.9)and(1.5,0.9)..(1.5,0);
\draw[dashed](0.5,0)..controls(0.5,0.5)and(1,0.5)..(1,0);
\draw(0,0)node[anchor=north]{$0$}(0.5,0)node[anchor=north]{$1$}
     (1,0)node[anchor=north]{$1$}(1.5,0)node[anchor=north]{$0$};
}
\quad
\tikzpic{-0.5}{
\draw[dashed](0,0)..controls(0,0.7)and(1,0.7)..(1,0);
\draw[dashed](0.5,0)..controls(0.5,0.7)and(1.5,0.7)..(1.5,0);
\draw(0,0)node[anchor=north]{$0$}(0.5,0)node[anchor=north]{$0$}
     (1,0)node[anchor=north]{$1$}(1.5,0)node[anchor=north]{$0$};
}
\quad
\tikzpic{-0.5}{
\draw[dashed](0,0)..controls(0,0.7)and(1,0.7)..(1,0);
\draw(0.5,0)..controls(0.5,0.7)and(1.5,0.7)..(1.5,0);
\draw(0,0)node[anchor=north]{$0$}(0.5,0)node[anchor=north]{$0$}
     (1,0)node[anchor=north]{$1$}(1.5,0)node[anchor=north]{$0$};
}
\quad
\tikzpic{-0.5}{
\draw[dashed](0,0)..controls(0,0.5)and(1,0.5)..(1,0);
\draw(0.5,0)node{$\diamond$};
\draw(0.5,0)node[anchor=north]{$0$}
     (0,0)node[anchor=north]{$0$}(1,0)node[anchor=north]{$1$};
} \\
\tikzpic{-0.5}{
\draw[dashed](0.5,0)..controls(0.5,0.5)and(1,0.5)..(1,0);
\draw(0,0)node{$\diamond$};
\draw(0.5,0)node[anchor=north]{$1$}
     (0,0)node[anchor=north]{$0$}(1,0)node[anchor=north]{$1$};
}
\quad
\tikzpic{-0.5}{
\draw(0,0)..controls(0,0.9)and(1.5,0.9)..(1.5,0);
\draw[dashed](0.5,0)..controls(0.5,0.5)and(1,0.5)..(1,0);
\draw(0,0)node[anchor=north]{$0$}(0.5,0)node[anchor=north]{$1$}
     (1,0)node[anchor=north]{$1$}(1.5,0)node[anchor=north]{$0$};
}
\quad
\tikzpic{-0.5}{
\draw(0.5,0)..controls(0.5,0.5)and(1,0.5)..(1,0);
\draw(0,0)node{$\diamond$};
\draw(0.5,0)node[anchor=north]{$1$}
     (0,0)node[anchor=north]{$0$}(1,0)node[anchor=north]{$0$};
}
\end{gathered}
\end{eqnarray}
For nesting, we define a partial sequence of integers in addition to Eqn.(\ref{weight-PMI}):
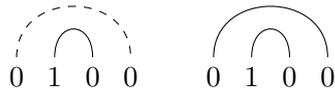
\begin{eqnarray*}
\tikzpic{-0.5}{
\draw[dashed](0,0)..controls(0,0.9)and(1.5,0.9)..(1.5,0);
\draw(0.5,0)..controls(0.5,0.5)and(1,0.5)..(1,0);
\draw(0,0)node[anchor=north]{$0$}(0.5,0)node[anchor=north]{$1$}
     (1,0)node[anchor=north]{$0$}(1.5,0)node[anchor=north]{$0$};
}
\quad
\tikzpic{-0.5}{
\draw(0,0)..controls(0,0.9)and(1.5,0.9)..(1.5,0);
\draw(0.5,0)..controls(0.5,0.5)and(1,0.5)..(1,0);
\draw(0,0)node[anchor=north]{$0$}(0.5,0)node[anchor=north]{$1$}
     (1,0)node[anchor=north]{$0$}(1.5,0)node[anchor=north]{$0$};
}
\end{eqnarray*} 
A partial sequence of integers is zero for other configurations.
Similarly for crossing, we define a partial sequence of integers 
in addition to Eqn.(\ref{weight-PMI}): 
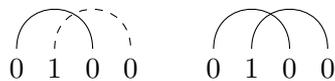
\begin{eqnarray*}
\tikzpic{-0.5}{
\draw(0,0)..controls(0,0.7)and(1,0.7)..(1,0);
\draw[dashed](0.5,0)..controls(0.5,0.7)and(1.5,0.7)..(1.5,0);
\draw(0,0)node[anchor=north]{$0$}(0.5,0)node[anchor=north]{$1$}
     (1,0)node[anchor=north]{$0$}(1.5,0)node[anchor=north]{$0$};
}
\quad\tikzpic{-0.5}{
\draw(0,0)..controls(0,0.7)and(1,0.7)..(1,0);
\draw(0.5,0)..controls(0.5,0.7)and(1.5,0.7)..(1.5,0);
\draw(0,0)node[anchor=north]{$0$}(0.5,0)node[anchor=north]{$1$}
     (1,0)node[anchor=north]{$0$}(1.5,0)node[anchor=north]{$0$};
}
\end{eqnarray*}
A partial sequence of integers is zero for other configurations.
Sequences $\mathbf{p}_{\mathrm{nes}}$ and $\mathbf{p}_{\mathrm{cr}}$ 
are defined as a sum of partial sequences corresponding to diagrams 
depicted above. 
The quantities $\mathrm{nes}(\pi)$ and $\mathrm{cr}(\pi)$ is defined 
as the sum of $p_{i}$ with $1\le i\le l$.

For nesting of a generalized perfect matching in $\mathrm{PM}_{II}(\lambda)$, 
we define a partial sequence of integers for a configuration of 
arcs, dashed arcs (no diamond since we consider a plane tree without arrows) 
as follows:
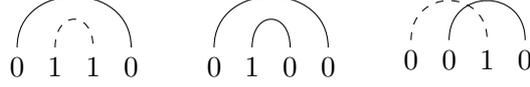
\begin{eqnarray*}
\tikzpic{-0.5}{
\draw(0,0)..controls(0,0.9)and(1.5,0.9)..(1.5,0);
\draw[dashed](0.5,0)..controls(0.5,0.5)and(1,0.5)..(1,0);
\draw(0,0)node[anchor=north]{$0$}(0.5,0)node[anchor=north]{$1$}
     (1,0)node[anchor=north]{$1$}(1.5,0)node[anchor=north]{$0$};
}
\quad
\tikzpic{-0.5}{
\draw(0,0)..controls(0,0.9)and(1.5,0.9)..(1.5,0);
\draw(0.5,0)..controls(0.5,0.5)and(1,0.5)..(1,0);
\draw(0,0)node[anchor=north]{$0$}(0.5,0)node[anchor=north]{$1$}
     (1,0)node[anchor=north]{$0$}(1.5,0)node[anchor=north]{$0$};
}
\quad
\tikzpic{-0.5}{
\draw[dashed](0,0)..controls(0,0.7)and(1,0.7)..(1,0);
\draw(0.5,0)..controls(0.5,0.7)and(1.5,0.7)..(1.5,0);
\draw(0,0)node[anchor=north]{$0$}(0.5,0)node[anchor=north]{$0$}
     (1,0)node[anchor=north]{$1$}(1.5,0)node[anchor=north]{$0$};
}
\end{eqnarray*}
A integer sequence $\mathbf{p}_{\mathrm{nes}}$ is defined as 
a sum of partial sequence corresponding to nesting diagrams 
depicted above. We define $\mathrm{nes}(\pi)$ as the sum of $p_i$.

\subsection{Ballot tiling strip}
\label{sec:BTS}
We give a bijection called {\it ballot tiling strip} (BTS for short)
from a natural labelling $L$ of $A(\lambda)$ to a ballot tiling.
We consider the case where the tree $A(\lambda)$ has no arrows.

We define two operations {\it Dyck spread} and {\it ballot spread}
on a ballot path $\rho$.
We first define the Dyck spread following \cite[Section 2]{KMPW12}.
Given a ballot path $\rho$ and a column $s$, the Dyck spread of 
$\rho$ at $s$ is the ballot path $\rho'$ consisting  
of the following points:
\begin{eqnarray*}
\{(x-1,y)\ |\ (x,y)\in\rho, x\le s\}\cup
\{(x,y+1)\ |\ (x,y)\in\rho, x=s\}\cup
\{(x+1,y)\ |\ (x,y)\in\rho, x\ge s\}.
\end{eqnarray*}
Similarly, the ballot spread of $\rho$ at $s$ is the ballot path 
$\rho'$ consisting of the following points:
\begin{eqnarray*}
\{(x-1,0)\ |\ (x,y)\in\rho, x\le s\}\cup
\{(s,y+1)\ |\ (s,y)\in\rho,\}\cup\{(x+1,y+2)\ |\  (x,y)\in\rho, x\ge s\}.
\end{eqnarray*}

We define the spread of a ballot tile at column $s$ by spreading 
the ballot path which characterize the ballot tile.
Then, the spread of a ballot tiling at column $s$ by spreading 
the upper path, the lower path and ballot tiles in-between them.  

We define a growth process of a ballot tiling $T$ which we call 
the {\it strip-grow}.
Given a ballot tiling $T$ of length $N=2n+n'$  
and a column $s$ with $0\le s\le N$, 
we define the Dyck strip grow $\mathrm{DSG}(T,s)$ 
(resp. ballot strip grow $\mathrm{BSG}(T,s)$) to be a ballot 
tiling $T'$ obtained as follows.
Suppose that we obtain a tiling $T''$ by (Dyck or ballot) spreading 
the tiling $T$ at $s$. Let $\mu$ be the upper path of $T''$.
We put one-box tiles on the up steps of $\mu$ which is right to 
the column $s$.
The tiling $T'$ obtained by the above procedure is denoted by 
$\mathrm{DSG}(T,s)$ (resp. $\mathrm{BSG}(T,s)$).

Let $L$ be a natural labelling of a tree $A(\lambda)$.
Recall that one constructs $\mathbf{p}=(p_1,\ldots,p_E)$ from $L$.
The map $\mathrm{BTS}$ from integer sequences $\mathbf{p}$ to 
cover-inclusive ballot tilings is given by successive applications 
of the strip-grow.
We define 
\begin{eqnarray*}
BTS(\mathbf{p}):=
\begin{cases}
\emptyset & E=0, \\
BSG(BTS(p_1,\ldots,p_{n-1}),p_{n}-(n-1)) & n>0, \  \text{$p_{E}$ is underlined}, \\
DSG(BTS(p_1,\ldots,p_{n-1}),p_{n}-(n-1)) & otherwise. 
\end{cases}
\end{eqnarray*}

\begin{example}
Let $L$ be a natural labelling of the following tree:
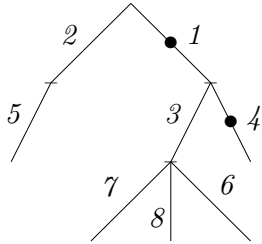
\begin{eqnarray*}
\scalebox{0.7}{
\tikzpic{-0.5}{
\coordinate
	child{coordinate (c0)   
		child{coordinate (c01)}
		child[missing]
	}
	child[missing]
	child {coordinate (c2)}{
        	child {coordinate (c3) 
			child {coordinate (c4)} child {coordinate (c5)} 
			child {coordinate (c6)}}
		child {coordinate (c7)}
	};
\draw(c0)node{$-$}(c2)node{$-$}(c3)node{$-$};
\node at ($(0,0)!.5!(c2)$){\scalebox{1.5}{$\bullet$}};
\node at ($(c2)!.5!(c7)$){\scalebox{1.5}{$\bullet$}};
\node[fill=white,anchor=south east] at ($(0,0)!.6!(c0)$){\scalebox{1.425}{2}};
\node[fill=white,anchor=south east] at ($(c0)!.6!(c01)$){\scalebox{1.425}{5}};

\node[fill=white,anchor=south west] at ($(0,0)!.6!(c2)$){\scalebox{1.425}{1}};
\node[fill=white,anchor=south east] at ($(c2)!.6!(c3)$){\scalebox{1.425}{3}};
\node[fill=white,anchor=south east] at ($(c3)!.6!(c4)$){\scalebox{1.425}{7}};
\node[fill=white,anchor=east] at ($(c3)!.7!(c5)$){\scalebox{1.425}{8}};
\node[fill=white,anchor=south west] at ($(c3)!.5!(c6)$){\scalebox{1.425}{6}};
\node[fill=white,anchor=south west] at ($(c2)!.7!(c7)$){\scalebox{1.425}{4}};
}
}
\end{eqnarray*}
The growth of ballot tilings are 
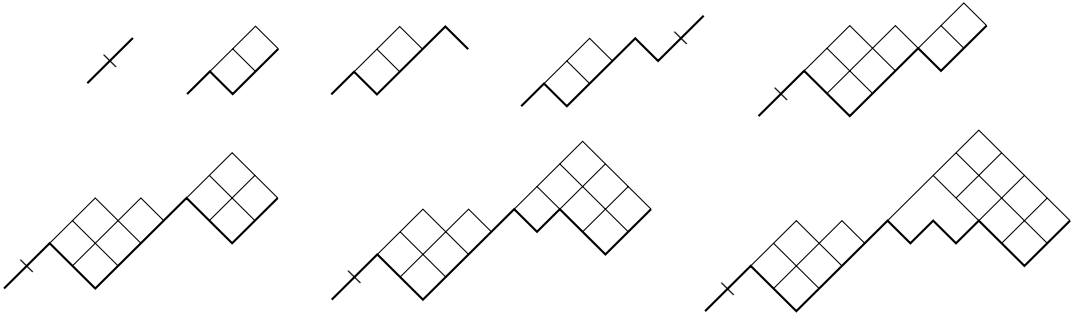
\begin{eqnarray*}
\begin{gathered}
\nonumber
\tikzpic{-0.5}{
\draw[thick](0,0)--(0.6,0.6);
\draw(0.3,0.3)node{\rotatebox{-45}{$-$}};
}\quad
\tikzpic{-0.5}{
\draw[thick](0,0)--(0.3,0.3)--(0.6,0)--(1.2,0.6);
\draw(0.3,0.3)--(0.9,0.9)--(1.2,0.6)(0.6,0.6)--(0.9,0.3);
}\quad
\tikzpic{-0.5}{
\draw[thick](0,0)--(0.3,0.3)--(0.6,0)--(1.5,0.9)--(1.8,0.6);
\draw(0.3,0.3)--(0.9,0.9)--(1.2,0.6)(0.6,0.6)--(0.9,0.3);
}\quad
\tikzpic{-0.5}{
\draw[thick](0,0)--(0.3,0.3)--(0.6,0)--(1.5,0.9)--(1.8,0.6)--(2.4,1.2);
\draw(2.1,0.9)node{\rotatebox{-45}{$-$}};
\draw(0.3,0.3)--(0.9,0.9)--(1.2,0.6)(0.6,0.6)--(0.9,0.3);
}\quad
\tikzpic{-0.5}{
\draw[thick](-0.6,0)--(0,0.6)--(0.6,0)--(1.5,0.9)--(1.8,0.6)--(2.4,1.2);
\draw(-0.3,0.3)node{\rotatebox{-45}{$-$}};
\draw(0.3,0.3)--(0.9,0.9)--(1.2,0.6)(0.6,0.6)--(0.9,0.3);
\draw(0,0.6)--(0.6,1.2)--(0.9,0.9)(0.3,0.9)--(0.6,0.6);
\draw(0.9,0.9)--(1.2,1.2)--(1.5,0.9)--(2.1,1.5)--(2.4,1.2)(1.8,1.2)--(2.1,0.9);
} \\
\tikzpic{-0.5}{
\draw[thick](-0.6,0)--(0,0.6)--(0.6,0)--(1.5,0.9)--(1.8,1.2)--(2.4,0.6)--(3,1.2);
\draw(-0.3,0.3)node{\rotatebox{-45}{$-$}};
\draw(0.3,0.3)--(0.9,0.9)--(1.2,0.6)(0.6,0.6)--(0.9,0.3);
\draw(0,0.6)--(0.6,1.2)--(0.9,0.9)(0.3,0.9)--(0.6,0.6);
\draw(0.9,0.9)--(1.2,1.2)--(1.5,0.9);
\draw(1.8,1.2)--(2.4,1.8)--(3,1.2)(2.1,1.5)--(2.7,0.9)(2.1,0.9)--(2.7,1.5);
}\quad
\tikzpic{-0.5}{
\draw[thick](-0.6,0)--(0,0.6)--(0.6,0)--(1.5,0.9)--(1.8,1.2)--
	(2.1,0.9)--(2.4,1.2)--(3,0.6)--(3.6,1.2);
\draw(-0.3,0.3)node{\rotatebox{-45}{$-$}};
\draw(0.3,0.3)--(0.9,0.9)--(1.2,0.6)(0.6,0.6)--(0.9,0.3);
\draw(0,0.6)--(0.6,1.2)--(0.9,0.9)(0.3,0.9)--(0.6,0.6);
\draw(0.9,0.9)--(1.2,1.2)--(1.5,0.9);
\draw(2.4,1.2)--(3,1.8)--(3.6,1.2)(2.7,1.5)--(3.3,0.9)(2.7,0.9)--(3.3,1.5);
\draw(1.8,1.2)--(2.7,2.1)--(3,1.8)(2.1,1.5)--(2.4,1.2)(2.4,1.8)--(2.7,1.5);
}
\quad
\tikzpic{-0.5}{
\draw[thick](-0.6,0)--(0,0.6)--(0.6,0)--(1.5,0.9)--(1.8,1.2)--
	(2.1,0.9)--(2.4,1.2)--(2.7,0.9)--(3,1.2)--(3.6,0.6)--(4.2,1.2);
\draw(-0.3,0.3)node{\rotatebox{-45}{$-$}};
\draw(0.3,0.3)--(0.9,0.9)--(1.2,0.6)(0.6,0.6)--(0.9,0.3);
\draw(0,0.6)--(0.6,1.2)--(0.9,0.9)(0.3,0.9)--(0.6,0.6);
\draw(0.9,0.9)--(1.2,1.2)--(1.5,0.9);
\draw(3,1.2)--(3.6,1.8)--(4.2,1.2)(3.3,1.5)--(3.9,0.9)(3.3,0.9)--(3.9,1.5);
\draw(2.7,1.5)--(3.3,2.1)--(3.6,1.8)(2.7,1.5)--(3,1.2)(3,1.8)--(3.3,1.5);
\draw(1.8,1.2)--(2.4,1.8)--(2.7,1.5);
\draw(2.4,1.8)--(3,2.4)--(3.3,2.1)(2.7,2.1)--(3,1.8);
}
\end{gathered}
\end{eqnarray*}
\end{example}

Given a ballot path $\lambda$, a sequence of integers $\mathbf{p}$ 
is bijective to an inverse pre-order word $\sigma$.
Thus we define 
\begin{eqnarray*}
\mathrm{BTS}(\lambda,\sigma):=\mathrm{BTS}(\mathbf{p}),
\end{eqnarray*}
where $\mathbf{p}$ is represented by a pair $(\lambda,\sigma)$.
\begin{theorem}
\label{thrm:BTSinpo} 
For a ballot path $\lambda$ and an inverse pre-order word $\sigma$
associated to a natural labelling $L$ of the tree $A(\lambda)$, 
we have 
\begin{eqnarray*}
\mathrm{art}(\mathrm{BTS}(\lambda,\sigma))=\mathrm{inv}(\sigma).
\end{eqnarray*}
\end{theorem}

\subsection{Hermite history for a ballot tiling}
\label{sec:HhBT}
We consider a ballot tiling in the shape $\lambda/\mu$.
We define an Hermite history for a ballot tiling by replacing 
Dyck with ballot in the definitions in the first three paragraphs 
of Section \ref{sec:Hh}.
Though we have three types of Hermite histories, 
we consider only two types of Hermite histories for a generalized 
perfect matching in $\mathrm{PM}_{I}(\mu)$: type I and type III.
For a generalized perfect matching $\pi\in\mathrm{PM}_{II}(\lambda)$,
we consider only Hermite histories of type I. 

The main difference between Hermite history for a Dyck tiling and 
the one for a ballot tiling is that 
the statistics tiles $\mathrm{tiles}(T)$ is zero for a ballot 
tile of length $(2n,n')$ with $n'\ge1$.

Let $\pi$ be a generalized perfect matching in $\mathrm{PM}'_{I}(\mu)$.
By definition, a non-zero element in the integer sequence 
$\mathbf{p}_{\mathrm{nes}}(\pi)$ is on an up step of $\mu$.
Thus a labelling $H$ can be obtained by deleting zeros, which are 
on the down steps of $\mu$, from $\mathbf{p}_{\mathrm{nes}}$.  
The lower path $\lambda$ can be obtained as follows.
Recall that a path $\mu$ consists of $U$'s and $D$'s.
We put an integer sequence $\mathbf{p}_{\mathrm{nes}}$ on 
the path $\mu$.
We enumerate the up steps in $\mu$ from right to left by 
$1,2,..,N_{U}$ where $N_{U}$ is the number of up steps in $\mu$.
We perform the following operation starting from the first up step and 
ending with the $N_{U}$-th up step.
Suppose that the $i$-th up step has an integer $n_i$. 
Let $\mu_i$ be a partial path in $\mu$ which is right to the $i$-th up step.
We perform the operations (A) and (B) (See Section \ref{subsec:im}) on $\mu_{i}$.
We have paired $UD$'s and unpaired $D$'s and $U$'s.
We enumerate the unpaired $D$'s and $U$'s from left to right by $1,2,\ldots$ and 
let $m_{j}$ be the position of the $j$-th $U$ or $D$ in $\mu$.
We denote by $m_0$ the position of the $i$-th up step in $\mu$.
For $1\le j\le n_{i}$, we move the $j$-th unpaired $U$ or $D$ to the position 
$m_{j-1}$ and put an up step at the position $m_{n_{i}}$.
We obtain the lower path $\lambda$ in this way.

\begin{example}
For a ballot path $\mu=UUUUDDUDD$ and $\mathbf{p}_{\mathrm{nes}}=(0,0,2,0,0)$, 
the lower path $\lambda$ is obtained as 
\begin{eqnarray*}
\tikzpic{-0.5}{
\draw(0,0)node[anchor=north]{$U$}node[anchor=south]{$0$};
\draw(0.4,0)node[anchor=north]{$U$}node[anchor=south]{$0$};
\draw(0.8,0)node[anchor=north]{$U$}node[anchor=south]{$2$};
\draw(1.2,0)node[anchor=north]{$U$}node[anchor=south]{$0$};
\draw(1.6,0)node[anchor=north]{$D$};
\draw(2.0,0)node[anchor=north]{$D$};
\draw(2.4,0)node[anchor=north]{$U$}node[anchor=south]{$0$};
\draw(2.8,0)node[anchor=north]{$D$};
\draw(3.2,0)node[anchor=north]{$D$};
}
\rightarrow
\tikzpic{-0.5}{
\draw(0,0)node[anchor=north]{$U$};
\draw(0.4,0)node[anchor=north]{$U$};
\draw(0.8,0)node[anchor=north]{$U$}node[anchor=south]{$2$};
\draw(1.2,0)node[anchor=north]{$U$};
\draw(1.6,0)node[anchor=north]{$D$};
\draw(2.0,0)node[anchor=north]{$D$};
\draw(2.4,0)node[anchor=north]{$U$};
\draw(2.8,0)node[anchor=north]{$D$};
\draw(3.2,0)node[anchor=north]{$D$};
\draw(1.2,0)..controls(1.2,0.3)and(1.6,0.3)..(1.6,0);
\draw(2.4,0)..controls(2.4,0.3)and(2.8,0.3)..(2.8,0);
}
\rightarrow
\tikzpic{-0.5}{
\draw(0,0)node[anchor=north]{$U$};
\draw(0.4,0)node[anchor=north]{$U$};
\draw(0.8,0)node[anchor=north]{$D$};
\draw(1.2,0)node[anchor=north]{$U$};
\draw(1.6,0)node[anchor=north]{$D$};
\draw(2.0,0)node[anchor=north]{$D$};
\draw(2.4,0)node[anchor=north]{$U$};
\draw(2.8,0)node[anchor=north]{$D$};
\draw(3.2,0)node[anchor=north]{$U$};
}
\end{eqnarray*}
\end{example}

The following theorem is clear from discussions in Section \ref{sec:bijHer} and 
the constructions of $\mathbf{p}_{\mathrm{nes}}$ and the lower path $\lambda$.
\begin{theorem}
\label{thrm:IdashI}
The map from $\mathbf{p}_{\mathrm{nes}}(\pi)$ to a labelling $H$ of type I 
is a bijection.
\end{theorem}

Similarly, we can construct a labelling $H''$ of type III from 
$\mathbf{p}_{\mathrm{cr}}(\pi)$ by deleting zeros, which are on the 
up steps of $\mu$, from $\mathbf{p}_{\mathrm{cr}}(\pi)$. 
The lower path $\lambda$ can be obtained in a similar way to 
$\mathbf{p}_{\mathrm{nes}}(\pi)$.
We replace $\mathbf{p}_{\mathrm{nes}}$ with $\mathbf{p}_{\mathrm{cr}}$,an up step 
with a down step, left with right, right with left
in the operation to obtain $\lambda$ from $\mathbf{p}_{\mathrm{nes}}(\pi)$.
Then, we have the following theorem.
\begin{theorem}
\label{thrm:IdashIII}
The map from $\mathbf{p}_{\mathrm{cr}}(\pi)$ to a labelling $H''$ of type III
is a bijection.
\end{theorem}

\begin{example}
Let $\pi_1$ be a generalized link pattern as in Example \ref{ex:GLP}.
We have two types of labellings: $H=(0,0,1,0,3,1,1,1,1)$ for type I and 
$H''=(1,0,1,4)$.
The lower paths $\lambda_{\mathrm{nes}}$ and $\lambda_{\mathrm{cr}}$ 
are given by $\lambda_{\mathrm{nes}}=UUDDUDUDUUUUU$ and 
$\lambda_{\mathrm{cr}}=UDUUDUDUDUUUU$.
Ballot tilings corresponding to $(\mu,H)$ and $(\mu,H'')$ are depicted as 
\begin{eqnarray*}
(\mu,H)=
\tikzpic{-0.5}{
\draw[thick](0,0)--(0.8,0.8)--(1.6,0)--(2,0.4)--(2.4,0)--(2.8,0.4)
             --(3.2,0)--(5.2,2);
\draw[thick](1.2,0.4)--(2.4,1.6)--(3.2,0.8)--(4.8,2.4)--(5.2,2);
\draw(2,1.2)--(2.8,0.4)(2.4,0.8)--(2.8,1.2)(2.8,0.4)--(3.2,0.8)--(3.6,0.4);
\draw(4,0.8)--(3.6,1.2)(4.4,1.2)--(4,1.6)(4.8,1.6)--(4.4,2);
\draw(0.4,0.4)node{\rotatebox{-45}{$-$}}(1.6,0.8)node{\rotatebox{-45}{$-$}};
\draw[red](1.4,0.6)--(1.6,0.4)--(2,0.8)--(2.6,0.2);
\draw[red](2.2,1.4)--(3.4,0.2);
\foreach \x in {0,1,2,3}{
\draw[red](3.4+0.4*\x,1+0.4*\x)--(3.8+0.4*\x,0.6+0.4*\x);
}
\draw(0.2,0.2)node[anchor=south east]{$0$};
\draw(0.6,0.6)node[anchor=south east]{$0$};
\draw(1.4,0.6)node[anchor=south east]{$1$};
\draw(1.8,1)node[anchor=south east]{$0$};
\draw(2.2,1.4)node[anchor=south east]{$3$};
\draw(3.4,1)node[anchor=south east]{$1$};
\draw(3.8,1.4)node[anchor=south east]{$1$};
\draw(4.2,1.8)node[anchor=south east]{$1$};
\draw(4.6,2.2)node[anchor=south east]{$1$};
}\qquad
(\mu,H'')=
\tikzpic{-0.5}{
\draw[thick](0,0)--(0.4,0.4)--(0.8,0)--(1.6,0.8)--(2,0.4)--(2.4,0.8)
            --(2.8,0.4)--(3.2,0.8)--(3.6,0.4)--(5.2,2);
\draw[thick](0.4,0.4)--(0.8,0.8)--(1.2,0.4)--(2.4,1.6)--(3.2,0.8)
            --(4.8,2.4)--(5.2,2);
\draw(4.4,2)--(4.8,1.6)(4,1.6)--(4.4,1.2)(3.6,1.2)--(4,0.8);
\draw(2,1.2)node{\rotatebox{-45}{$-$}};
\draw(2.8,1.2)node{\rotatebox{45}{$-$}};
\draw[red](1,0.6)--(0.6,0.2)(1.8,0.6)--(2.4,1.2)--(2.8,0.8)--(3,1);
\draw[red](3.4,0.6)--(5,2.2);
\draw(1,0.6)node[anchor=south west]{$1$};
\draw(2.6,1.4)node[anchor=south west]{$0$};
\draw(3,1)node[anchor=south west]{$1$};
\draw(5,2.2)node[anchor=south west]{$4$};
}
\end{eqnarray*}
\end{example}

Let $\pi$ be a generalized perfect matching in $\mathrm{PM}_{I}(\mu)$. 
By definition, a non-zero element in the integer sequences 
$\mathbf{p}_{\mathrm{nes}}(\pi)$ and $\mathbf{p}_{\mathrm{cr}}(\pi)$ 
is on an up step of $\mu$.
The labelling $H_{\mathrm{nes}}$ (resp. $H_{\mathrm{cr}}$) can be obtained 
by deleting zeros, which are on the down steps of $\mu$, from 
$\mathbf{p}_{\mathrm{nes}}$ (resp. $\mathbf{p}_{\mathrm{cr}}$).
The lower paths $\lambda_{\mathrm{nes}}$ and $\lambda_{\mathrm{cr}}$ can 
be obtained as follow.
Since we have the same algorithm for $\mathbf{p}_{\mathrm{nes}}$ 
and $\mathbf{p}_{\mathrm{cr}}$, we abbreviate $\mathbf{p}_{\mathrm{nes}}$ or
$\mathbf{p}_{\mathrm{nes}}$ as $\mathbf{p}$.
We put the integer sequence $\mathbf{p}$ on the path $\mu$.
We enumerate the up steps in $\mu$ from right to left by $1,2,\ldots,N_{U}$ 
where $N_{U}$ is the number of up steps in $\mu$.
We perform the following operations on $\mu$ starting from the first up step and 
ending with the $N_{U}$-th up step.
Suppose that the $i$-th up step has an integer $n_i$.
Let $\mu_{i}$ be a partial path in $\mu$ which is right to the $i$-th up step.
We perform the operations (A) and (B) on $\mu_{i}$ as in Section \ref{subsec:im}.
We have paired $UD$'s and unpaired $D$'s and $U$'s.
For the unpaired $U$'s, we make a pair of two $U$'s from right to left and 
connect them by a dashed arc.
We have at most one $U$ which does not form a pair. 
We enumerate $U$ in a paired $UU$'s, unpaired $D$'s and a $U$ in $\mu_{i}$ 
from left to right by $1,2,\ldots$ and let $m_j$ be the position of 
the $j$-th $U$ or $D$ in $\mu$.
Note that we do not enumerate paired $UD$'s.
We denote by $m_0$ the position of the $i$-th up step in $\mu$.
Let $N(i)$ be an integer such that 
the sum of the numbers of paired $UU$, unpaired $D$ and unpaired $U$ 
is $n_{i}$ in a partial sequence $\nu_{i}$ in $\mu_{i}$
which is left to the $(N(i)+1)$-th element. 
If such $N(i)$ does not exist, then we attach a sequence consisting of only $D$'s 
to the right of $\mu_{i}$ and define $N(i)$ as above.
For $1\le j\le N(i)$, we move the $j$-th $U$ or $D$ to
the position $m_{j-1}$ and put an up step at the position $m_{N(i)}$.
Further, if two $U$'s are connected by dashed arc and moved leftward, 
we flip these two $U$'s into two $D$'s.
The obtained path is the lower path $\lambda$.

\begin{example}
Let $\mu:=UUUDUUDUUUU$ and $\mathbf{p}:=(0,0,0,0,0,2,0,0,0,1,1)$. 
The lower path $\lambda$ is obtained as 
\begin{eqnarray*}
\tikzpic{-0.5}{
\draw(0,0)node[anchor=north]{$U$}node[anchor=south]{$0$};
\draw(0.4,0)node[anchor=north]{$U$}node[anchor=south]{$0$};
\draw(0.8,0)node[anchor=north]{$U$}node[anchor=south]{$0$};
\draw(1.2,0)node[anchor=north]{$D$};
\draw(1.6,0)node[anchor=north]{$U$}node[anchor=south]{$0$};
\draw(2.0,0)node[anchor=north]{$U$}node[anchor=south]{$2$};
\draw(2.4,0)node[anchor=north]{$D$};
\draw(2.8,0)node[anchor=north]{$U$}node[anchor=south]{$0$};
\draw(3.2,0)node[anchor=north]{$U$}node[anchor=south]{$0$};
\draw(3.6,0)node[anchor=north]{$U$}node[anchor=south]{$1$};
\draw(4,0)node[anchor=north]{$U$}node[anchor=south]{$1$};
}
&\rightarrow&
\tikzpic{-0.5}{
\draw(0,0)node[anchor=north]{$U$}node[anchor=south]{$0$};
\draw(0.4,0)node[anchor=north]{$U$}node[anchor=south]{$0$};
\draw(0.8,0)node[anchor=north]{$U$}node[anchor=south]{$0$};
\draw(1.2,0)node[anchor=north]{$D$};
\draw(1.6,0)node[anchor=north]{$U$}node[anchor=south]{$0$};
\draw(2.0,0)node[anchor=north]{$U$}node[anchor=south]{$2$};
\draw(2.4,0)node[anchor=north]{$D$};
\draw(2.8,0)node[anchor=north]{$U$}node[anchor=south]{$0$};
\draw(3.2,0)node[anchor=north]{$U$}node[anchor=south]{$0$};
\draw(3.6,0)node[anchor=north]{$U$}node[anchor=south]{$1$};
\draw(4,0)node[anchor=north]{$D$};
}\rightarrow
\tikzpic{-0.5}{
\draw(0,0)node[anchor=north]{$U$};
\draw(0.4,0)node[anchor=north]{$U$};
\draw(0.8,0)node[anchor=north]{$U$};
\draw(1.2,0)node[anchor=north]{$D$};
\draw(1.6,0)node[anchor=north]{$U$};
\draw(2.0,0)node[anchor=north]{$U$}node[anchor=south]{$2$};
\draw(2.4,0)node[anchor=north]{$D$};
\draw(2.8,0)node[anchor=north]{$U$};
\draw(3.2,0)node[anchor=north]{$U$};
\draw(3.6,0)node[anchor=north]{$D$};
\draw(4,0)node[anchor=north]{$U$};
} \\
&\rightarrow&
\tikzpic{-0.5}{
\draw(0,0)node[anchor=north]{$U$};
\draw(0.4,0)node[anchor=north]{$U$};
\draw(0.8,0)node[anchor=north]{$U$};
\draw(1.2,0)node[anchor=north]{$D$};
\draw(1.6,0)node[anchor=north]{$U$};
\draw(2.0,0)node[anchor=north]{$U$}node[anchor=south]{$2$};
\draw(2.4,0)node[anchor=north]{$D$};
\draw(2.8,0)node[anchor=north]{$U$};
\draw(3.2,0)node[anchor=north]{$U$};
\draw(3.6,0)node[anchor=north]{$D$};
\draw(4,0)node[anchor=north]{$U$};
\draw(3.2,0)..controls(3.2,0.3)and(3.6,0.3)..(3.6,0);
\draw[dashed](2.8,0)..controls(2.8,0.6)and(4,0.6)..(4,0);
}
\rightarrow
\tikzpic{-0.5}{
\draw(0,0)node[anchor=north]{$U$};
\draw(0.4,0)node[anchor=north]{$U$};
\draw(0.8,0)node[anchor=north]{$U$};
\draw(1.2,0)node[anchor=north]{$D$};
\draw(1.6,0)node[anchor=north]{$U$};
\draw(2.0,0)node[anchor=north]{$D$};
\draw(2.4,0)node[anchor=north]{$D$};
\draw(2.8,0)node[anchor=north]{$D$};
\draw(3.2,0)node[anchor=north]{$U$};
\draw(3.6,0)node[anchor=north]{$D$};
\draw(4,0)node[anchor=north]{$U$};
}
\end{eqnarray*}
The Hermite history $(\mu,H)$ is given by
\begin{eqnarray*}
(\mu,H)=
\tikzpic{-0.5}{
\draw(0,0)--(1.2,1.2)--(1.6,0.8)--(2,1.2)--(3.2,0)--(3.6,0.4)--(4,0)--(4.8,0.8);
\draw(2,1.2)--(2.4,1.6)--(2.8,1.2)--(4.4,2.8)--(4.8,2.4);
\draw(2.4,0.8)--(2.8,1.2)--(3.2,0.8);
\draw(2.8,0.4)--(3.6,1.2)--(4,0.8)--(4.8,1.6);
\draw(4.4,1.2)--(4.8,0.8)(4,2.4)--(4.8,1.6)(4.4,2)--(4.8,2.4);
\draw(3.6,2)--(4,1.6)--(4.4,2);
\draw(0.4,0.4)node{\rotatebox{-45}{$-$}};
\draw(0.8,0.8)node{\rotatebox{-45}{$-$}};
\draw(3.2,1.6)node{\rotatebox{-45}{$-$}};
\draw(4.4,0.4)node{\rotatebox{-45}{$-$}};
\draw(0.2,0.2)node[anchor=south east]{$0$};
\draw(0.6,0.6)node[anchor=south east]{$0$};
\draw(1,1)node[anchor=south east]{$0$};
\draw(1.8,1)node[anchor=south east]{$0$};
\draw(2.2,1.4)node[anchor=south east]{$2$};
\draw(3,1.4)node[anchor=south east]{$0$};
\draw(3.4,1.8)node[anchor=south east]{$0$};
\draw(3.8,2.2)node[anchor=south east]{$1$};
\draw(4.2,2.6)node[anchor=south east]{$1$};
\draw[red](2.2,1.4)--(3.2,0.4)--(3.6,0.8)--(4,0.4)--(4.6,1);
\draw[red](3,1.4)--(3.2,1.2)--(3.6,1.6)--(4,1.2)--(4.6,1.8);
\draw[red](3.8,2.2)--(4.2,1.8)(4.2,2.6)--(4.6,2.2);
}
\end{eqnarray*}
Note that the statistics tiles for a ballot tile of length $(2,1)$ is zero.
\end{example}

Starting with $\mu$ and $\mathbf{p}:=\mathbf{p}_{\mathrm{nes}}$ or 
$\mathbf{p}_{\mathrm{cr}}$, we construct  a sequence of lower paths 
$\lambda_1,\ldots\lambda_{p}$ corresponding to an up step in $\mu$.
One can put a ballot tile (or two ballot tiles of length $(2n,n')$ with $n'\ge1$) 
in the shape $\lambda_{i}/\mu$.
Thus, we obtain the pair $(\mu,H)$. Given a pair $(\mu,H)$, $\mathbf{p}$ is 
obtained by inserting zeros into a label $H$ at the position of down steps in $\mu$. 
Thus, we have the following theorem.
\begin{theorem}
\label{thrm:PMI-Hh}
The map from $\mathbf{p}_{\mathrm{nes}}$ or $\mathbf{p}_{\mathrm{cr}}$ to 
a labelling $H$ of type I is a bijection. 
\end{theorem}

Let $\pi$ be a generalized perfect matching in $\mathrm{PM}_{II}(\lambda)$.
We construct a pair $(\mu,H)$ of type I from $\pi$.
Recall that the generalized perfect matching $\pi$ consists of arcs and 
dashed arcs.
If the $i$-th node and the $j$-th node are connected by an arc, we define 
$\mu_{i}:=U$ and $\mu_{j}:=D$. If the $i$-th and the $j$-th nodes are connected
by a dashed arc, we define $\mu_{i}:=U$ and $\mu_{j}:=U$.
Then, $\mu=\mu_1\ldots \mu_{N}$ is a ballot path where $N$ is the number of
nodes in $\pi$.
The labelling $H$ can be obtained by deleting zeros, which are on the down steps 
of $\mu$, from $\mathbf{p}_{\mathrm{nes}}$. 
\begin{prop}
\label{prop:HhBTS}
We have 
\begin{eqnarray*}
|H|=\mathrm{nes}(\pi)=\mathrm{tiles}(\mathrm{BTS}(\lambda,\sigma)).
\end{eqnarray*}
\end{prop}

\section{Ballot tiling of type BI with a fixed upper path}
\label{sec:BI-up}
\subsection{Dyck tiling with a fixed upper path}
\label{sec:D-up}
Let $\gamma$ be a ballot path of length $(2n,n')$ and 
$\gamma_0:=\underbrace{UD\ldots UD}_{2n}\underbrace{U\ldots U}_{n'}$. 
Since a Dyck path can be obtained from $\gamma$ by adding $n'$ $U$'s from right,
we also denote by $\gamma$ the obtained Dyck path. 
Then, the shape $\gamma'/\gamma$ with $\gamma_{0}\le\gamma'\le\gamma$ can be 
obtained as a region determined by two Dyck paths $\gamma$ and $\gamma_{0}$.
We consider the generating function 
\begin{eqnarray*}
\widetilde{P}_{\gamma}:=\sum_{\gamma_{0}\le\gamma'\le\gamma}
\sum_{D\in\mathcal{D}(\gamma'/\gamma)}q^{\mathrm{tiles}(D)}.
\end{eqnarray*}
Let $\mathrm{D}(\gamma_{0})$ be the set of $D$ steps in $\gamma_0$.
Given a down step $d$ in $\gamma_0$, we denote by $\mathrm{ht}(d;\gamma)$ 
one plus the number of boxes in $\gamma_0/\gamma$ which lie in the $(1,1)$-direction 
from $d$.
\begin{theorem}
\label{thrm:gD-up}
We have 
\begin{eqnarray*}
\widetilde{P}_{\gamma}=\prod_{d\in\mathrm{D}(\gamma_{0})}[\mathrm{ht}(d;\mu)]
\end{eqnarray*}
\end{theorem}
We omit the proof since one can apply the same argument as the proof of the first 
equality in Theorem \ref{thrm-A2}.

Recall that $\mathrm{PM}'_{I}(\gamma)$ is the set of generalized perfect matching
of type I and satisfying two conditions (C1) and (C2).
Then,  
\begin{theorem}
\label{thrm:gD-up-gpm}
We have 
\begin{align}
\label{eqn:GDyckU}
\begin{split}
\widetilde{P}_{\gamma}&=\sum_{\pi\in\mathrm{PM}'_{I}(\gamma)}q^{\mathrm{nes}(\pi)} \\
&=\sum_{\pi\in\mathrm{PM}'_{I}(\gamma)}q^{\mathrm{cr}(\pi)}
\end{split}
\end{align}
\end{theorem}
\begin{proof}
The generating function $\widetilde{P}_{\gamma}$ is equal to the sum of $q^{|H|}$ over all possible 
Hermite histories $H$'s. 
We have $|H|=\mathrm{nes}(\pi)$ 
for a generalized perfect matching $\pi$ corresponding to the pair $(\gamma,H)$.
From Theorem \ref{thrm:IdashI}, we have 
$\widetilde{P}_{\gamma}=\sum_{\pi\in\mathrm{PM}'_{I}(\gamma)}q^{\mathrm{nes}(\pi)}$.

The second equality in Eqn.(\ref{eqn:GDyckU}) follows from the similar argument with 
Theorem \ref{thrm:IdashIII}.
\end{proof}

\begin{example}
Let $\gamma=UUUDUD$. From Theorem \ref{thrm:gD-up}, the generating function 
$\widetilde{P}_{\gamma}=[3]^{2}$.
We have nine generalized perfect matchings in $\mathrm{PM}_{I}'(\gamma)$:
\begin{eqnarray*}
&
\begin{array}{c|ccccc}
\pi & 
\tikzpic{-0.3}{
\draw(0,0)..controls(0,0.9)and(2,0.9)..(2,0);
\draw(0.4,0)..controls(0.4,0.6)and(1.2,0.6)..(1.2,0);
\draw[dashed](0.8,0)..controls(0.8,0.6)and(1.6,0.6)..(1.6,0);
} &
\tikzpic{-0.3}{
\draw(0,0)..controls(0,0.9)and(2,0.9)..(2,0);
\draw(0.4,0)..controls(0.4,0.6)and(1.6,0.6)..(1.6,0);
\draw[dashed](0.8,0)..controls(0.8,0.3)and(1.2,0.3)..(1.2,0);
} & 
\tikzpic{-0.3}{
\draw(0,0)..controls(0,0.9)and(1.2,0.9)..(1.2,0);
\draw(0.4,0)..controls(0.4,0.9)and(2,0.9)..(2,0);
\draw[dashed](0.8,0)..controls(0.8,0.6)and(1.6,0.6)..(1.6,0);
} &
\tikzpic{-0.3}{
\draw[dashed](0,0)..controls(0,0.9)and(1.6,0.9)..(1.6,0);
\draw(0.4,0)..controls(0.4,0.9)and(2,0.9)..(2,0);
\draw(0.8,0)..controls(0.8,0.4)and(1.2,0.4)..(1.2,0);
} &
\tikzpic{-0.3}{
\draw(0,0)..controls(0,0.6)and(1.2,0.6)..(1.2,0);
\draw[dashed](0.4,0)..controls(0.4,0.9)and(1.6,0.9)..(1.6,0);
\draw(0.8,0)..controls(0.8,0.6)and(2,0.6)..(2,0);
} \\[9pt] \hline
\mathrm{nes}(\pi) & 3 & 4 & 2 & 3 & 1 \\
\mathrm{cr}(\pi)  & 3 & 2 & 4 & 1 & 3 
\end{array}& \\
&
\begin{array}{c|ccccc}
\pi & 
\tikzpic{-0.3}{
\draw[dashed](0,0)..controls(0,0.9)and(1.6,0.9)..(1.6,0);
\draw(0.4,0)..controls(0.4,0.6)and(1.2,0.6)..(1.2,0);
\draw(0.8,0)..controls(0.8,0.6)and(2,0.6)..(2,0);
} &
\tikzpic{-0.3}{
\draw(0,0)..controls(0,0.7)and(1.2,0.7)..(1.2,0);
\draw[dashed](0.4,0)..controls(0.4,0.4)and(0.8,0.4)..(0.8,0);
\draw(1.6,0)..controls(1.6,0.3)and(2,0.3)..(2,0);
} & 
\tikzpic{-0.3}{
\draw[dashed](0,0)..controls(0,0.5)and(0.8,0.5)..(0.8,0);
\draw(0.4,0)..controls(0.4,0.5)and(1.2,0.5)..(1.2,0);
\draw(1.6,0)..controls(1.6,0.4)and(2,0.4)..(2,0);
} &
\tikzpic{-0.3}{
\draw[dashed](0,0)..controls(0,0.4)and(0.4,0.4)..(0.4,0);
\draw(0.8,0)..controls(0.8,0.4)and(1.2,0.4)..(1.2,0);
\draw(1.6,0)..controls(1.6,0.4)and(2,0.4)..(2,0);
} \\[9pt] \hline
\mathrm{nes}(\pi) & 2 & 2 & 1 & 0 \\
\mathrm{cr}(\pi)  & 2 & 2 & 1 & 0  
\end{array}
&
\end{eqnarray*}
From Theorem \ref{thrm:gD-up-gpm}, we have 
$\widetilde{P}_{\gamma}=\sum_{\pi}q^{\mathrm{nes}(\pi)}
=\sum_{\pi}q^{\mathrm{cr}(\pi)}$.
\end{example}

\subsection{Ballot tiling of type BI with a fixed upper path}
\label{subsec:BI-up}
Recall that the lowest ballot path of length $N=2n+n'$ 
is the zig-zag path $\mu_{0}$, {\it i.e.}, 
$\mu_{0}=\underbrace{UDUD\ldots}_{N}$. 
Below, we restrict $N$ to be odd and 
consider the generating function $P_{\ast,\mu}^{I}$.  

Recall that the two paths $\mu$ and $\mu_{0}$ determine the skew shape 
$\mu_{0}/\mu$.
Let $\mathrm{D}(\mu_0)$ be the set of $D$ steps in $\mu_{0}$.
Given a down step $d$ in $\mu_0$, we denote by $\mathrm{ht}(d;\mu)$ 
one plus the number of boxes in $\mu_0/\mu$ which lie  
in the $(1,1)$-direction from $d$.

A permutation $w$ in $\mathfrak{S}_{N+1}$ is called 
a {\it down-up} (or {\it alternating}) permutation if 
the one-line notation $w$ satisfies 
$w_{2m-1}>w_{2m}<w_{2m+1}$ for $m\ge1$.  	
Let $\mathfrak{S}^{\mathrm{DU}}_{N+1}$ be the set of down-up permutations on 
$\{1,2,\ldots,N+1\}$ for which the entries in the even 
positions are increasing, {\it i.e.}, $w_{2m}<w_{2m+2}$ for $m\ge1$. 

We enumerate $D$ steps in $\mu_0$ from left to right by $1,2,\ldots,(N-1)/2$.
Let $S_{0}:=\{1,2,\ldots,N+1\}$ and $\mu$ be a path above $\mu_0$.
We define a permutation $\sigma=\sigma_1\ldots \sigma_{N+1}:=\sigma(\mu)=$ 
in $\mathfrak{S}^{\mathrm{DU}}_{N+1}$ as follows.
For $1\le i\le(N-1)/2$, we define $\sigma_{2i-1}$ as the $(\mathrm{ht}(d_{2i-1};\mu)+1)$-th 
smallest element in $S_{2i-2}$ and $S_{2i-1}:=S_{2i-2}\setminus\{\sigma_{2i-1}\}$.
Then, $\sigma_{2i}$ is the smallest element in $S_{2i-1}$ and we 
define $S_{2i}:=S_{2i-1}\setminus\{\sigma_{2i}\}$.
We define $\sigma_{N}=\max(S_{N-1})$ and $\sigma_{N+1}=\min(S_{N-1})$.
By construction, the obtained permutation is in $\mathfrak{S}^{\mathrm{DU}}_{N+1}$.
\begin{example}
We have $(\mathrm{ht}(d_1;\mu),\mathrm{ht}(d_2;\mu))=(3,3)$ 
for $\mu=UUUDU$. Thus we have $\sigma(\mu)=416253$.
\end{example}

\begin{theorem}
\label{thrm:BIII-Hhfac}
We have 
\begin{align}
\label{eqn:BalHhfac}
\begin{split}
P_{\ast,\mu}^{I}&=\prod_{d\in\mathrm{D}(\mu_0)}[\mathrm{ht}(d;\mu)] \\
&=\sum_{\substack{\pi\le\sigma(\mu) \\\pi\in\mathfrak{S}^{\mathrm{DU}}_{N+1}}}
q^{l(\pi)-l(\sigma(\mu))}
\end{split}
\end{align}
where $\le$ is the Bruhat order on $\mathfrak{S}_{N+1}$ and 
$l$ is the inversion on $\mathfrak{S}_{N+1}$.
\end{theorem}
\begin{proof}

Fix a ballot path $\mu$ of odd length and a down step $d$ of $\mu$ or a 
northeast edge of an anchor box in $\mu_{0}/\mu$.
The number of boxes in $\mu_0/\mu$ in the southwest direction from $d$
is $\mathrm{ht}(d,\mu)$.
The $q$-integer $[\mathrm{ht}(d,\mu)]$ is written as a expansion
$1+q+\ldots+q^{\mathrm{ht}(d,\mu)-1}$.
Thus, when we take a term $q^{i}$ from $[\mathrm{ht}(d,\mu)]$, 
we place $i$ successive boxes in the region between $d$ and the zig-zag path
such that the northeast edge of the northeast box is attached to the 
down step $d$. 
As in Section \ref{sec:bijHer}, we place boxes for each down step of $\mu$ 
for a product of terms of $\prod_{d}[\mathrm{ht}(d;\mu)]$.
We denote by $B'$ the obtained diagram consisting boxes.

We construct a ballot tiling from $B'$ as follows.
We enumerate the down steps of $\mu$ from left to right  
by $d_1,\ldots,d_{n}$ where $n$ is the number of down steps in $\mu$.
We also enumerate the boxes attached to the down step $d_{i}$ from 
top to bottom by $1,2,\ldots,r(i)$ where $r(i)$ is the number of boxes 
attached to $d_{i}$. 
Given a pair $(d_i,j)$ with $1\le j\le r(i)$, we call the corresponding box 
$(d_i,j)$ box. 
If there exists a box to the northwest of $(d_i,j)$ box and it forms a ballot 
tile, we move to the next pair $(d_i,j+1)$ for $1\le j\le r(i)-1$ or $(d_{i+1},1)$ 
for $j=r(i)$.
Otherwise, there is no box to the northwest of $(d_i,j)$ box. 
Let $b$ (resp. $b'$) be a rightmost box in $\mu_0/\mu$ left to the $(d_i,j)$ box such that 
the translation of $b$ by $(2,0)$ (resp. $(-1,1)$) is either outside of the region 
$\mu_0/\mu$ or contained in a ballot tile.
In some cases, $b$ coincides with $b'$ but $b$ and $b'$ are different boxes in general.
By definition, the box $b$ is right to $b'$.
If $b$ is the $(d_i,j)$ box, we have a unique ballot tile $B$ such that it satisfies 
the cover-inclusive property, it start from $b$ and ends at the tile $b'$.
We move the boxes $(d_i,k)$ with $i+1\le k\le r(i)$ to the southwest of the box $b'$. 
Then, we move the next pair $(d_{i},j+1)$.
If $b$ is not the $(d_i,j)$ box, we have two ballot tiles $B$ and $B'$ of length 
$(2n,n')$ with $n'\ge1$ such that $B$ is just above $B'$ and southwest of the leftmost 
box in $B$ is $b$.
We also have a unique ballot tile (Dyck tile more precisely) starting from $b$ and 
ending at $b'$ as above. 
We move the boxes $(d_i,k)$ with $i+1\le k\le r(i)$ to the southwest of the box $b'$. 
Then we move the next pair $(d_i,j+1)$.
\begin{figure}[ht]
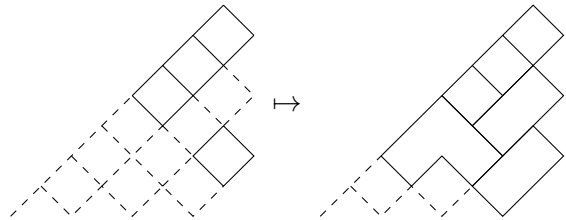

\begin{eqnarray*}
\tikzpic{-0.5}{
\draw[dashed](0,0)--(1.6,1.6);
\draw(1.6,1.6)--(2.8,2.8)(1.6,1.6)--(2,1.2)--(3.2,2.4)
(2.8,2.8)--(3.2,2.4)(2.4,2.4)--(2.8,2)(2,2)--(2.4,1.6);
\draw[dashed](0.4,0.4)--(0.8,0)--(1.2,0.4)--(1.6,0)--(2,0.4)--(2.4,0)--(2.8,0.4);
\draw[dashed](0.8,0.8)--(1.2,0.4)(1.2,1.2)--(2,0.4)(2,1.2)--(2.4,0.8)
             (1.2,0.4)--(2,1.2)(2,0.4)--(2.4,0.8)(2.8,1.2)--(3.2,1.6)--(2.8,2)
             (2.4,1.6)--(2.8,1.2);
\draw(2.8,1.2)--(3.2,0.8)--(2.8,0.4)--(2.4,0.8)--(2.8,1.2);
}\mapsto
\tikzpic{-0.5}{
\draw[dashed](0,0)--(0.8,0.8);
\draw[dashed](0.4,0.4)--(0.8,0)--(1.2,0.4)--(1.6,0)--(2,0.4);
\draw(0.8,0.8)--(1.6,1.6)--(2.4,0.8)--(2,0.4)--(1.6,0.8)--(1.2,0.4)--(0.8,0.8);
\draw(1.6,1.6)--(2.8,2.8)--(3.2,2.4)--(2,1.2)--(1.6,1.6)(2,2)--(2.4,1.6)(2.4,2.4)--(2.8,2);
\draw(2.8,2)--(3.2,1.6)--(2.4,0.8)--(2,1.2)--(2.8,2);
\draw(2,0.4)--(2.4,0)--(3.2,0.8)--(2.8,1.2)--(2,0.4);
}
\end{eqnarray*}
\caption{The bijection from a diagram $B'$ to a ballot tiling.
The statistics $\mathrm{tiles}$ of the ballot tiling is $\mathrm{tiles}(B)=4$.}
\label{fig:BalHh}
\end{figure}
We obtain a ballot tiling by visiting all boxes $(d_i,j)$ and performing the above-mentioned 
operation on them. 
It is obvious that the above operation is invertible, {\it i.e.}, one can construct a 
diagram $B'$ from a ballot tiling. 
See Fig. \ref{fig:BalHh} for an example of this bijection.
Thus we obtain a bijective proof of the first equality in Eqn.(\ref{eqn:BalHhfac}).

To prove the second equality, we show that there exists a bijection between 
$B'$ and $\pi\le\sigma(\mu)$ with $\pi\in\mathfrak{S}^{\mathrm{DU}}_{N+1}$.
It is easy to show that the cardinality of $\mathfrak{S}^{\mathrm{DU}}_{N+1}$
is given by $|\mathfrak{S}^{\mathrm{DU}}_{N+1}|=N!!$.

We construct $B''$ from $B'$ by translating the $i$ successive boxes attached to 
a down step in $\mu$ by $(-1,-1)$ direction such that the southwest edge of the 
southwest box is attached to the path $\mu_{0}$.
We enumerate down steps of $\mu_0$ from left to right by $1,2,\ldots,(N-1)/2$ 
and denote by $e_{j}$ the $j$-th down step.
Suppose that the down step $e_j$ has $r(j)$ boxes in the $(1,1)$ direction from $e_{j}$.
We define a product of simple transpositions as $\rho_{j}:=s_{2j+r(j)-1}\ldots s_{2j+1}s_{2j}$.
We consider the following ordered product of simple transpositions in $\mathfrak{S}_{N+1}$:
\begin{eqnarray*}
\rho(B'):=\rho_1\rho_2\ldots\rho_{(N-1)/2}.
\end{eqnarray*}
The smallest element with respect to the Bruhat order in $\mathfrak{S}_{N+1}^{\mathrm{DU}}$
is $w^{0}:=w_{1}\ldots w_{N+1}$ with $w_{2j-1}=2j$ and $w_{2j}=2j-1$ in one-line notation. 
The $\rho_{j}$ transpose $2j$ with $2j+r(j)$ and $k$ with $k-1$ for $2j<k\le 2j+r(j)$.
One can easily show that $\rho_{(N-1)/2}w^{0}$ is in $\mathfrak{S}^{\mathrm{DU}}_{N+1}$.
More in general, we have $\rho(B')w^{0}\in\mathfrak{S}^{\mathrm{DU}}_{N+1}$. 
The cardinality of possible $B'$'s is given by 
$\prod_{d\in\mathrm{D}(\mu_0)}\mathrm{ht}(d;\mu)$, which implies 
that $|B'|=N!!$ for $\mu=UU\ldots U$.
Thus we have a bijection between $B'$ and a permutation $\pi\in\mathfrak{S}^{\mathrm{DU}}_{N+1}$.
The second equality in Eqn.(\ref{eqn:BalHhfac}) follows from the direct consequence of this bijection.
\end{proof}

The generating function $P_{\ast,\mu}^{I}$ is related to generalized 
perfect matchings in $\mathrm{PM}_{I}(\mu)$ as follows.
\begin{theorem}
\label{thrm:bIgPM}
We have 
\begin{eqnarray*}
P_{\ast,\mu}^{I}&=&
\sum_{\pi\in\mathrm{PM}_{I}(\mu)}q^{\mathrm{nes}(\pi)} \\
&=&\sum_{\pi\in\mathrm{PM}_{I}(\mu)}q^{\mathrm{cr}(\pi)}.
\end{eqnarray*}
\end{theorem}
\begin{proof}
From Theorem \ref{thrm:BIII-Hhfac}, the number of ballot tilings 
whose fixed upper path is $\mu$ is given by 
$\prod_{d\in\mathrm{D}(\mu_{0})}\mathrm{ht}(d;\mu)$.
By a direct computation, it is easy to show that the cardinality of 
the set $\mathrm{PM}_{I}(\mu)$ is also given 
by $\prod_{d\in\mathrm{D}(\mu_{0})}\mathrm{ht}(d;\mu)$.
Thus we have a bijection between a ballot tiling $B$ with a fixed upper 
path $\mu$ and a generalized perfect matching $\pi$ in $\mathrm{PM}_{I}(\mu)$.
Let $(\mu,H)$ be an Hermite history for a ballot tiling with a fixed upper path 
$\mu$.
From $\mathrm{tiles}(B)=|H|$, Theorem \ref{thrm:PMI-Hh} and 
Proposition \ref{prop:HhBTS}, we have 
$\mathrm{tiles}(B)=|H|=\mathrm{nes}(\pi)=\mathrm{cr}(\pi)$.
This completes the proof.
\end{proof}

\section{Ballot tiling of type BI with a fixed lower path}
\label{sec:BI-low}
\subsection{\texorpdfstring{Generating function $P(M,N)$}{Generating function P(M,N)}}
\label{sec:BalP}
Let $\lambda_{M,N}$ be a path 
$\lambda_{M,N}:=\underbrace{D\ldots D}_{N}\underbrace{U\ldots U}_{M}$.
We abbreviate the generating function $P^{I}_{\lambda_{M,N},\ast}$ of
ballot tilings of type BI as $P(M,N)$. 

\begin{lemma}
\label{lemma:recr}
The generating function $P(M,N)$ satisfies the following recurrence relation: 
\begin{eqnarray}
\label{recrel}
\qquad
P(M,N)=
\begin{cases}
P(M,N-1)+q^{N}P(M-1,N), & M: \text {even}, \\
P(M,N-1)+q^{N}P(M-1,N)+q^{M+N}P(M+1,N-2), & M: \text{odd}.
\end{cases}
\end{eqnarray}
\end{lemma}
\begin{proof}
We first consider the case where $M$ is even.
Let $Y$ be the shifted shape determined by $\lambda_{M,N}$ as in 
Section \ref{sec:im}. 
Let $b$ be the leftmost box of $Y$. In other words, the southwest edge 
of $b$ is attached to the first down step in $\lambda_{M,N}$.
Let $B$ be a cover-inclusive ballot tiling of type I in the shape $Y$.
The generating function $P(M,N)$ is the sum of two contributions:
the first one is that the box $b$ does not form a ballot tile in $B$,
and the second one is that the box $b$ forms a ballot tile in $B$.
In the first case, we have $P(M,N-1)$.
In the second case, since $\lambda_{M,N}$ does not have a partial 
sequence $DU$, there is no ballot tiles of length $(2n,n')$ with 
$n\ge1$ and $n'\ge0$ in $Y$.
The box $b$ forms a ballot tile of length $(0,0)$ and so do the $N-1$ boxes
which lie in the southeast direction of $b$.
The second contribution is $q^{N}P(M-1,N)$, which implies Eqn.(\ref{recrel}) for 
$M$ even.

Below, we consider the case where $M$ is odd.
We have three contributions for $P(M,N)$ and two of them are the same as 
in case of $M$ even.
The third contribution is that the box $b$ forms a ballot tile of 
length $(0,0)$ and there are two ballot tiles of length $(0,M)$ in $Y$.
Note that the unique lowest box in $Y$ forms the ballot tile of length $(0,M)$. 
The contribution of the third configuration is 
$q^{M+N}P(M+1,N-2)$, which implies Eqn.(\ref{recrel}).

\end{proof}

For positive integers $m$ and $N$, we define 
\begin{eqnarray*}
a_{(2m-1,N)}&:=&\genfrac{}{}{}{}{[N+2m]}{[2m]}, \\
a_{(2m,N)}&=&\genfrac{}{}{}{}{[2N+2m]}{[N+2m]}.
\end{eqnarray*}
The recurrence relations (Lemma \ref{lemma:recr}) can be solved 
by using $a_{(M,N)}$.
\begin{prop}
\label{prop:BalP}
The generating function $P_{M,N}$ satisfies
\begin{eqnarray}
P(M,N)=\prod_{1\le i\le N}(1+q^{i})\cdot 
\prod_{1\le j\le M}a_{(j,N)}.
\end{eqnarray}
\end{prop}
\begin{proof}
We prove Proposition by induction. 
When $(M,N)=(1,1)$, we have $a_{(1,1)}=[3]/[2]$. 
Since the generating function is calculated as $P(1,1)=[2][3]/[2]=[3]$,
Proposition holds true when $(M,N)=(1,1)$.

Assume that Proposition holds true for $P(m,n)$ such that $m+n\le M+N-1$.
We first consider the case where $M=2m$.
From the definition of $a_{(M,N)}$, we have 
\begin{eqnarray}
\label{eqn:a1}
\prod_{1\le j\le 2m}a_{(j,N-1)}&=&\prod_{1\le j\le m}a_{(2j,N-1)}a_{(2j-1,N-1)} \nonumber \\
&=&\prod_{1\le j\le m}\genfrac{}{}{}{}{[2N+2j-2]}{[2j]}, 
\end{eqnarray}
and 
\begin{eqnarray}
\label{eqn:a2}
\prod_{j=1}^{2m-1}a_{(j,N)}&=&a_{(2m-1,N)}\prod_{1\le j\le m-1}a_{(2j,N)}a_{(2j-1,N)} \nonumber \\
&=&\genfrac{}{}{}{}{[N+2m]}{[2N+2m]}\prod_{1\le j\le m}\genfrac{}{}{}{}{[2N+2j]}{[2j]}.
\end{eqnarray} 
The right hand side of Eqn.(\ref{recrel}) is written as 
\begin{eqnarray*}
P(M,N-1)+q^{N}P(M-1,N)
&=&\prod_{i=1}^{N-1}(1+q^{i})\left\{
\prod_{1\le j\le M}a_{(j,N-1)}+q^{N}(1+q^{N})\prod_{1\le j\le M-1}a_{(j,N)}
\right\} \\
&=&\prod_{i=1}^{N-1}(1+q^{i})\prod_{j=1}^{m}\genfrac{}{}{}{}{[2N+2j]}{[2j]}
\cdot\frac{[2N]+q^{N}(1+q^{N})[N+2m]}{[2N+2m]} \\
&=&\prod_{i=1}^{N}(1+q^{i})\prod_{j=1}^{m}\genfrac{}{}{}{}{[2N+2j]}{[2j]} \\
&=&\prod_{i=1}^{N}(1+q^{i})\prod_{1\le j\le M}a_{(j,N)},
\end{eqnarray*}
where we apply Eqn.(\ref{eqn:a1}) and Eqn.(\ref{eqn:a2}) to the first line.
This implies Eqn.(\ref{recrel}) holds for $P_{2m,N}$.

When $M=2m-1$, we have 
\begin{eqnarray*}
&&P(2m-1,N-1)+q^{N}P(2m-2,N)+q^{N+2m-1}P(2m,N-2)　\\
&&\qquad=
\prod_{i=1}^{N-2}(1+q^{i})\left\{
(1+q^{N-1})\frac{[N+2m-1]}{[2N+2m-2]}\prod_{j=1}^{m}\frac{[2N+2j-2]}{[2j]}\right. \\
&&\quad\qquad\left. +q^{N}(1+q^{N-1})(1+q^{N})\prod_{j=1}^{m-1}\frac{[2N+2j]}{[2j]} 
+q^{N+2m-1}\prod_{j=1}^{m}\frac{[2N+2j-4]}{[2j]}\right\} \\
&&\qquad=\prod_{i=1}^{N}(1+q^{i})\cdot\prod_{j=1}^{m}\frac{[2N+2j-2]}{[2j]}\cdot
\frac{[N+2m]}{[2N]} \\
&&\qquad=\prod_{i=1}^{N}(1+q^{i})\prod_{1\le j\le M}a_{(j,N)}.
\end{eqnarray*}
This completes the proof.
\end{proof}

\begin{example}
The first few values of $\prod_{1\le j\le M}a_{(j,N)}$ are as follows:
\begin{center}
\begin{tabular}{cc|c|c|c|c|c|}
 &  & \multicolumn{5}{c|}{$N$} \\
 &  & 1 & 2 & 3 & 4 & 5 \\ \hline 
 \multirow{10}{*}{$M$} & 1 & $\displaystyle{\frac{[3]}{[2]}}$ & 
$\displaystyle{\frac{[4]}{[2]}}$ & $\displaystyle{\frac{[5]}{[2]}}$ 
& $\displaystyle{\frac{[6]}{[2]}}$ & $\displaystyle{\frac{[7]}{[2]}}$ \\ \cline{2-7} 
 & 2 & $\displaystyle{\frac{[4]}{[2]}}$ & $\displaystyle{\frac{[6]}{[2]}}$ 
& $\displaystyle{\frac{[8]}{[2]}}$ & $\displaystyle{\frac{[10]}{[2]}}$ 
& $\displaystyle{\frac{[12]}{[2]}}$ \\ \cline{2-7}
 & 3 & $\displaystyle{\frac{[5]}{[2]}}$ & $\displaystyle{\frac{[6]^2}{[2][4]}}$ 
& $\displaystyle{\frac{[7][8]}{[2][4]}}$ & $\displaystyle{\frac{[8][10]}{[2][4]}}$ 
& $\displaystyle{\frac{[9][12]}{[2][4]}}$ \\ \cline{2-7}
 & 4 & $\displaystyle{\frac{[6]}{[2]}}$ & $\displaystyle{\frac{[6][8]}{[2][4]}}$ 
& $\displaystyle{\frac{[8][10]}{[2][4]}}$ & $\displaystyle{\frac{[10][12]}{[2][4]}}$ 
& $\displaystyle{\frac{[12][14]}{[2][4]}}$\\ \cline{2-7} 
& 5 & $\displaystyle{\frac{[7]}{[2]}}$ & $\displaystyle{\frac{[8]^2}{[2][4]}}$ 
& $\displaystyle{\frac{[8][9][10]}{[2][4][6]}}$ & $\displaystyle{\frac{[10]^2[12]}{[2][4][6]}}$ 
& $\displaystyle{\frac{[11][12][14]}{[2][4][6]}}$ \\ \hline
\end{tabular}
\end{center}
\end{example}

\subsection{\texorpdfstring{Factorization of $P^{I}_{\lambda,\ast}$}{Factorization of PI}}
\label{sec:BI-fac}
A {\it valley} of a ballot path $\lambda$ is a middle point of a partial
path $DU$ in $\lambda$. 
We call the lowest and leftmost valley the {\it minimum point} of a path $\lambda$. 
A path $\lambda$ can be written as a concatenation of two paths $\lambda_1$ and 
$\lambda_2$ at the minimum point, {\it i.e.}, $\lambda=\lambda_1\circ\lambda_2$.
Let $M_2$ (resp. $N_2$) be the number of $U$ (resp. $D$) in $\lambda_2$ and 
$N_{1}$ be the number of $D$'s in $\lambda_1$.

\begin{theorem}
\label{thrm:BI-fac}
Let $\lambda, N_1, M_2$ and $N_2$ be as above. 
We have 
\begin{eqnarray*}
P^{I}_{\lambda,\ast}=P_{\lambda_1}^{\mathrm{Dyck}}\cdot P(M_2+N_2,N_1)\cdot P^{I}_{\lambda_2}. 
\end{eqnarray*}
\end{theorem}
Before proceeding with the proof of Theorem \ref{thrm:BI-fac}, 
we introduce some notations and lemmas used later.

We recursively define the set of valleys $\mathrm{Val}(\lambda)$ from left to
right as follows.
The leftmost valley is in $\mathrm{Val}(\lambda)$.
Take a valley $v$. 
If the rightmost valley in $\mathrm{Val}(\lambda)$ which is left to $v$ 
is lower than $v$, then $v\not\in\mathrm{Val}(\lambda)$. 
Otherwise, $v\in\mathrm{Val}(\lambda)$. 
The heights of valleys in $\mathrm{Val}(\lambda)$ are weakly decreasing.
The minimum point is the rightmost valley in $\mathrm{Val}(\lambda)$.

We construct a path $\overline{\lambda}$ from $\lambda$ as follows.
The path $\overline{\lambda}$ is the same as $\lambda$ from the starting 
point to the leftmost valley. 
From the leftmost valley to the rightmost valley in $\mathrm{Val}(\lambda)$, 
the path $\overline{\lambda}$ is the highest path which passes through 
all the valleys in $\mathrm{Val}(\lambda)$.
The path $\overline{\lambda}$ is again the same as $\lambda$ from 
the rightmost valley in $\mathrm{Val}(\lambda)$ to the ending point.

Let $r+1$ be the cardinality of the set $\mathrm{Val}(\lambda)$. 
We denote by $\lambda'_{0}$ the partial path of $\lambda$ from the starting 
point to the leftmost valley.
The path $\lambda$ is written as a concatenation of paths $\lambda'_{i}$,
$0\le i\le r$, and $\lambda_2$ such that $\lambda'_{i}$, $1\le i\le r$, is a partial 
path of $\lambda$ starting from the $i$-th valley in $\mathrm{Val}(\lambda)$
and ending at the $(i+1)$-th valley in $\mathrm{Val}(\lambda)$, 
{\it i.e.}, $\lambda=\lambda'_{0}\circ\lambda'_1\circ\ldots\circ\lambda'_r\circ\lambda_2$.
We denote by $\overline{\lambda_1}$ the path such that 
$\overline{\lambda}=\overline{\lambda_1}\circ\lambda_2$.

Let $n_0$ be the number of $D$ in $\lambda$ from the starting point 
of the path to the leftmost valley in $\mathrm{Val}(\lambda)$.
We denote by $n_i$ (resp. $m_i$) the number of $D$ (resp. $U$) 
in the path $\lambda'_{i}$, $1\le i\le r$.
From the definition of $\mathrm{Val}(\lambda)$, we have 
$n_{i}\ge m_{i}$ for $1\le i\le r$.
Let $s$ be an integer such that $n_{s-1}>m_{s-1}$ and 
$n_{i}=m_{i}$ for $i\ge s$. 
Thus, the path $\overline{\lambda}$ consists of $n_0$ $D$'s,
followed by $m_1$ $U$'s, followed by $n_1$ $D$'s, followed by 
$m_2$ $U$'s, $\cdots$, followed by $n_r$ $D$'s and ending with 
the path $\lambda_2$.

Let $\mu$ be a partial path of $\lambda$ from the starting point 
to $m_{s-1}$ $U$'s. 
We define 
\begin{eqnarray*}
\Box(n_0,m_1,n_1,\ldots,m_{s-1}):=P_{\mu}^{\mathrm{Dyck}},
\end{eqnarray*}
and $P_{T}^{\mathrm{Dyck}}:=P_{\mu}^{\mathrm{Dyck}}$ where the 
Young diagram $T$ is determined by the path $\mu$.
We construct sequences of integers $\mathbf{n}_{i}$, $0\le i\le s-1$ 
from $\mathbf{n}:=(n_0,m_1,n_1,\ldots,m_{s-1})$ as follows.
We define $\mathbf{n}_0:=(n_0-1,m_1,n_1,\ldots,m_{s-1})$, 
$\mathbf{n}_{1}:=(n_0,m_1-1,n_1,m_2,\ldots,m_{s-1})$, 
and $\mathbf{n}_{i}:=(n_0-1,m_1,n_1,\ldots,n_{i-1}+1,m_i-1,\ldots,m_{s-1})$
for $2\le i\le s-1$.

\begin{lemma}
\label{lemma:recDyck}
We have 
\begin{eqnarray*}
\Box(\mathbf{n})=\Box(\mathbf{n}_0)+\sum_{i=0}^{s-2}
q^{\mathrm{deg}_{i}(\mathbf{n})}\Box(\mathbf{n}_{i+1})
\end{eqnarray*}
where $\mathrm{deg}_{i}(\mathbf{n}):=\sum_{0\le j\le i}n_{j}$.
\end{lemma}
\begin{proof}
Recall that a Dyck path $\mu$ determines the Young diagram $Y$.
We denote by $b$ the leftmost box in $Y$. 
There are two cases: the box $b$ is empty or forms a Dyck tile.
In the first case, we are not able to put a Dyck tile on boxes 
on the leftmost column (rotate $Y$ in 45 degrees) of $Y$.
The contribution of this case to $\Box(\mathbf{n})$ is given 
by $\Box(\mathbf{n}_{0})$.
In the second case, we have $(s-1)$ subcases.
The first subcase is the case where the boxes $B_1$ at the bottom row 
form a Dyck tile of length zero. This gives a factor $q^{n_0}$.
We consider a Dyck tiling in the remaining diagram $Y\setminus B_{1}$, 
which gives $\Box(\mathbf{n}_1)$.
Thus, the contribution to $\Box(\mathbf{n}_{0})$ is 
$q^{n_1}\Box(\mathbf{n}_{1})$.
In the $i$-th ($2\le i\le s-1$) subcase, we consider the following 
Dyck tiling.
In the bottom row, we have $n_0-1$ Dyck tiles of length zero 
from left to right. 
We have a Dyck tile of length $m_j$, ($1\le j\le i-1$), which start with 
the rightmost box in the $(1+\sum_{1\le k\le j}m_{k})$-th row from 
the bottom.
If two Dyck tiles of length $m_j$ and $m_{j+1}$ shares the rightmost 
box in the $(1+\sum_{1\le k\le j}m_{k})$-th row from the bottom,
we merge them into a larger Dyck tile of length $m_{j}+m_{j+1}$.
If $n_j>m_j$ and $j<i-1$, we put $n_{j}-m_{j}-1$ Dyck tiles of length zero 
in the $(1+\sum_{1\le k\le j}m_{k})$-th row in-between Dyck tiles 
of length $m_{j}$ and $m_{j+1}$.
If $n_{i-1}>m_{i-1}$, we put $n_{i-1}-m_{i-1}$ Dyck tiles of length zero
right to the Dyck tile of length $m_{i-1}$.
This configuration of Dyck tiles gives the factor $q^{\mathrm{deg}_{i}(\mathbf{n})}$.
In the remaining region in $Y$, we consider a Dyck tiling, which gives 
$\Box(\mathbf{n}_{i})$. 
By summing all the contribution, we obtain the right hand side of 
Lemma \ref{lemma:recDyck}.

\end{proof}

Let $T_{i}$ be the Young diagram determined by $\overline{\lambda}$ and 
$\lambda'_{i}$. 
We denote by $\mathbf{T}:=(T_{1},\ldots,T_{r})$ a sequence of Young diagrams.
Then, the path $\lambda$ is characterized a pair $(\overline{\lambda};\mathbf{T})$.
More generally, let $\mathrm{Val}'(\mu)$ be a subset of the set of valleys and 
$r-1$ be the cardinality of $\mathrm{Val}'(\mu)$.
We denote by $\mu'_{i}$, $1\le i\le r$, be the highest path from the $(i-1)$-th valley 
to the $i$-th valley in $\mathrm{Val}'(\lambda)$ where zeroth valley is the starting point 
of $\mu$ and $(r+1)$-th valley is the ending point of $\mu$.
We denote by $\overline{\mu}'$ the concatenation of paths 
$\mu'_{1}\circ\cdots\circ\mu'_{r}$.
Let $\nu_{i}$ be a path from the $(i-1)$-th valley to the $i$-th valley
such that the first step of $\nu_i$ is $U$. 
We denote by $T'_{i}$ the Young diagram determined by $\mu'_{i}$ 
and $\nu_{i}$.
The path $\mu$ is characterized by the pair $(\overline{\mu}',\mathbf{T}')$.
Note that there are several choices for $(\overline{\mu}',\mathbf{T}')$.
The set $\mathrm{Val}(\lambda)$ is a special choice of $\mathrm{Val}'(\lambda)$.

We define
\begin{eqnarray*}
Q(\overline{\lambda};\mathbf{T}):=P_{\lambda,\ast}^{I},
\end{eqnarray*}
where $\lambda$ is characterized by $(\overline{\lambda};\mathbf{T})$.

We construct paths $\nu_i$, $0\le i\le r+2$, from $\overline{\lambda}$ as follows.
The path $\nu_0$ is the path obtained from $\overline{\lambda}$ by replacing 
$n_0$ with $n_0-1$.
The path $\nu_1$ is the path obtained from $\overline{\lambda}$ by replacing 
$m_1$ with $m_1-1$.
The path $\nu_{i}$, $2\le i\le r$, is the path obtained from $\overline{\lambda}$
by replacing $n_{0}$ with $n_0-1$, $n_{i-1}$ with $n_{i-1}+1$ and 
$m_{i}$ with $m_{i}-1$.
Let $\nu'$ be a path obtained from $\overline{\lambda_{1}}$ by  
replacing $n_0$ with $n_{0}-1$, $n_r$ with $n_{r}+1$ and 
$\nu''$ be a path obtained from $\lambda_2$ by deleting the leftmost 
$U$ in $\lambda_{2}$.
Then, we define the path $\nu_{r+1}:=\nu'\circ\nu''$.
Let $\rho$ be a path obtained from $\overline{\lambda_1}$ by replacing 
$n_0$ with $n_0-1$ and $n_{s-1}$ with $n_{s-1}-1$, and 
$\rho'$ be a path $U\lambda_2$.
Then, we define $\nu_{r+2}:=\rho\circ\rho'$.

\begin{lemma}
\label{lemma:recBaleven}
Let $\lambda$, $\nu$ and $\mathbf{T}$ as above and $M_2+N_2$ be even.
We have 
\begin{eqnarray*}
P^{I}_{\lambda,\ast}=
Q(\nu_0;\mathbf{T})+\sum_{i=0}^{r}q^{\mathrm{deg}_{i}(\lambda)}Q(\nu_{i+1};\mathbf{T}),
\end{eqnarray*}
where $\mathrm{deg}_{i}(\lambda):=\sum_{0\le j\le i}n_{j}$ for $0\le i\le r$.
\end{lemma}

\begin{proof}

By a similar argument to Lemma \ref{lemma:recDyck}, it is easy to show 
that Lemma holds true when all $T_{i}=\emptyset$ for $1\le i\le r$ (compare 
paths $v_{i}$ with $\mathbf{n}_{i}$).

Below, we consider the case where at least one $T_{k}\neq\emptyset$. 
The skew shape $\overline{\lambda}/\nu_{i}$ contains only Dyck tiles of length zero 
and Dyck tiles of length $2m_{j}$ with $1\le j\le i-1$. 
Let $\nu'$ be a path characterized by the pair $(\nu_{i},\mathbf{T})$.
If $T_{k}\neq\emptyset$ with $1\le k\le i-1$, we deform the Dyck tile of length 
$2m_{k}$ into the Dyck tile $D$ of length $2m_k$ such that the shape of $D$ fits 
$\nu'$. 
The weight of Dyck tiles is $q^{\mathrm{deg}_{i}(\lambda)}$.
If a ballot tiling contains a Dyck tile of length $2n$ with $n<m_k$ over $\lambda$, 
such a configuration is included in the calculation of 
$q^{\mathrm{deg}_{l}(\lambda)}Q(\nu_{l+1};\mathbf{T})$ with some $l<k$.
In the remaining region above $\nu'$, we can put a ballot tiling without any constrains.
Especially, since $M_{2}+N_{2}$ is even, the box at the minimum point are not 
occupied by a ballot tile of length $(2n,n')$ with $n'\ge1$.
These imply that we have the generating function $Q(\nu_{i},\mathbf{T})$ 
in the shape $\nu'$.

The sum of all contributions (the right hand side of Lemma \ref{lemma:recBaleven}) 
is equal to the left hand side of Lemma \ref{lemma:recBaleven}.

\end{proof}

\begin{lemma}
\label{lemma:recBalodd}
Let $\lambda$, $\nu$ and $\mathbf{T}$ as above and $M_2+N_2$ be odd.
We have 
\begin{eqnarray*}
P^{I}_{\lambda,\ast}=
Q(\nu_0;\mathbf{T})+\sum_{i=0}^{r}q^{\mathrm{deg}_{i}(\lambda)}Q(\nu_{i+1};\mathbf{T})
+q^{\mathrm{deg}_{r+1}(\lambda)}Q(\nu_{r+2};\mathbf{T}),
\end{eqnarray*}
where $\mathrm{deg}_{i}(\lambda):=\sum_{0\le j\le i}n_{j}$ for $0\le i\le r$
and $\mathrm{deg}_{r+1}(\lambda):=\sum_{i=0}^{r}n_{i}+M_2+\sum_{i=s}^{r}n_i$.
\end{lemma}

\begin{proof}
The proof of Lemma is similar to the one of Lemma \ref{lemma:recBaleven}.
We have the first two terms in the right hand side of Lemma \ref{lemma:recBalodd} 
as a contribution to $P_{\lambda,\ast}^{I}$.

The difference is that the box at the minimum point can be occupied by 
a ballot tile $B_1$ of length $(2n,n')$ with $n'\ge1$. 
From the definition of ballot tilings of type BI, we have another 
ballot tile $B_2$ over $B_1$.
Note that the shapes of $B_1$ and $B_2$ are the same. 
Let $S$ be a valley such that $S$ is the leftmost valley which 
is the same height as the minimum point.
Then, we consider the configuration such that the leftmost box
of $B_1$ is at the point $S$.
The length of $B_1$ is $(\sum_{i=s}^{r}n_i,M_2)$.
In the remaining shape $(\lambda\setminus(B_1\cup B_2))/\nu_{r+2}$,
we put Dyck tiles like in the case of the term 
$q^{\mathrm{deg}_{r}(\lambda)}Q(\nu_{r+1};\mathbf{T})$.
Thus the contribution is given by $q^{\mathrm{deg}_{r+1}(\lambda)}Q(\nu_{r+2};\mathbf{T})$.
This completes the proof.
\end{proof}

\begin{lemma}
\label{lemma:PBoxeven}
Let $N=n_1+n_2$. 
We have
\begin{eqnarray*}
P(n_3,N)\Box(n_1,n_2-1)+q^{n_2}P(n_3-1,N)\Box(n_1-1,n_2)
=P(2n_2+n_3-1,n_1)\cdot P(n_3,n_2).
\end{eqnarray*}
\end{lemma}
\begin{proof}
We consider the case where $n_3$ is even, since the proof for the odd case 
is essentially the same.
Since we have 
\begin{eqnarray*} 
\Box(r,s)&=&\genfrac{[}{]}{0pt}{}{r+s}{r}, \\
P(n_3-1,N)&=&\frac{[2N+n_3]}{[N+n_3]}P(n_3,N),
\end{eqnarray*}
the left hand side of Lemma \ref{lemma:PBoxeven} is 
\begin{eqnarray}
\label{PBoxeven1}
\qquad
P(n_3,N)\Box(n_1,n_2-1)+q^{n_2}P(n_3-1,N)\Box(n_1-1,n_2)
=\genfrac{[}{]}{0pt}{}{N}{n_1}\frac{[N+n_2+n_3]}{[2N+n_3]}P(n_3,N).
\end{eqnarray}
We show that the right hand side of Eqn.(\ref{PBoxeven1}) is 
equal to the right hand side of Lemma \ref{lemma:PBoxeven} by 
induction on $n_1$.
When $n_{1}=0$, the both terms are equal to $P(n_3,n_2)$.
Suppose that the statement holds true when $n_1=n$.
When $n_1=n+1$, the right hand side of Eqn.(\ref{PBoxeven1}) 
is 
\begin{align}
\begin{split}
\nonumber 
&\genfrac{[}{]}{0pt}{}{N+1}{n_1+1}\frac{[N+n_2+n_3+1]}{[2N+n_3+2]}P(n_3,N+1) \\
&\qquad\quad=\frac{[n_2+1][N+n_2+n_3+1]}{[n_1+1][N+n_2+n_3+2]}\cdot
\genfrac{[}{]}{0pt}{}{N+1}{n_1}\frac{[N+n_2+n_3+2]}{[2N+n_3+2]}P(n_3,N+1)
\end{split} \\
&\qquad\quad=\frac{[n_2+1][N+n_2+n_3+1]}{[n_1+1][N+n_2+n_3+2]}
P(2n_2+n_3+1,n_1)P(n_3,n_2+1),
\label{Pbox1}
\end{align}
where we have used the assumption for $n_1=n$ and $n_2+1$.
Since $n_3$ is even, we have 
\begin{align}
\label{Prec1}
P(n_3,n_2+1)&=\frac{[2n_2+n_3+2]}{[n_2+1]}P(n_3,n_2), \\
\begin{split}
\label{Prec2}
P(2n_2+n_3+1,n_1)&=
\frac{[N+n_2+n_3+2][n_1+1]}{[N+n_2+n_3+3][2N+n_3+2]}P(2n_2+n_3+1,n_1+1) \\
&=\frac{[N+n_2+n_3+2][n_1+1]}{[2n_2+n_3+2][N+n_2+n_3+1]}P(2n_2+n_3-1,n_1+1)
\end{split}
\end{align}
Substituting Eqn.(\ref{Prec1}) and (\ref{Prec2}) into the right hand side 
of Eqn.(\ref{Pbox1}), 
we obtain $P(2n_2+n_3-1,n_1+1)\cdot P(n_3,n_2)$. 
This completes the proof.
\end{proof}

\begin{proof}[Proof of Theorem \ref{thrm:BI-fac}]
We prove Theorem by induction.
We first consider the case where $M_2+N_2$ is even.

Suppose that Theorem holds true up to $|\lambda|-1$ and $M_2+N_2$ is even.
We calculate the right hand side of Lemma \ref{lemma:recBaleven}
with the assumption.
From the assumption, we have 
\begin{align*}
\begin{split}
Q(\nu_{0};\mathbf{T})&=P_{\lambda_{2},\ast}^{I}\cdot \Box(\mathbf{n}_{0})\cdot
\prod_{j=1}^{r}P_{T_j}^{\mathrm{Dyck}} \cdot 
P(l_s,l'_{s-1}-1)\cdot\prod_{j=s+1}^{r+1}P(l_{j},n_{j-1})\cdot,
\end{split}\\ 
\begin{split}
Q(\nu_{i};\mathbf{T})&=
P_{\lambda_2,\ast}^{I}\cdot \Box(\mathbf{n}_i)\cdot
\prod_{j=1}^{r}P_{T_j}^{\mathrm{Dyck}} 
\cdot P(l_{s},l'_{s-1})\cdot \prod_{j=s+1}^{r+1}P(l_j,n_{j-1}),\qquad \text{for\ } 1\le i\le s-1, 
\end{split}　\\
\begin{split}
Q(\nu_{i};\mathbf{T})&=
P_{\lambda_2,\ast}^{I}\cdot \Box(n_0-1,m_1,n_1,\ldots,n_{i-1},n_{i-1}+1,m_{i}-1)\cdot
\prod_{j=1}^{r}P_{T_j}^{\mathrm{Dyck}}  \\
&\quad\times P(l_{i+1},l'_{i})\cdot \prod_{j=i+2}^{r+1}P(l_j,n_{j-1}),\qquad \text{for\ } s\le i\le r, 
\end{split}　\\
\begin{split}
Q(\nu_{r+1};\mathbf{T})&=
P_{\lambda_2,\ast}^{I}\cdot P(M_2+N_2-1,l'_{r})\cdot 
\Box(n_0-1,m_1,n_1,\ldots,n_{r-1},m_{r})\cdot
\prod_{i=1}^{r}P_{T_i}^{\mathrm{Dyck}}
\end{split}
\end{align*}
where $l_i:=M_2+N_2+\sum_{j=i}^{r}(m_j+n_j)$ and $l'_{i}:=\sum_{k=0}^{i}n_{k}$.
We calculate the right hand side of Lemma \ref{lemma:recBaleven}. 
Since $\mathrm{deg}_{i}(\lambda)=l'_{i}$, we have 
\begin{align}
\label{facBal1}
\begin{split}
\sum_{i=s-1}^{r}q^{l'_i}Q(\nu_{i+1};\mathbf{T})
=q^{l'_{s-1}}P_{\lambda_2,\ast}^{I}\cdot \Box(\mathbf{n}_{0})
\cdot P(l_{s}-1,l'_{s-1})\cdot
\prod_{j=s}^{r}P(l'_{j+1},n_{j})
\cdot \prod_{i=1}^{r}P_{T_{i}}^{\mathrm{Dyck}}
\end{split}
\end{align} 
where we have used Lemma \ref{lemma:PBoxeven} and 
\begin{eqnarray*}
\Box(n_{0}-1,m_1,\ldots,n_{i-1}+1,m_{i}-1)
=\genfrac{[}{]}{0pt}{}{\sum_{j=0}^{i-1}n_{j}+m_{i}-1}{m_{i}-1}\Box(n_{0}-1,m_1,\ldots,n_{i-2},m_{i-1}).
\end{eqnarray*}
By applying Lemma \ref{lemma:recr} to the sum of Eqn.(\ref{facBal1}) and $Q(\nu_0;\mathbf{T})$,
we have   
\begin{eqnarray}
\label{facBal2}
\qquad
Q(\nu_0;\mathbf{T})+\sum_{i=s-1}^{r}q^{l'_{i}}Q(\nu_{i+1};\mathbf{T})
=
P_{\lambda_2,\ast}^{I}\cdot \Box(\mathbf{n}_0)\cdot P(l_s,l'_{s-1})\cdot 
\prod_{j=1}^{r}P_{T_{j}}^{\mathrm{Dyck}}\cdot \prod_{j=s}^{r}P(l_{j+1},n_j).
\end{eqnarray}
By applying Lemma \ref{lemma:recDyck} to the sum of Eqn.(\ref{facBal2}) and 
$\sum_{i=0}^{s-2}q^{l'_{i}}Q(\nu_{i+1};\mathbf{T})$, 
we obtain
\begin{eqnarray*}
Q(\nu_0;\mathbf{T})+\sum_{i=0}^{r}q^{\mathrm{deg}_{i}(\lambda)}Q(\nu_{i+1};\mathbf{T})
&=&
\Box(\mathbf{n})\cdot \prod_{j=1}^{r}P_{T_j}^{\mathrm{Dyck}}\cdot
P(l_s,l'_{s-1})\cdot \prod_{j=s}^{r}P(l_{j+1},n_{j}) \\
&=&P_{\lambda_2,\ast}^{I}\cdot P(l_s,l'_{s-1})\cdot \prod_{j=s}^{r}P(l_{j+1},n_{j})
\cdot P_{\lambda_1}^{\mathrm{Dyck}}.
\end{eqnarray*}
Thus the statement holds true when the shape is $\lambda$.  

We consider the case where $M_2+N_2$ is odd. 
From the induction assumption, we have 
\begin{eqnarray*}
Q(\nu_{r+2};\mathbf{T})
=P_{\lambda_2,\ast}^{I}\cdot P(l_s+1,l'_{s-1}-2)\cdot\prod_{j=s}^{r}P(l_{j+1},n_{j})
\cdot \Box(n_{0}-1,m_1,n_1,\ldots,m_{r})\cdot \prod_{j=1}^{r}P_{T_j}^{\mathrm{Dyck}}.
\end{eqnarray*}
The rest of the proof is essentially the same as in the case of $M_2+N_2$ even except 
that we use Lemma \ref{lemma:recBalodd} instead of Lemma \ref{lemma:recBaleven}.
This completes the proof.
\end{proof}

\subsection{Factorization in terms of trees}
\label{sec:Fac-trees}
Let $T$ be a tree corresponding to a ballot path $\lambda$.
Theorem \ref{thrm:BI-fac} can be translated into the following 
operations on a partial tree:
\begin{eqnarray*}
\tikzpic{-0.5}{
\draw(0,0)--(-0.3,-0.3)(-0.7,-0.7)--(-1,-1);
\draw(-0.2,-0.2)node{\rotatebox{-45}{$-$}}(-0.8,-0.8)node{\rotatebox{-45}{$-$}};
\draw[dashed](-0.3,-0.3)--(-0.7,-0.7);
\draw[decoration={brace,mirror,raise=5pt},decorate]
  (0,0) --(-1,-1);
\draw(-0.2,-0.2)node[left=9pt]{$N$};
\draw(0,0)--(0.3,-0.3)(0.7,-0.7)--(1,-1);
\draw[dashed](0.3,-0.3)--(0.7,-0.7);
\draw(0.2,-0.2)node{\rotatebox{45}{$-$}}(0.8,-0.8)node{\rotatebox{45}{$-$}};
\draw[decoration={brace,mirror,raise=5pt},decorate]
  (1,-1)--(0,0);
\draw(0.2,-0.2)node[right=9pt]{$M$};
}
&\mapsto&\genfrac{[}{]}{0pt}{}{M+N}{M}\cdot
\tikzpic{-0.4}{
\draw(0,0)--(0,-0.3)(0,-0.7)--(0,-1);
\draw(0,-0.2)node{$-$}(0,-0.8)node{$-$};
\draw[dashed](0,-0.3)--(0,-0.7);
\draw[decoration={brace,mirror,raise=5pt},decorate]
    (0,-1)-- (0,0);
\draw(0,-0.5)node[right=9pt]{$M+N$};
} \\
\tikzpic{-0.5}{
\draw(0,0)--(-0.3,-0.3)(-0.7,-0.7)--(-1,-1);
\draw(-0.2,-0.2)node{\rotatebox{-45}{$-$}}(-0.8,-0.8)node{\rotatebox{-45}{$-$}};
\draw[dashed](-0.3,-0.3)--(-0.7,-0.7);
\draw[decoration={brace,mirror,raise=5pt},decorate]
  (0,0) --(-1,-1);
\draw(-0.2,-0.2)node[left=9pt]{$N$};
\draw(0,0)--(0.35,-0.35)(0.65,-0.65)--(1,-1);
\draw[dashed](0.35,-0.35)--(0.65,-0.65);
\draw(0.3,-0.3)node{\rotatebox{45}{$-$}}(0.7,-0.7)node{\rotatebox{45}{$-$}};
\draw(0.15,-0.15)node{$\bullet$}(0.85,-0.85)node{$\bullet$};
\draw[decoration={brace,mirror,raise=5pt},decorate]
  (1,-1)--(0,0);
\draw(0.2,-0.2)node[right=9pt]{$M$};
}&\mapsto&
\genfrac{[}{]}{0pt}{}{M+N}{M}_{q^{2}}\prod_{i=1}^{N}(1+q^{i})\cdot
\tikzpic{-0.4}{
\draw(0,0)--(0,-0.35)(0,-0.65)--(0,-1);
\draw(0,-0.3)node{$-$}(0,-0.7)node{$-$};
\draw[dashed](0,-0.4)--(0,-0.6);
\draw(0,-0.15)node{$\bullet$}(0,-0.85)node{$\bullet$};
\draw[decoration={brace,mirror,raise=5pt},decorate]
    (0,-1)-- (0,0);
\draw(0,-0.5)node[right=9pt]{$M+N$};
} \\
\tikzpic{-0.5}{
\draw(0,0)--(-0.45,-0.45)(-0.7,-0.7)--(-1,-1);
\draw(-0.4,-0.4)node{\rotatebox{-45}{$-$}}(-0.8,-0.8)node{\rotatebox{-45}{$-$}};
\draw[dashed](-0.3,-0.3)--(-0.7,-0.7);
\draw[decoration={brace,mirror,raise=5pt},decorate]
  (0,0) --(-1,-1);
\draw(-0.2,-0.2)node[left=9pt]{$N$};
\draw(0,0)--(0.45,-0.45)(0.65,-0.65)--(1,-1);
\draw[dashed](0.45,-0.45)--(0.65,-0.65);
\draw(0.4,-0.4)node{\rotatebox{45}{$-$}}(0.7,-0.7)node{\rotatebox{45}{$-$}};
\draw(0.15,-0.15)node{$\bullet$}(0.85,-0.85)node{$\bullet$};
\draw[decoration={brace,mirror,raise=5pt},decorate]
  (1,-1)--(0,0);
\draw(0.2,-0.2)node[right=9pt]{$M$};
\draw[latex-,dashed](-0.3,-0.3)--(0.3,-0.3);
}&\mapsto&
\genfrac{[}{]}{0pt}{}{M+N}{M}_{q^{2}}
\frac{[2M+N]}{[2(M+N)]}\prod_{i=1}^{N}(1+q^{i}) \cdot
\tikzpic{-0.4}{
\draw(0,0)--(0,-0.35)(0,-0.65)--(0,-1);
\draw(0,-0.3)node{$-$}(0,-0.7)node{$-$};
\draw[dashed](0,-0.4)--(0,-0.6);
\draw(0,-0.15)node{$\bullet$}(0,-0.85)node{$\bullet$};
\draw[decoration={brace,mirror,raise=5pt},decorate]
    (0,-1)-- (0,0);
\draw(0,-0.5)node[right=9pt]{$M+N$};
\draw[latex-,dashed](-0.4,-0.15)node{$($}--(-0.1,-0.15)node{$)$};
}
\end{eqnarray*}
Here $(\leftarrow)$ in the right hand side of the third operation means that 
if the leftmost top edge in the left hand side of the third operation has an outgoing arrow,
we put an outgoing arrow on the top edge of the right hand side.

Then, we define operations on the following trees (not a partial tree):
\begin{eqnarray*}
\tikzpic{-0.4}{
\draw(0,0)--(0,-0.3)(0,-0.7)--(0,-1);
\draw(0,-0.2)node{$-$}(0,-0.8)node{$-$};
\draw[dashed](0,-0.3)--(0,-0.7);
\draw[decoration={brace,mirror,raise=5pt},decorate]
    (0,-1)-- (0,0);
\draw(0,-0.5)node[right=9pt]{$N$};
}&\mapsto&\prod_{i=1}^{N}(1+q^{i}), \\
\tikzpic{-0.4}{
\draw(0,0)--(0,-0.35)(0,-0.65)--(0,-1);
\draw(0,-0.3)node{$-$}(0,-0.7)node{$-$};
\draw[dashed](0,-0.4)--(0,-0.6);
\draw(0,-0.15)node{$\bullet$}(0,-0.85)node{$\bullet$};
\draw[decoration={brace,mirror,raise=5pt},decorate]
    (0,-1)-- (0,0);
\draw(0,-0.5)node[right=9pt]{$N$};
}&\mapsto& 1.
\end{eqnarray*}
Then, we have a map from a tree $T$ to $\mathbb{Z}[q,q^{-1}]$ by 
successive applications of the operations defined above.

\section{Various expressions of the generating function}
\label{sec:varigf}
In this section, we show various expressions of the generating function 
$P_{\lambda,\ast}^{I}$ for a ballot path $\lambda$.
Given a tree $T:=A(\lambda)$, we abbreviate $P_{\lambda,\ast}^{I}$ by $T$.
It is clear from the context whether $T$ stands for a tree or a generating 
function.
We denote by $|T|$ the number of edges in $T$.

The factorization of the generating function (Theorem \ref{thrm:BI-fac})
can be translated into the following operation on a tree.
Suppose that the root of a tree $T:=A(\lambda)$ has $p$ edges and the leftmost
edge $e_1$ has a partial tree $T_{1}$ ($T_1$ can be an empty tree).
We denote by $T_2:=T\setminus(T_{1}\cup e_{1})$ the partial tree obtained from $T$ by deleting
the partial tree $(T_{1}\cup e_{1})$.
\begin{lemma}
\label{lemma:fac-tree}
If the edge $e_1$ does not have an incoming arrow, the generating function $T$ satisfies
\begin{eqnarray*}
T=P^{\mathrm{Dyck}}_{T_{1}}\cdot P(2|T_{2}|,|T_{1}|+1)\cdot T_{2}.
\end{eqnarray*}
Similarly, if $e_{1}$ has an incoming arrow, the generating function $T$ satisfies
\begin{eqnarray*}
T=P^{\mathrm{Dyck}}_{T_{1}}\cdot P(2|T_{2}|-1,|T_{1}|+1)\cdot T_{2}.
\end{eqnarray*}
\end{lemma}

\subsection{\texorpdfstring{Expansion by $\mathrm{tree}_{1}$}{Expansion by tree1}}
\label{sec:ext1}
Lemma \ref{lemma:recBaleven} and \ref{lemma:recBalodd} can be translated 
into the following expressions in terms of trees.

Let $T$ be a tree $A(\lambda)$ and $e_1$ be an edge which is the leftmost
and lowest in $T$.
We denote by $\mathrm{tree}_{1}(T,e_1)$ the tree obtained from $T$ by deleting $e_1$.
We have a unique sequence of edges  from the edge $e_1$ to the root.
Along the sequence of edges, we enumerate the ramification point 
from $e_1$ to the root by $1,\ldots,M$ where $M$ is the number assigned 
to the root.
We denote by $r_{i,j}$ the $j$-th edge from left in the $i$-th ramification 
point and by $r_{M,\bullet}$ an edge with a dot connecting to the root.
Let $T'_{i,j}$ be the partial tree in $T$ such that the $i$-th ramification 
point is the root of $T'_{i,j}$ and edges in $T'_{i,j}$ are left to the edge $r_{i,j+1}$.
We denote by $T_{i,j}$ the tree obtained from $T'_{i,j}$ by deleting the edge $e_1$ 
and adding an edge above the root of $T'_{i,j}$.
We define $T\setminus T'_{i,j}$ as a tree obtained from $T$ by deleting the partial
tree $T'_{i,j}$.
A tree $\mathrm{tree}_{1}(T,r_{i,j})$ with $(i,j)\neq(M,\bullet)$ is defined as a tree obtained 
from $T\setminus T'_{i,j}$ and $T_{i,j}$ by putting $T_{i,j}$ below the edge $r_{i,j}$ of 
$T\setminus T_{i,j}$ from left. 
A tree $\mathrm{tree}_{1}(T,r_{M,\bullet})$ is defined as a tree obtained  by concatenating 
$T\setminus T'_{M,\bullet}$ and $T_{M,\bullet}$ at the root and putting an arrow from the edge 
$r_{M,\bullet}$ to the unique edge of $T_{M,\bullet}$ connected to the root.
We define $\mathrm{deg}_{1}(T,r_{i,j})$ as the number of edges in $T'_{i,j}$, {\it i.e.},
$\mathrm{deg}_{1}(T,r_{i,j})=|T'_{i,j}|$.

When $T$ does not have arrows on the edges connecting to the root, we have
\begin{eqnarray*}
T=\mathrm{tree}_{1}(T,e_1)+\sum_{i,j}q^{\mathrm{deg}_{1}(T,r_{i,j})}\mathrm{tree}_{1}(T,r_{i,j}).
\end{eqnarray*}

Suppose that the $r_{M,j}$-th ($1\le j\le p$) edge in $T$ does not have an arrow
and the $r_{M,j}$-th ($p+1\le j\le r$) edge in $T$ has an incoming arrow.
The $r_{M,\bullet}$-th edge with a dot has an outgoing arrow.
A tree $\mathrm{tree}_{1}(T,r_{M,p})$ is defined as a tree obtained 
by concatenating $T_{M,p}$ and $T\setminus T'_{M,p}$ at the root and 
putting an incoming arrow on the unique edge of $T_{M,p}$ connected 
to the root.
We define 
\begin{eqnarray*}
S_{p}:=\{(i,j) \ |\ i\le M-1 \}\cup \{(M,j) \ |\  j\le p \},
\end{eqnarray*}
and $\mathrm{deg}_{1}(T,r_{M,p}):=|T'_{M,p}|$. 
We have 
\begin{eqnarray*}
T=\mathrm{tree}_{1}(T,e_1)+\sum_{(i,j)\in S_{p}}q^{\mathrm{deg}_{1}(T,r_{i,j})}
\mathrm{tree}_{1}(T,r_{i,j}).
\end{eqnarray*} 

Suppose that all the edges connecting to the root have arrows in $T$.
We define $r_{i,j}$, $\mathrm{tree}_{1}(T,r_{i,j})$ with $i\le M$ as above.
Let (S2) be the following statement for $T$:
\begin{enumerate}
\item[(S2)] The depth of the edge $e_1$ is more than or equal to two, and 
the leftmost edge connected to the root has an incoming arrow.
\end{enumerate}
Here, the depth of an edge $e$ in $T$ is defined as the distance from $e$ 
to the root of the tree.
When $T$ satisfies the statement (S2), we define a tree 
$\mathrm{tree}_{1}(T,r_{M+1})$ as follows.
Let $T'$ be a tree obtained from $T_{M,1}$ by deleting two 
successive edges connected to the root.
The tree $\mathrm{tree}_{1}(T,r_{M+1})$ is obtained by a concatenation
of $T'$ and $T\setminus T'_{M,1}$ at the root. 
We do not put an incoming arrow to the edge of $T'$ connected to the root.
We define 
\begin{eqnarray*}
\mathrm{deg}_{1}(T,r_{M+1}):=|T'_{M,1}|+2|T\setminus T'_{M,1}|-1.
\end{eqnarray*}
Then, we have
\begin{eqnarray*}
T=\mathrm{tree}_{1}(T,e_1)
+\sum_{i,j}q^{\mathrm{deg}_{1}(T,r_{i,j})}\mathrm{tree}_{1}(T,r_{i,j})
+\delta_{(S2)}\cdot q^{\mathrm{deg}_{1}(T,r_{M+1})}\mathrm{tree}_{1}(T,r_{M+1}),
\end{eqnarray*}
where $\delta_{(S2)}=1$ if (S2) is true and $\delta_{(S2)}=0$ otherwise.

\subsection{\texorpdfstring{Expansion by $\mathrm{tree}_{2}$ and 
$\mathrm{tree}_{3}$}{Expansion by tree2 and tree3}}
\label{sec:ext23}
Let $r_{i}$, $1\le i\le p+r-1$, be the $i$-th edge without a $\bullet$ from left 
which is connected to the root of a tree $T$ and $r_{p+r}:=r_{\bullet}$ be the edge 
with a $\bullet$ which is connected to the root.
We consider the tree $T$ where the edges $r_{i}$ with $1\le i\le p$ do not have 
incoming arrows and the edges $r_{i}$ with $p+1\le i<p+r$ have incoming arrows.
We denote by $T_{i}$ a partial tree connected to the edge $r_{i}$ for $1\le i\le p+r$.
We define $\mathrm{tree}_{2}(T,r_{i})$ as a tree obtained from $T$ by deleting the edge $r_{i}$
and connecting two partial trees at the root in-between the $r_{i-1}$-th edge and $r_{i+1}$-th
edge. We do not put incoming arrows on the edges left to the edge $r_{i+1}$.
Let $E(r_{i})$ be the edges left to the $r_{i}$-th edge.
We define 
\begin{eqnarray*}
\mathrm{deg}_{2}(T,r_{i}):=
\begin{cases}
\#\{e \ |\ e\in E(r_{i}) \} & 1\le i\le p, \\
2\#\{e \ |\  e\in E(r_{p+1})\}
+\#\{e \ |\  e\in (E(r_{i})\setminus E(r_{p+1}))\}  &  \text{$p+1\le i\le p+r$}.
\end{cases}
\end{eqnarray*}

We introduce two lemmas used later.
\begin{lemma}
\label{lemma:PP}
Let $\lambda_1$ and $\lambda_{2}$ be Dyck paths. 
We denote by $\lambda$ the path $U\lambda_1DU\lambda_2D$ and 
by $|\mu|$ the number of edges for a tree $A(\mu)$.
Then, we have 
\begin{align}
\label{eq-QQ}
\begin{split}
P(2(N+|\lambda_2|+1),|\lambda_1|+1)\cdot P(2N,|\lambda_2|+1)\cdot 
P^{\mathrm{Dyck}}_{\lambda_1}\cdot P^{\mathrm{Dyck}}_{\lambda_2}
=
\cdot P(2N,|\lambda|)\cdot P^{\mathrm{Dyck}}_{\lambda}.
\end{split}
\end{align}
\end{lemma}

\begin{proof}
We have
\begin{eqnarray*}
P(2N,|\lambda|)&=&\prod_{j=1}^{|\lambda_1|+1}
\frac{[2N+2(|\lambda|+1-j)]}{[|\lambda|+1-j]}\cdot P(2N,|\lambda_2|+1), \\
P(2(N+|\lambda_2|+1),|\lambda_{1}|+1)
&=&\prod_{j=1}^{|\lambda_{1}|+1}\frac{[2N+2(|\lambda|+1-j)]}{[|\lambda_1|+2-j]},
\end{eqnarray*}
from Proposition \ref{prop:BalP} and  
\begin{eqnarray*}
P^{\mathrm{Dyck}}_{\lambda}=\genfrac{[}{]}{0pt}{}{|\lambda|}{|\lambda_1|+1}
P^{\mathrm{Dyck}}_{\lambda_1}P^{\mathrm{Dyck}}_{\lambda_2}
\end{eqnarray*}
from Lemma \ref{lemma:treedel} and Lemma \ref{lemma:treefac}.
Substituting these into the left hand side of Eqn.(\ref{eq-QQ}), we obtain the 
right hand side of Eqn.(\ref{eq-QQ}).
\end{proof}

\begin{lemma}
\label{lemma:tree2}
We have 
\begin{align}
\label{eq:tree21}
\frac{\mathrm{tree}_{2}(T,r_{i})}{T}&=\frac{[|T_{i}|+1]}{[2(\sum_{1\le k\le p+r}|T_{k}|+p+r)]}, \quad 1\le i\le p,\\
\begin{split}
\label{eq:tree22}
\frac{\mathrm{tree}_{2}(T,r_{i})}{T}&=
\frac{[2(\sum_{k=p+1}^{p+r}|T_{k}|+r)]}{[2(\sum_{k=1}^{p+r}|T_{k}|+p+r)]}
\cdot \frac{[|T_{i}|+1]}{[2(\sum_{k=i+1}^{p+r}|T_{k}|+p+r-i)+|T_{i}|+1]} \\
&\quad\times\prod_{p+1\le j\le i-1}
\frac{[2(\sum_{k=j+1}^{p+r}|T_{k}|+p+r-j)]}{[2(\sum_{k=j+1}^{p+r}|T_{k}|+p+r-j)+|T_{j}|+1]}, 
\quad\text{for $p+1\le i\le p+r$}.
\end{split}
\end{align}
\end{lemma}
\begin{proof}
We prove Lemma for $1\le i\le p$ since one can apply the similar computation to 
$p+1\le i\le p+r$ case.

Let $r_{i,j}$ $1\le i\le s$, be the $j$-th edge without a $\bullet$ from left which 
is connected to the root of a partial tree $T_{i}$.
We denote by $T_{i,j}$ a partial tree connected to the edge $r_{i,j}$.

From the factorization (Lemma \ref{lemma:fac-tree}), for $1\le i\le p$ we have 
\begin{eqnarray*}
T&=&\prod_{j=1}^{i}P^{\mathrm{Dyck}}_{T_j}
\cdot\prod_{j=1}^{i}P\left(2\left(\sum_{k=j+1}^{p+r}|T_{k}|+p+r-j\right), |T_{j}|+1\right) 
\cdot (T\setminus \bigcup_{1\le j\le i} T_{j})\\
\mathrm{tree}_{2}(T,r_{i})
&=&\prod_{j=1}^{i-1}P^{\mathrm{Dyck}}_{T_j}
\cdot \prod_{j}P^{\mathrm{Dyck}}_{T_{i,j}} \cdot (T\setminus \bigcup_{1\le j\le i-1} T_{j})\\
&&\times \prod_{j=1}^{s}
P\left(2\left(\sum_{j+1\le k\le s}|T_{i,k}|+s-j+\sum_{k=i+1}^{p+r}|T_{k}|+p+r-i\right),|T_{i,j}|+1\right)
\\
&&\times\prod_{j=1}^{i-1}P\left(2\left(\sum_{k=j+1}^{p+r}|T_{k}|+p+r-j-1\right), |T_{j}|+1\right), \\
\end{eqnarray*}
where $P(M,N)$ is defined in Section \ref{sec:BalP}.
We have 
\begin{align*}
\begin{split}
\prod_{j}P^{\mathrm{Dyck}}_{T_{i,j}}\cdot\prod_{j=1}^{s}
P\left(2\left(\sum_{j+1\le k\le s}|T_{i,k}|+s-j+\sum_{k=i+1}^{p+r}|T_{k}|+p+r-i\right),|T_{i,j}|+1\right) \\
=P\left(2\left(\sum_{k=i+1}^{p+r}|T_{k}|+p+r-i\right), |T_{i}|\right) 
\end{split}
\end{align*}
from the successive use of Lemma \ref{lemma:PP}. We also have
\begin{eqnarray*}
\frac{P\left(2\left(\sum_{k=j+1}^{p+r}|T_{k}|+p+r-j-1\right), |T_{j}|+1\right)}
{P\left(2\left(\sum_{k=j+1}^{p+r}|T_{k}|+p+r-j\right), |T_{j}|+1\right)}
&=&\frac{\left[2\left(\sum_{k=j+1}^{p+r}|T_{k}|+p+r-j\right)\right]}
{\left[2\left(\sum_{k=j}^{p+r}|T_{k}|+p+r-j+1\right)\right]}, \\
\frac{P\left(2\left(\sum_{k=i+1}^{p+r}|T_{k}|+p+r-i\right), |T_{i}|\right)}
{P\left(2\left(\sum_{k=i+1}^{p+r}|T_{k}|+p+r-i\right), |T_{i}|+1\right)}
&=&\frac{\left[|T_{i}|+1\right]}{\left[2\left(\sum_{k=i}^{p+r}|T_{k}|+p+r+1-i\right)\right]}.
\end{eqnarray*}
Substituting these into the left hand side of Eqn.(\ref{eq:tree21}), 
we obtain the right hand side of Eqn.(\ref{eq:tree21}).

\end{proof}

\begin{prop}
\label{lemma:recT2}
We have 
\begin{align}
\label{eq:recT2}
T=(1+q^{\mathrm{deg}_{2}(T,r_{p+1})})
\sum_{1\le i\le p}q^{\mathrm{deg}_{2}(T,r_{i})}\mathrm{tree}_{2}(T,r_{i})
+\sum_{p+1\le i\le p+r}q^{\mathrm{deg}_{2}(T,r_{i})}\mathrm{tree}_{2}(T,r_{i}) 
\end{align}
\end{prop}
\begin{proof}
From Lemma \ref{lemma:tree2}, we have
\begin{align*}
\begin{gathered}
(1+q^{\mathrm{deg}_{2}(T,r_{p+1})})
\sum_{1\le i\le p}q^{\mathrm{deg}_{2}(T,r_{i})}\frac{\mathrm{tree}_{2}(T,r_{i})}{T}
=\frac{\left[2\sum_{k=1}^{p}|T_{k}|+2p\right]}{\left[2(\sum_{k=1}^{p+r}|T_{k}|+p+r)\right]}, \\
\sum_{p+1\le i\le p+r}q^{\mathrm{deg}_{2}(T,r_{i})}\frac{\mathrm{tree}_{2}(T,r_{i})}{T}
=q^{2\sum_{k=1}^{p}|T_{k}|+2p}\frac{\left[2(\sum_{k=p+1}^{p+r}|T_{k}|+r)\right]}
{\left[2(\sum_{k=1}^{p+r}|T_{k}|+p+r)\right]}.
\end{gathered}
\end{align*}
The sum of these two terms gives Eqn.(\ref{eq:recT2}).
\end{proof}

Let $T_{i}$ be a partial tree connected to the $r_{i}$-th ($1\le i\le p+r$) edge.
We define 
\begin{eqnarray*}
C_{i}(T):=
\begin{cases}
q^{\mathrm{deg}_{4}(T,r_{i})}(1+q^{\mathrm{deg}_{3}(T,r_{i})}), & 1\le i\le p, \\
q^{\mathrm{deg}_{4}(T,r_{i})}, & p+1\le i\le p+r,
\end{cases}
\end{eqnarray*}
where 
\begin{eqnarray*}
\mathrm{deg}_{3}(T,r_{i})&:=&2(|T_{i+1}|+\ldots +|T_{p+r}|+p+r-i|)+|T_{i}|+1, \\
\mathrm{deg}_{4}(T,r_{i})&:=&\sum_{1\le j\le i-1}(|T_{j}|+1).
\end{eqnarray*}
Then, we have 
\begin{prop}
\begin{eqnarray*}
T=\sum_{1\le i\le p+r}C_{i}(T)\mathrm{tree}_{2}(T,r_{i}).
\end{eqnarray*}
\end{prop}
\begin{proof}
We show that 
\begin{eqnarray}
\label{eq:c2}
\sum_{1\le i\le p}C_{i}(T)\frac{\mathrm{tree}_{2}(T,r_i)}{T}+
\sum_{p+1\le i\le p+r}C_{i}(T)\frac{\mathrm{tree}_{2}(T,r_i)}{T}=1.
\end{eqnarray}
We first compute the second sum in the left hand side of Eqn.(\ref{eq:c2}).
By taking the sum from $i=p+r$ to $i=1$, we obtain 
\begin{eqnarray}
\label{eq:c2-2}
\sum_{p+1\le i\le p+r}C_{i}(T)\frac{\mathrm{tree}_{2}(T,r_i)}{T}
=q^{\sum_{k=1}^{p}|T_{k}|+p}
\frac{\left[2(\sum_{i=p+1}^{p+r}|T_{i}|+r)\right]}{\left[2(\sum_{i=1}^{p+r}|T_{i}|+p+r)\right]}.
\end{eqnarray}
We add terms from $i=p$ to $i=1$ one by one to the right hand side of Eqn.(\ref{eq:c2-2}).
We obtain that the left hand side of Eqn.(\ref{eq:c2}) is equal to one. 
\end{proof}

We consider a Dyck path $\lambda$ and the corresponding tree $T^{\mathrm{Dyck}}:=A(\lambda)$.
We enumerate edges connected to a leaf from left to right by $1,2,\ldots,p$ 
and call the $i$-th edge $e_{i}$.
A tree $\mathrm{tree}_{3}(T^{\mathrm{Dyck}},e_{i})$ is defined as a tree obtained from $T$
by deleting $e_{i}$. 
Let $\mathrm{deg}_{4}(T^{\mathrm{Dyck}},e_{i})$ be the number of edges in $T^{\mathrm{Dyck}}$ 
strictly left to the edge $e_i$. Here, ``strictly left" means that we do not include the edges 
which are parents of $e_{i}$.
Given an edge $e$ in $T$, we denote by $\mathrm{ht}(e)$ the number of edges below 
$e$ plus one, that is, the height of the chord corresponding to $e$.
\begin{lemma}
\label{lemma:recDyck0}
We have 
\begin{eqnarray}
\label{eqn:recDyck0}
T^{\mathrm{Dyck}}=\sum_{i=1}^{p}q^{\mathrm{deg}_{4}(T^{\mathrm{Dyck}},e_{i})}
\mathrm{tree}_{3}(T^{\mathrm{Dyck}},e_i).
\end{eqnarray}
\end{lemma}
\begin{proof}
We consider a partial tree $T'$ depicted as below:
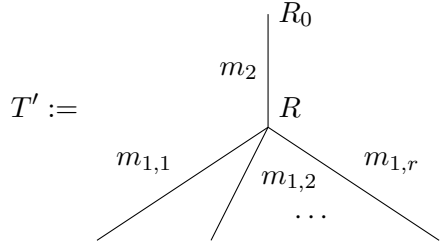
\begin{eqnarray*}
T':=\tikzpic{-0.5}{
\coordinate 
	child{coordinate (c0)   
		child{coordinate (c01)}
		child{coordinate (c02)}
		child[missing]
		child{coordinate (c03)}
	};

\draw($($(c0)!.8!(c02)$)!.5!($(c0)!.8!(c03)$)$)node{$\cdots$};
\draw(c0)node[anchor=south west]{$R$};
\draw(0,0)node[anchor=west]{$R_{0}$};
\draw($(0,0)!.5!(c0)$)node[anchor=east]{$m_{2}$};
\draw($(c0)!.5!(c01)$)node[anchor=south east]{$m_{1,1}$};
\draw($(c0)!.3!(c02)$)node[anchor=north west]{$m_{1,2}$};
\draw($(c0)!.5!(c03)$)node[anchor=south west]{$m_{1,r}$};
}
\end{eqnarray*}
The point $R_{0}$, $R$ is a ramification point and we have $m_{1,i}$ edges 
with $1\le i\le r$ below $R$ and $m_{2}$ edges above $R$.
We denote by $m_1$ the sum $\sum_{i=1}^{r}m_{1,i}$.
We have a unique sequence of edges from $R_{0}$ to the root.
We denote by $\mathbf{E}$ the set of such edges.
Let $T'_{i}$, $1\le i\le r$, be a tree obtained from $T'$ by 
deleting the $i$t-th leftmost and lowest edge connected to a leaf.
We have 
\begin{eqnarray*}
\frac{T'_{i}}{T^{\mathrm{Dyck}}}
=\frac{[m_{1,i}][m_1+m_2]}{[N][m_1]}\prod_{e\in\mathbf{E}}\frac{[\mathrm{ht}(e)]}{[\mathrm{ht}(e)-1]},
\end{eqnarray*}
for $1\le i\le r$ where $N$ is the number of edges in $T$ 
and $\mathrm{deg}_{4}(T,e_{i})=\sum_{1\le j\le i-1}m_{1,i-1}$.
Thus we have 
\begin{eqnarray}
\label{eqn:recDyck}
\sum_{i=1}^{r}q^{\mathrm{deg}_{4}(T^{\mathrm{Dyck}},e_{i})}\frac{T'_{i}}{T^{\mathrm{Dyck}}}
=\frac{[m_1+m_2]}{[N]}\prod_{e\in\mathbf{E}}\frac{[\mathrm{ht}(e)]}{[\mathrm{ht}(e)-1]}.
\end{eqnarray}
The right hand side of Eqn.(\ref{eqn:recDyck}) implies that $T'$ can be transformed 
to a tree with $m_1+m_2$ edges without a ramification point.
By repeating the above procedure for partial trees, we obtain Eqn.(\ref{eqn:recDyck0}).

\end{proof}

Below, we assume that a tree $T$ does not have arrows.
Let $r_i$, $1\le i\le p$ be the $i$-th edge without a $\bullet$
from left which is connected to the root of $T$ and 
$r_{p+1}$ be the edge with a $\bullet$ which is connected to 
the root of $T$.
We denote by $T_{i}$, $1\le i\le p+1$ a partial tree connected to 
the edge $r_{i}$.
We define 
\begin{eqnarray*}
\mathrm{deg}_{5}(T,r_{i})&:=&2\left(p+1-i+\sum_{i+1\le j\le p+1}|T_{j}|\right)
+\sum_{1\le j\le i}|T_{j}|+i, \qquad 1\le i\le p, \\
\mathrm{deg}_{5}(T,r_{p+1})&:=&\sum_{1\le j\le p}|T_{j}|+p.
\end{eqnarray*}
Let $s_{i}$ be the $i$-th edge from left which is connected to a leaf of partial 
trees $T_1,\ldots,T_{p}$.
We define 
\begin{eqnarray*}
\mathrm{deg}_{6}(T,s_i):=\#\{e \ |\ \text{an edge $e$ is strictly left to $s_{i}$}\},
\end{eqnarray*}
where ``strictly left" means that we do not include the edges which are parents of $s_{i}$.

\begin{prop}
\label{prop:fac-leftorder}
Suppose that $T$ has no arrows. 
We have
\begin{eqnarray*}
T=\sum_{1\le i\le p+1}q^{\mathrm{deg}_{5}(T,r_i)}\mathrm{tree}_{2}(T,r_{i})
+\sum_{j}q^{\mathrm{deg}_{6}(T,s_{j})}\mathrm{tree}_{3}(T,s_j).
\end{eqnarray*}
\end{prop}
\begin{proof}
From Lemma \ref{lemma:tree2}, we have an expression $\mathrm{tree}_{2}(T,r_{i})/T$ and 
\begin{eqnarray*}
\frac{\mathrm{tree}_{2}(T,r_{p+1})}{T}
=\frac{\left[2(|T_{p+1}|+1)\right]}{\left[2(\sum_{j=1}^{p+1}|T_{j}|+p+1)\right]}.
\end{eqnarray*}
By a straightforward calculation, we have 
\begin{eqnarray}
\label{eq:deg5tree2}
\sum_{1\le i\le p+1}q^{\mathrm{deg}_{5}(T,r_i)}\frac{\mathrm{tree}_{2}(T,r_i)}{T}
=q^{\mathrm{deg}_{5}(T,r_{p+1})}
\frac{\left[\sum_{k=1}^{p}|T_{k}|+p+2(|T_{p+1}|+1)\right]}
{\left[2(\sum_{k=1}^{p+1}|T_{k}|+p+1)\right]}.
\end{eqnarray}
We denote by $e_{i,j}$ the $j$-th edge connected to a leaf 
in a partial tree $T_{i}$.
We have 
\begin{eqnarray*}
\frac{\mathrm{tree}_{3}(T,e_{i,j})}{T}
=\frac{[|T_{i}|+1]}{[2(\sum_{k=1}^{p+1}|T_{k}|+p+1)]}
\cdot\frac{P^{\mathrm{Dyck}}(T_{i}\setminus e_{i,j})}{P^{\mathrm{Dyck}}(T_{i})}, 
\quad \text{for $1\le i\le p$},
\end{eqnarray*}
where $P^{\mathrm{Dyck}}(T)$ is the generating function $P^{\mathrm{Dyck}}_{\lambda,*}$
for $T:=A(\lambda)$.

\begin{eqnarray}
\nonumber
\sum_{j}q^{\mathrm{deg}_{6}(T,s_{j})}\frac{\mathrm{tree}_{3}(T,s_{j})}{T}
&=&\sum_{1\le j\le p}\sum_{l}q^{\mathrm{deg}_{6}(T,e_{j,l})}
\frac{\mathrm{tree}_{3}(T,e_{j,l})}{T}  \\
&=&\sum_{1\le j\le p}\sum_{l}
\frac{q^{\mathrm{deg}_{6}(T,e_{j,l})}[|T_{j}|+1]}{[2(\sum_{k=1}^{p+1}|T_{k}|+p+1)]}
\cdot\frac{P^{\mathrm{Dyck}}(T_{j}\setminus e_{j,l})}{P^{\mathrm{Dyck}}(T_{j})} \\
\nonumber
&=&\sum_{1\le k\le p}q^{\sum_{1\le j\le k-1}(|T_{j}|+1)}
\frac{[|T_{k}|+1]}{[2(\sum_{k=1}^{p+1}|T_{k}|+p+1)]} \\
\label{eqn:recT3}
&=&\frac{[\sum_{k=1}^{p}|T_{k}|+p]}{[2(\sum_{k=1}^{p+1}|T_{k}|+p+1)]},
\end{eqnarray}
where we have used Lemma \ref{lemma:recDyck0} in the third equality.
By adding Eqn.(\ref{eq:deg5tree2}) to Eqn.(\ref{eqn:recT3}), we obtain
the desired expression. 
\end{proof}

\subsection{Expressions for a general tree}
\label{sec:BI-gtree}
Let $L_{\mathrm{ref}}$ be a labelling of a tree $T$ such that the $\mathrm{post}'(L)=id$ and 
$n^{\mathrm{ref}}_e$ be an integer on an edge $e$ in $L_{\mathrm{ref}}$.
Given $L_{\mathrm{ref}}$, we define 
\begin{eqnarray*}
S(e):=|T|+1-n^{\mathrm{ref}}_{e}.
\end{eqnarray*}
We denote by $\mathcal{E}$ and $\mathcal{E}_{\bullet}$ the set of edges 
in $T$ and the set of edges with $\bullet$. 
Let $A(T)\subseteq\mathcal{E}\setminus\mathcal{E}_{\bullet}$ be the set 
of edges such that an edge $e$ does not have parents with incoming arrows.
We define $B(T):=\mathcal{E}\setminus(\mathcal{E}_{\bullet}\cup A(T))$.

Fix a natural labelling $L$ in $T$ and let $n_e$ be an integer 
on an edge $e$.
We put a circle on an edge in $T$ by the following rule:
\begin{enumerate}
\item[(R1)] There is no circle on an edge in $\mathcal{E}_{\bullet}$.
\item[(R2)] We put a circle on an edge without an incoming arrow.
\item[(R3)] Suppose we have a sequence of edges with arrows 
$e(1)\leftarrow \ldots\leftarrow e(p)\leftarrow e$.
If $n_{e(i)}>n_{e}$ for $1\le i\le p$, then 
we put circles on the edges $e(1)$ to $e(p)$.
\end{enumerate}
Let $\mathcal{E}_{\circ}(L)$ be the set of edges with a circle.
Note that $\mathcal{E}_{\circ}(L)$ depends on a natural labelling $L$.
From (R2), we have $A(T)\subseteq\mathcal{E}_{\circ}$ and 
the equality holds when $T$ does not have arrows.
We define
\begin{eqnarray*}
\mathrm{deg}_{7}(L):=\sum_{e\in\mathcal{E}}
\#\{e' \ |\  n_{e'}<n_{e}, \ \text{$e'$ is strictly right to $e$}\}.
\end{eqnarray*}
\begin{theorem}
\label{thrm:treewt1}
We have 
\begin{eqnarray*}
T=\left(\sum_{L}q^{\mathrm{deg}_{7}(L)}\prod_{e\in B(T)\cap\mathcal{E}_{\circ}(L)}(1+q^{S(e)-1})
\right)
\cdot\prod_{e\in A(T)}(1+q^{S(e)}).
\end{eqnarray*}
\end{theorem}
\begin{proof}
We prove Theorem by induction.
When the number of edges in a tree is two, we have six trees.
In terms of ballot paths, the six trees are 
$UDUD$, $UUDD$, $UUD$, $UUU$, $UDUU$ and $UDU$.
It is easy to show that Theorem is true for these six trees. 

Suppose that $T$ can be expressed as a concatenation of trees $T_{1},\ldots,T_{p+s}$ where
$T_{i}$, $1\le i\le p+s$, can not be decomposed into a concatenation of trees.
There exists a unique edge in $T_{i}$ which is connected to the root. We call the $i$-th 
edge from left connected to the root $r_{i}$. 
Further, we assume that the edges $r_{i}$ ,$1\le i\le p$, do not have incoming arrows, 
the edges $r_{i}$, $p+1\le i\le p+s-1$, have incoming arrows and the edge $r_{p+s}$
has a $\bullet$.
We consider a set $L(i)$ of natural labellings such that the edge $r_{i}$ has the integer one.
Let $\mathrm{tree}_{4}(L(i))$ be the the following sum: 
\begin{eqnarray*}
\mathrm{tree}_{4}(L(i)):=
\left(\sum_{L\in L(i)}q^{\mathrm{deg}_{7}(L)}\prod_{e\in B(T)\cap\mathcal{E}_{\circ}(L)}(1+q^{S(e)-1})
\right)
\cdot\prod_{e\in A(T)}(1+q^{S(e)}).
\end{eqnarray*}
Then, from the induction assumption, we have 
\begin{eqnarray*}
\mathrm{tree}_{4}(L(i))
&=&q^{\sum_{j=1}^{i-1}|T_{j}|}\mathrm{tree}_{2}(T,r_{i})(1+q^{|T|}),\quad \text{for $1\le i\le p$}, \\
\mathrm{tree}_{4}(L(i))
&=&q^{\sum_{j=1}^{i-1}|T_{j}|}\mathrm{tree}_{2}(T,r_{i})
\frac{(1+q^{|T|})}{(1+q^{\sum_{j=p+1}^{p+s}|T_{j}|})}, \quad \text{for $p+1\le i\le p+s$}.
\end{eqnarray*}
The term $q^{\sum_{j=1}^{i-1}|T_{j}|}$ comes from the fact that there are 
$\sum_{j=1}^{i-1}|T_{j}|$ integers bigger than one in $L(i)$. 
Note that the number of edges in $\mathrm{tree}_{2}(T,r_{i})$ is one less than that
of $T$. 
From Lemma \ref{lemma:tree2}, we have 
\begin{eqnarray*}
\sum_{i=1}^{p+s}\mathrm{tree}_{4}(L(i))=T,
\end{eqnarray*}
which implies that Lemma holds true.
\end{proof}

\begin{example}
Let $T$ be the tree as depicted below:
\begin{eqnarray*}
T:=
\scalebox{0.7}{
\tikzpic{-0.6}{
\coordinate
	child{coordinate (c0)
		child {coordinate(c3)}
		child[missing]
		child[missing]}
	child{coordinate (c1)}
	child{coordinate (c2)};
\node at ($(0,0)!.5!(c2)$){\scalebox{1.5}{$\bullet$}};
\draw[latex-,dashed](0.1,-1)--(0.9,-1);
\node at (c0){\scalebox{1.25}{$-$}};
}
}
\end{eqnarray*}
We compute the generating function by Theorem \ref{thrm:treewt1}. 
We have twelve natural labellings in $T$ and their weights are give 
as follows:
\begin{eqnarray*}
&&\tikzpic{-0.25}{
\node[draw,circle,inner sep=1pt] at (0,0) {$1$};
\node at (0.4,0){$3$};
\node at (0.8,0){$4$};
\node[draw,circle,inner sep=1pt] at (0,-0.4) {$2$};
\node at (0.4,-0.9){$1$};
}\quad
\tikzpic{-0.25}{
\node[draw,circle,inner sep=1pt] at (0,0) {$1$};
\node[draw,circle,inner sep=1pt] at (0.4,0){$4$};
\node at (0.8,0){$3$};
\node[draw,circle,inner sep=1pt] at (0,-0.4) {$2$};
\node at (0.4,-0.9){$q$};
}\quad
\tikzpic{-0.25}{
\node[draw,circle,inner sep=1pt] at (0,0) {$1$};
\node at (0.4,0){$2$};
\node at (0.8,0){$4$};
\node[draw,circle,inner sep=1pt] at (0,-0.4) {$3$};
\node at (0.4,-0.9){$q$};
}\quad
\tikzpic{-0.25}{
\node[draw,circle,inner sep=1pt] at (0,0) {$1$};
\node[draw,circle,inner sep=1pt] at (0.4,0){$4$};
\node at (0.8,0){$2$};
\node[draw,circle,inner sep=1pt] at (0,-0.4) {$3$};
\node at (0.4,-0.9){$q^2$};
}\quad
\tikzpic{-0.25}{
\node[draw,circle,inner sep=1pt] at (0,0) {$1$};
\node at (0.4,0){$2$};
\node at (0.8,0){$3$};
\node[draw,circle,inner sep=1pt] at (0,-0.4) {$4$};
\node at (0.4,-0.9){$q^2$};
}\quad
\tikzpic{-0.25}{
\node[draw,circle,inner sep=1pt] at (0,0) {$1$};
\node[draw,circle,inner sep=1pt] at (0.4,0){$3$};
\node at (0.8,0){$2$};
\node[draw,circle,inner sep=1pt] at (0,-0.4) {$4$};
\node at (0.4,-0.9){$q^3$};
} \\[12pt]
&&\tikzpic{-0.25}{
\node[draw,circle,inner sep=1pt] at (0,0) {$2$};
\node at (0.4,0){$1$};
\node at (0.8,0){$4$};
\node[draw,circle,inner sep=1pt] at (0,-0.4) {$3$};
\node at (0.4,-0.9){$q^2$};
}\quad
\tikzpic{-0.25}{
\node[draw,circle,inner sep=1pt] at (0,0) {$2$};
\node at (0.4,0){$1$};
\node at (0.8,0){$3$};
\node[draw,circle,inner sep=1pt] at (0,-0.4) {$4$};
\node at (0.4,-0.9){$q^3$};
}\quad
\tikzpic{-0.25}{
\node[draw,circle,inner sep=1pt] at (0,0) {$3$};
\node at (0.4,0){$1$};
\node at (0.8,0){$2$};
\node[draw,circle,inner sep=1pt] at (0,-0.4) {$4$};
\node at (0.4,-0.9){$q^4$};
}\quad
\tikzpic{-0.25}{
\node[draw,circle,inner sep=1pt] at (0,0) {$2$};
\node[draw,circle,inner sep=1pt] at (0.4,0){$4$};
\node at (0.8,0){$1$};
\node[draw,circle,inner sep=1pt] at (0,-0.4) {$3$};
\node at (0.4,-0.9){$q^3$};
}\quad
\tikzpic{-0.25}{
\node[draw,circle,inner sep=1pt] at (0,0) {$2$};
\node[draw,circle,inner sep=1pt] at (0.4,0){$3$};
\node at (0.8,0){$1$};
\node[draw,circle,inner sep=1pt] at (0,-0.4) {$4$};
\node at (0.4,-0.9){$q^4$};
}\quad
\tikzpic{-0.25}{
\node[draw,circle,inner sep=1pt] at (0,0) {$3$};
\node[draw,circle,inner sep=1pt] at (0.4,0){$2$};
\node at (0.8,0){$1$};
\node[draw,circle,inner sep=1pt] at (0,-0.4) {$4$};
\node at (0.4,-0.9){$q^5$};
}
\end{eqnarray*} 
The circles give the factor 
\begin{eqnarray*}
\tikzpic{-0.25}{
\node[draw,circle,inner sep=4pt] at (0,0) {};
\node[draw,circle,inner sep=4pt] at (0.4,0){};
\node at (0.8,0){$\ast$};
\node[draw,circle,inner sep=4pt] at (0,-0.4) {};
}\mapsto (1+q)(1+q^3)(1+q^4), \qquad
\tikzpic{-0.25}{
\node[draw,circle,inner sep=4pt] at (0,0) {};
\node at (0.4,0){$\ast$};
\node at (0.8,0){$\ast$};
\node[draw,circle,inner sep=4pt] at (0,-0.4) {};
}\mapsto (1+q^3)(1+q^4).
\end{eqnarray*}
The generating function is given by the sum of the contributions, namely we have 
\begin{eqnarray*}
T&=&(1+q)(1+q^3)(1+q^4)(q+q^2+q^3+q^3+q^4+q^5)\\
&&+(1+q^3)(1+q^4)(1+q+q^2+q^2+q^3+q^4)  \\
&=&[3][6][8]/[2].
\end{eqnarray*}
\end{example}

When a tree $T$ is written as $T=A(\lambda)$ for a Dyck path $\lambda$,
we abbreviate $P^{\mathrm{Dyck}}_{\lambda}$ as $P^{\mathrm{Dyck}}(T)$.

\begin{lemma}
\label{lemma:concatenation}
Suppose that a tree $T$ is written as a concatenation of three trees 
$T:=T_{1}\circ T_{2}\circ T_{3}$ and $T_{1}$ and $T_{2}$ do not have 
incoming arrows. Then we have the generating function satisfies
\begin{eqnarray}
\label{eqn:concaP}
P^{I}(T)=P(2|T_3|,|T_1|+|T_2|)\cdot P^{\mathrm{Dyck}}(T_1\circ T_2)\cdot P^{I}(T_{3}).
\end{eqnarray}
where $P^{I}(T):=P^{I}_{\lambda,\ast}$ for $T=A(\lambda)$.
\end{lemma}
\begin{proof}
Form the factorization (Theorem \ref{thrm:BI-fac}), we have 
\begin{eqnarray*}
P^{I}(T)=P(2|T_2|+2|T_3|,|T_1|)P(2|T_3|,|T_2|)P^{\mathrm{Dyck}}(T_1)P^{\mathrm{Dyck}}(T_2)
P^{I}(T_3).
\end{eqnarray*}
Substituting Lemma \ref{lemma:treefac} and Proposition \ref{prop:BalP} into 
the above equation, we obtain Eqn.(\ref{eqn:concaP}).
\end{proof}

Let $T^{\mathrm{Dyck}}$ be a tree for a Dyck path and $L^{\mathrm{D}}$ be 
a natural labelling of $T^{\mathrm{Dyck}}$.
Let $n_e$ be an integer on the edge $e$ in the labelling $L^{\mathrm{D}}$.
We define 
\begin{eqnarray*}
\mathrm{deg}^{D}(L):=\#\{e' \ |\ n_{e'}<n_{e}, \text{$e'$ is strictly right to $e$}\}.
\end{eqnarray*}
\begin{lemma}
\label{lemma:LSDyck}
We have
\begin{eqnarray*}
\sum_{L}q^{\mathrm{deg}^{D}(L)}=P^{\mathrm{Dyck}}(T^{\mathrm{Dyck}}).
\end{eqnarray*}
\end{lemma}
\begin{proof}
We consider a tree with capacities such that a capacity of a leaf 
is the number of edges strictly right to the leaf. 
Let $L^{\mathrm{LS}}$ be a labelling of type LS. 
Let $n^{\mathrm{LS}}_{e}$ be an integer on an edge $e$ in $L^{\mathrm{LS}}$.
We have a bijection from $L$ to $L^{\mathrm{LS}}$ such that 
$n^{\mathrm{LS}}_{e}=\#\{e'\ |\  n_{e'}<n_{e}, \text{$e'$ is strictly right to $e$}\}$.
Let $\widetilde{T}$ be a mirror image of $T^{\mathrm{Dyck}}$. 
Then, 
\begin{eqnarray*}
\sum_{L}q^{\mathrm{deg}^{D}(L)}&=&\sum_{L^{\mathrm{LS}}}q^{\sum_{e} n^{\mathrm{LS}}_e} \\
&=&P^{\mathrm{Dyck}}(\widetilde{T}) \\
&=&P^{\mathrm{Dyck}}(T^{\mathrm{Dyck}}),
\end{eqnarray*}
where we have used Theorem \ref{thrm:KL} and Theorem \ref{thrm:DyckKL} in the second
equality.
\end{proof}

We denote by $T^{\mathrm{rev}}$ a tree with capacities such that a capacity of a leaf $l$
is the number of edges strictly right to $l$. 
Let $L^{\mathrm{rev}}$ be a labelling of $T^{\mathrm{rev}}$ satisfying (LS1) and (LS2).
We denote by $\mathcal{L}^{\mathrm{rev}}(\lambda)$ the set of such labellings.
\begin{lemma}
\label{lemma:LSrev}
Let $T$ be a tree for a Dyck path $\lambda$. Then, we have 
\begin{eqnarray}
\label{eqn:LSrev}
\sum_{L\in\mathcal{L}^{\mathrm{LS}}(\lambda)}q^{\mathrm{deg}_{D}(L)}
=\sum_{L\in\mathcal{L}^{\mathrm{rev}}(\lambda)}q^{\mathrm{deg}_{D}(L)}
\end{eqnarray}
where $\mathrm{deg}_{D}(L)$ is the sum of the labels of a labelling $L$.
\end{lemma}
\begin{proof}
From Theorem \ref{thrm:KL} and Theorem \ref{thrm:DyckKL}, 
the left hand side of Eqn.(\ref{eqn:LSrev}) is equal to 
$P^{\mathrm{Dyck}}_{\lambda}$.
A path $\lambda$ is written as $\lambda=\lambda_1\lambda_2\ldots \lambda_{2N}$ 
with $\lambda_{i}=U$ or $D$.
Then we define $\overline{\lambda}:=\overline{\lambda}_{1}\ldots \overline{\lambda}_{2N}$ by 
$\overline{\lambda}_{i}=U$ for $\lambda_{2N+1-i}=D$ and 
$\overline{\lambda}_{i}=D$ for $\lambda_{2N+1-i}=U$.
Then, the right hand side of Eqn.(\ref{eqn:LSrev}) is equal to 
$P^{\mathrm{Dyck}}_{\overline{\lambda}}$.   
By definition, a path $\overline{\lambda}$ is also a mirror image of $\lambda$.
A mirror image of a Dyck tiling $D$ above $\lambda$ is also a Dyck 
tiling $\overline{D}$ above $\overline{\lambda}$. 
We also have $\mathrm{art}(D)=\mathrm{art}(\overline{D})$.
Thus we have $P^{\mathrm{Dyck}}_{\lambda}=P^{\mathrm{Dyck}}_{\overline{\lambda}}$, 
which implies Eqn.(\ref{eqn:LSrev}). 
\end{proof}

Proposition \ref{lemma:recT2} can be translated into the following theorem.
We keep the notation in Theorem \ref{thrm:treewt1}.
Let $L$ be a natural labelling of $T$, $n_{e}$ be an integer on the edge 
in $L$ and $\mathcal{L}(T)$ be the set of natural labellings of $T$.
We divide the edges into two blocks: the first block is the set of edges 
which is strictly left to the edge $r_{p+1}$ and the second block is 
the set of edges strictly right to the edge $r_{p}$.
We define 
\begin{eqnarray*}
m_{e}&:=&2\#\{e'\ |\ n_{e'}<n_{e}, \text{$e'$ is strictly right to $e$ and in a different block}\}\\
 &&\quad+\#\{e'\ |\  n_{e'}<n_{e}, \text{$e'$ is strictly right to $e$ and in the same block}\}, \\
\mathrm{deg}_{8}(L)&:=&\sum_{e}m_{e},
\end{eqnarray*}
and $N:=\sum_{i=1}^{p}|T_{i}|$.
\begin{theorem}
\label{thrm:treewt2}
We have 
\begin{eqnarray*}
T=\left(\sum_{L\in\mathcal{L}(T)}q^{\mathrm{deg}_{8}(L)}
\prod_{e\in B(T)\cap\mathcal{E}_{\circ}(L)}(1+q^{S(e)-1})
\right)\prod_{i=1}^{N}(1+q^{i}).
\end{eqnarray*}
\end{theorem}
\begin{proof}
We prove Theorem by induction.
We assume that the theorem holds true for a tree with at most $|T|-1$ edges.

Let $\mathcal{L}(i)$ be the set of natural labellings such that 
the edge $r_{i}$ has the integer one.
Given a tree $T$, let $T(i)$ be the sum 
\begin{eqnarray*}
T(i):=\left(\sum_{L\in\mathcal{L}(i)}q^{\mathrm{deg}_{8}(L)}
\prod_{e\in B(T)\cap\mathcal{E}_{\circ}(L)}(1+q^{S(e)-1})\right)
\cdot \prod_{i=1}^{N}(1+q^{i}).
\end{eqnarray*}
It is obvious that we have 
$T(i)=(1+q^{\mathrm{deg}_{2}(T,r_{p+1})})q^{\mathrm{deg}_{2}(T,r_i)}
\mathrm{tree}_{2}(T,r_{i})$ for $1\le i\le p$.

We will show that $T(i)=q^{\mathrm{deg}_{2}(T,r_{i})}\mathrm{tree}_{2}(T,r_{i})$ 
for $p+1\le i\le p+s$.
From Lemma \ref{lemma:concatenation} and Lemma \ref{lemma:LSDyck},
it is enough to consider the following $T'$. A tree $T'$ is a concatenation 
of three trees, {\it i.e.}, $T'=T_{1}\circ T_{2}\circ T_{3}$. 
Further, a tree $T_{i}$, $1\le i\le3$, does not have a ramification point, 
the tree $T_{3}$ consists of only edges with a $\bullet$, and there is an 
arrow from the top edge in $T_{3}$ to the top edge in $T_{2}$.
Let $M_i$, $1\le i\le 3$, be the number of edges in $T_{i}$.
We have 
\begin{eqnarray*}
\mathrm{tree}_{2}(T',r_{2})
&=&\genfrac{[}{]}{0pt}{}{M_{1}+M_{2}-1}{M_{1}}
\genfrac{[}{]}{0pt}{}{M_1+M_2+M_3-1}{M_3}_{q^{2}}
\prod_{j=1}^{M_1+M_2-1}\frac{[2j]}{[j]}, \\
T'(2)&=&
q^{M_1}\genfrac{[}{]}{0pt}{}{M_2+M_3-1}{M_3}\genfrac{[}{]}{0pt}{}{M_1+M_2+M_3-1}{M_1}_{q^{2}}
\prod_{j=1}^{M_1}\frac{[2j]}{[j]}\cdot \prod_{k=M_3+1}^{M_2+M_3-1}\frac{[2k]}{[k]}.
\end{eqnarray*}
By a straightforward calculation, we have $T'(2)=q^{M_1}\mathrm{tree}_{2}(T,r_2)$.
By a similar calculation, we have $T'(3)=q^{M_1+M_2}\mathrm{tree}_{2}(T,r_3)$.
The sum $\sum_{i}T'(i)$ is equal to $T'$. 
Therefore, the sum $\sum_{i}T(i)$ is equal to $T$ by the recurrence relation 
in Proposition \ref{lemma:recT2}, which implies Theorem holds true for $T$.
\end{proof}

Let $T$ be a tree with capacities $A(\lambda/\lambda_{0})$ where 
$\lambda_0$ consists of only $U$'s.
Let $L^{\mathrm{LS}}$ be a labelling of LS type in $T$ and $n_{e}$ be an integer 
on an edge $e$. 
The labels $n_{e}$ satisfy (LS1) and (LS2).
We put an circle on an edge in the following rules:
\begin{enumerate}
\item[(R4)]There is no circle on an edge with $\bullet$.
\item[(R5)]We put a circle on an edge without an incoming arrow.
\item[(R6)]Suppose we have a sequence of edges with arrows: 
$e(0)\leftarrow e(1)\leftarrow\ldots\leftarrow e(p)\leftarrow e$.
If $n_{e(i)}\ge n_{e}$ for $1\le i\le p$ and $n_{e(0)}<n_{e}$,
we put circles on edges $e(1)$ to $e(p)$.
\end{enumerate}
Let $\mathcal{E}_{\circ}(L^{\mathrm{LS}})$ be the set of edges with a circle for 
a label $L$.
We consider the post-order from right on $T$ and let $S'(e)$ be an
integer on an edge $e$.

We denote by $\mathcal{E}$ the set of edges in $A(\lambda)$.
Given a labelling $L^{\mathrm{LS}}$ of type LS, we define 
\begin{eqnarray*}
\mathrm{deg}_{9}(L^{\mathrm{LS}}):=\sum_{e\in\mathcal{E}}n_{e}.
\end{eqnarray*}

\begin{theorem}
\label{thrm:GFLS}
We have 
\begin{eqnarray}
\label{eqn:GFLS2}
T=\sum_{L^{\mathrm{LS}}}q^{\mathrm{deg}_{9}(L^{\mathrm{LS}})}
\prod_{e\in\mathcal{E}_{\circ}(L^{\mathrm{LS}})}(1+q^{S'(e)}).
\end{eqnarray}
\end{theorem}
Before proceeding with the proof of Theorem \ref{thrm:GFLS}, we introduce 
a lemma used later.

\begin{lemma}
\label{lemma:summult}
We have 
\begin{eqnarray}
\label{eqn:sum2}
\sum_{i=I}^{c}q^{i(M+N)}
\genfrac{[}{]}{0pt}{}{N+M+c-i-1}{M,N-1,c-i}
=q^{I(M+N)}\genfrac{[}{]}{0pt}{}{N+M-1}{M}\genfrac{[}{]}{0pt}{}{N+M+c-I}{c-I}.
\end{eqnarray}
\end{lemma}
\begin{proof}
We prove Lemma by induction.
When $I=c$, the both sides of Eqn.(\ref{eqn:sum2}) are equal to 
$q^{c(M+N)}\genfrac{[}{]}{0pt}{}{N+M-1}{M}$.
We assume that Lemma holds true for $I$.
Then, we have 
\begin{align*}
&\sum_{i=I-1}^{c}q^{i(M+N)}
\genfrac{[}{]}{0pt}{}{N+M+c-i-1}{M,N-1,c-i} \\
&\qquad=q^{I(M+N)}\genfrac{[}{]}{0pt}{}{N+M-1}{M}\genfrac{[}{]}{0pt}{}{N+M+c-I}{c-I}
+q^{(I-1)(M+N)}\genfrac{[}{]}{0pt}{}{N+M+c-I}{M,N-1,c-I+1} \\
&\qquad=q^{(I-1)(M+N)}\genfrac{[}{]}{0pt}{}{N+M-1}{M}\genfrac{[}{]}{0pt}{}{N+M+c-I+1}{c-I+1},
\end{align*}
which implies Lemma holds true.
\end{proof}

\begin{proof}[Proof of Theorem \ref{thrm:GFLS}]
To show Theorem is true, it is enough to compute the following two partial trees $T_1$ and $T_2$:
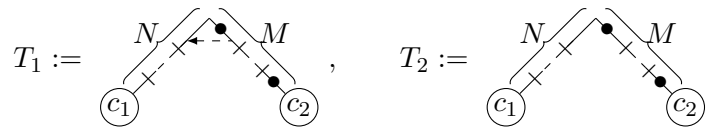
\begin{eqnarray*}
T_{1}:=
\tikzpic{-0.5}{
\draw(0,0)--(-0.45,-0.45)(-0.7,-0.7)--(-1,-1)node[draw,circle,anchor=north east,inner sep=1pt]{$c_{1}$};
\draw(-0.4,-0.4)node{\rotatebox{-45}{$-$}}(-0.8,-0.8)node{\rotatebox{-45}{$-$}};
\draw[dashed](-0.3,-0.3)--(-0.7,-0.7);
\draw[decoration={brace,mirror,raise=5pt},decorate]
  (0,0) --(-1,-1);
\draw(-0.2,-0.2)node[left=9pt]{$N$};
\draw(0,0)--(0.45,-0.45)(0.65,-0.65)--(1,-1)node[draw,circle,anchor=north west,inner sep=1pt]{$c_{2}$};
\draw[dashed](0.45,-0.45)--(0.65,-0.65);
\draw(0.4,-0.4)node{\rotatebox{45}{$-$}}(0.7,-0.7)node{\rotatebox{45}{$-$}};
\draw(0.15,-0.15)node{$\bullet$}(0.85,-0.85)node{$\bullet$};
\draw[decoration={brace,mirror,raise=5pt},decorate]
  (1,-1)--(0,0);
\draw(0.2,-0.2)node[right=9pt]{$M$};
\draw[latex-,dashed](-0.3,-0.3)--(0.3,-0.3);
},\qquad
T_{2}:=
\tikzpic{-0.5}{
\draw(0,0)--(-0.45,-0.45)(-0.7,-0.7)--(-1,-1)node[draw,circle,anchor=north east,inner sep=1pt]{$c_{1}$};
\draw(-0.4,-0.4)node{\rotatebox{-45}{$-$}}(-0.8,-0.8)node{\rotatebox{-45}{$-$}};
\draw[dashed](-0.3,-0.3)--(-0.7,-0.7);
\draw[decoration={brace,mirror,raise=5pt},decorate]
  (0,0) --(-1,-1);
\draw(-0.2,-0.2)node[left=9pt]{$N$};
\draw(0,0)--(0.45,-0.45)(0.65,-0.65)--(1,-1)node[draw,circle,anchor=north west,inner sep=1pt]{$c_{2}$};
\draw[dashed](0.45,-0.45)--(0.65,-0.65);
\draw(0.4,-0.4)node{\rotatebox{45}{$-$}}(0.7,-0.7)node{\rotatebox{45}{$-$}};
\draw(0.15,-0.15)node{$\bullet$}(0.85,-0.85)node{$\bullet$};
\draw[decoration={brace,mirror,raise=5pt},decorate]
  (1,-1)--(0,0);
\draw(0.2,-0.2)node[right=9pt]{$M$};
}
\end{eqnarray*}
where $c_2:=c_1+N$.

Let $g(i,c,N)$ be the generating function for a tree $T$ such that 
$T$ does not have a ramification point, the top edge has the integer
$i$, the capacity is $c$ and the number of edges in $T$ is N, {\it i.e.},  
\begin{eqnarray*}
g(i,c,N):=q^{iN}\genfrac{[}{]}{0pt}{}{N+c-1-i}{N-1}.
\end{eqnarray*}
We define $g_{\le}(i,c,N)$ (resp. $g_{>}(i,c,N)$) in a similar way as 
$g(i,c,N)$ except that the integer on the top edge is equal to or bigger than 
(resp. smaller than) $i$, {\it i.e.},
\begin{eqnarray*}
g_{\le}(i,c,N)&:=&\sum_{j=0}^{i}q^{jN}\genfrac{[}{]}{0pt}{}{N+c-i-1}{N-1}, \\
g_{>}(i,c,N)&:=&\sum_{j=i+1}^{c}q^{jN}\genfrac{[}{]}{0pt}{}{N+c-i-1}{N-1}.
\end{eqnarray*}

Let $e_{1}$ (resp. $e_2$) be the edge with an incoming (outgoing) arrow 
in the partial tree $T_{1}$.
Since we have an arrow in $T_{1}$, the generating function for $T_{1}$ 
is the sum of two contributions: the first case is that the integer on 
$e_{1}$ is smaller or equal to the integer on $e_2$, and the second case 
is the integer on $e_1$ is bigger than the integer $e_{2}$.
We have 
\begin{eqnarray*}
T_{1}=\prod_{i=M+1}^{M+N-1}(1+q^{i})\left\{
\sum_{i=0}^{c_1}(1+q^{M+N})g(i,c_1,N)g_{\le}(i,c_2,M)
+\sum_{i=0}^{c_1}g(i,c_1,N)g_{>}(i,c_2,M)
\right\}.
\end{eqnarray*}
Since we have $g_{\le}(i,c_2,M)+g_{>}(i,c_2,M)=\genfrac{[}{]}{0pt}{}{M+c_2}{M}$,
we have 
\begin{eqnarray*}
T_{1}&=&\prod_{i=M+1}^{M+N-1}(1+q^{i})\left\{
(1+q^{M+N})\genfrac{[}{]}{0pt}{}{N+M+c_1}{M,N,c_1}
-\sum_{i=0}^{c_1}q^{(i+1)(N+M)+M}\genfrac{[}{]}{0pt}{}{M+N+c_1-i-1}{M,N-1,c_1-i}
\right\} \\
&=&\prod_{i=M+1}^{M+N-1}(1+q^{i})
\left\{
(1+q^{M+N})\genfrac{[}{]}{0pt}{}{M+N+c_1}{M,N,c_1}-q^{2M+N}\genfrac{[}{]}{0pt}{}{N+M-1}{M}
\genfrac{[}{]}{0pt}{}{N+M+c_1}{c_1}
\right\} \\
&=&\genfrac{[}{]}{0pt}{}{N+M}{N}_{q^{2}}
\frac{[N+2M]}{[2(N+M)]}\prod_{i=1}^{N}(1+q^{i})
\cdot
\tikzpic{-0.4}{
\draw(0,0)--(0,-0.35)(0,-0.65)--(0,-1);
\draw(0,-0.3)node{$-$}(0,-0.7)node{$-$};
\draw[dashed](0,-0.4)--(0,-0.6);
\draw(0,-0.15)node{$\bullet$}(0,-0.85)node{$\bullet$};
\draw[decoration={brace,mirror,raise=5pt},decorate]
    (0,-1)-- (0,0);
\draw(0,-0.5)node[right=9pt]{$M+N$};
\draw(0,-1)node[draw,circle,anchor=north,inner sep=1pt]{$c_{1}$};
},
\end{eqnarray*}
where we have used Lemma \ref{lemma:summult} in the second equality.

By a similar calculation, we have 
\begin{eqnarray*}
T_{2}&=&\genfrac{[}{]}{0pt}{}{N+M+c_1}{N,M,c_1}\prod_{j=M+1}^{N+M}(1+q^{i}) \\
&=&\genfrac{[}{]}{0pt}{}{N+M}{N}_{q^{2}}\prod_{j=1}^{N}(1+q^{j})
\cdot \tikzpic{-0.4}{
\draw(0,0)--(0,-0.35)(0,-0.65)--(0,-1);
\draw(0,-0.3)node{$-$}(0,-0.7)node{$-$};
\draw[dashed](0,-0.4)--(0,-0.6);
\draw(0,-0.15)node{$\bullet$}(0,-0.85)node{$\bullet$};
\draw[decoration={brace,mirror,raise=5pt},decorate]
    (0,-1)-- (0,0);
\draw(0,-0.5)node[right=9pt]{$M+N$};
\draw(0,-1)node[draw,circle,anchor=north,inner sep=1pt]{$c_{1}$};
}.
\end{eqnarray*}
From Lemma \ref{lemma:tritree} and Section \ref{sec:Fac-trees}, 
it is clear that Eqn.(\ref{eqn:GFLS2}) gives an expression of the 
generating function for $T$.
\end{proof}

\begin{example}
Let $T$ be the following tree with capacities:
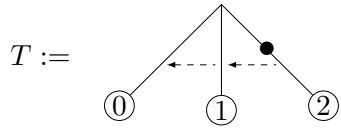
\begin{eqnarray*}
T:=
\scalebox{0.8}{
\tikzpic{-0.5}{
\coordinate
	child{coordinate (c0)}
	child{coordinate (c1)}
	child{coordinate (c2)};
\node at ($(0,0)!.5!(c2)$){\scalebox{1.5}{$\bullet$}};

\draw[latex-,dashed](-0.9,-1)--(-0.1,-1);
\draw[latex-,dashed](0.1,-1)--(0.9,-1);
\node[draw,circle,anchor=north east,inner sep=1pt] at (c0){\scalebox{1.25}{$0$}};
\node[draw,circle,anchor=north,inner sep=1pt] at (c1){\scalebox{1.25}{$1$}};
\node[draw,circle,anchor=north west,inner sep=1pt] at (c2){\scalebox{1.25}{$2$}};
}
}
\end{eqnarray*}
We compute the generating function by Theorem \ref{thrm:GFLS}.
We have six configurations with weights:
\begin{eqnarray*}
\tikzpic{-0.25}{
\node[draw,circle,inner sep=1pt] at (0,0) {$0$};
\node[draw,circle,inner sep=1pt] at (0.4,0){$0$};
\node at (0.8,0){$0$};
\node at (0.4,-0.6){$1$};
}\quad
\tikzpic{-0.25}{
\node[draw,circle,inner sep=1pt] at (0,0) {$0$};
\node at (0.4,0){$0$};
\node at (0.8,0){$1$};
\node at (0.4,-0.6){$q$};
}
\quad
\tikzpic{-0.25}{
\node[draw,circle,inner sep=1pt] at (0,0) {$0$};
\node at (0.4,0){$0$};
\node at (0.8,0){$2$};
\node at (0.4,-0.6){$q^{2}$};
}\quad
\tikzpic{-0.25}{
\node[draw,circle,inner sep=1pt] at (0,0) {$0$};
\node[draw,circle,inner sep=1pt] at (0.4,0){$1$};
\node at (0.8,0){$0$};
\node at (0.4,-0.6){$q$};
}\quad
\tikzpic{-0.25}{
\node at (0,0) {$0$};
\node[draw,circle,inner sep=1pt] at (0.4,0){$1$};
\node at (0.8,0){$1$};
\node at (0.4,-0.6){$q^2$};
}\quad
\tikzpic{-0.25}{
\node at (0,0) {$0$};
\node at (0.4,0){$1$};
\node at (0.8,0){$2$};
\node at (0.4,-0.6){$q^3$};
}
\end{eqnarray*}
The circles give the factor 
\begin{eqnarray*}
\tikzpic{-0.25}{
\node[draw,circle,inner sep=4pt] at (0,0) {\ };
\node[draw,circle,inner sep=4pt] at (0.4,0){\ };
\node at (0.8,0){$\ast$};
}\mapsto (1+q^2)(1+q^3),\quad
\tikzpic{-0.25}{
\node[draw,circle,inner sep=4pt] at (0,0) {\ };
\node at (0.4,0){$\ast$};
\node at (0.8,0){$\ast$};
}\mapsto (1+q^3),\quad
\tikzpic{-0.25}{
\node at (0,0) {$\ast$};
\node[draw,circle,inner sep=4pt] at (0.4,0){\ };
\node at (0.8,0){$\ast$};
}\mapsto (1+q^2)
\end{eqnarray*}
Then, the generating function $T$ is given by 
\begin{eqnarray*}
T&=&(1+q)(1+q^2)(1+q^3)+(q+q^2)(1+q^3)+q^2(1+q^2)+q^3 \\
&=&[3][5]. 
\end{eqnarray*}
\end{example}

\subsection{Tree without arrows}
\label{sec:twoarrow}
In this subsection, we consider trees without arrow.
Given two permutations $u$ and $v$ in the set of signed permutations, 
we define the weak left (Bruhat) order $\ge_{L}$ as follows.
We have $v\ge_{L}u$ if some final subword of some reduced word 
for $v$ is a reduced word for $u$.
\begin{theorem}
\label{thrm:BI-post}
Let $L_0$ be the natural labelling of $T$ such that $\mathrm{pre}(L_{0})=id$ 
and a modified post-order word $\sigma_{0}=\mathrm{post}(L_{0})$.
We have 
\begin{eqnarray*}
P_{\lambda,\ast}^{I}
=\sum_{\sigma\ge_{L}\sigma_{0}}q^{\mathrm{Inv}(\sigma)-\mathrm{Inv}(\sigma_{0})}
\end{eqnarray*}
where $\ge_{L}$ is the weak left (Bruhat) order for signed permutations. 
\end{theorem}
\begin{proof}
We prove Theorem by induction. 
Theorem holds true when a tree $T$ has one edge. 
We assume that Theorem is true at most $N-1$ edges.
We have a recurrence formula from Proposition \ref{prop:fac-leftorder}.
We compare $\sigma:=\sigma_{1}\ldots \sigma_{N}$ to the modified post-order word 
$\sigma_0:=\sigma_{0,1}\ldots\sigma_{0,N}$ of $L_{0}$.
Recall that some $\sigma_{0,i}$'s are underlined and correspond to the edges 
with $\bullet$.
We consider a word $\sigma\ge_{L}\sigma_{0}$ 
such that $\sigma_{1}=\sigma_{0,i}$ and $\sigma_{0,i}$ is not underlined. 
Since $\sigma_{0}$ is a modified post-order word, such $\sigma_{0,i}$ corresponds 
to the edge $s_{j}$ in $T$.
Let $\sigma'_{0}$ be the post-order word of a labelling $L_{0}$ on
the tree $\mathrm{tree}_{3}(T,s_{j})$.
Then, it is obvious that $i-1=\mathrm{deg}_{6}(T,s_{j})$ and 
a subword $\sigma':=\sigma_{2}\ldots\sigma_{N}$ satisfies
$\sigma'\ge_{L}\sigma'_{0}$ by induction assumption.
Similarly, we consider the case where $\sigma_{0,i}$ is not underlined 
and $\sigma_{1}=\underline{\sigma_{0,i}}$, {\it i.e.}, $\sigma_{1}$ is 
underlined.
A possible such $\sigma_{0,i}$ corresponds to the edge $r_{j}$.
In signed permutations, we need $\mathrm{deg}_{5}(T,r_{j})$ simple 
transpositions to move $\sigma_{0,i}$ to the $N$-th element, put 
an underline on it, and move forward to the first element in $\sigma$.
Let $\sigma'_{0}$ be the post-order word of a labelling $L_{0}$ on
the tree $\mathrm{tree}_{2}(T,r_{j})$.
The subword $\sigma':=\sigma_2\ldots\sigma_{N}$ satisfies 
$\sigma'\ge_{L}\sigma'_{0}$ by induction assumption.
The sum of these two contributions implies that Theorem holds true.
\end{proof}

\begin{example}
Let $\lambda=UUDDUU$ and $T=A(\lambda)$.
Then, $L_{0}$ is the natural labelling
\begin{eqnarray*}
L_{0}:=
\scalebox{0.8}{
\tikzpic{-0.5}{
\coordinate
	child{coordinate (c0)
		child{coordinate (c1)}
		child[missing]
	     }
	child{coordinate (c2)};
\draw (c0)node{$-$};
\node[anchor=south east] at($(0,0)!.6!(c0)$){\scalebox{1.25}{$1$}};
\node[anchor=south east] at($(c0)!.6!(c1)$){\scalebox{1.25}{$2$}};
\node[anchor=south west] at($(0,0)!.6!(c2)$){\scalebox{1.25}{$3$}};
\node at ($(0,0)!.5!(c2)$){\scalebox{1.25}{$\bullet$}}; 
}}
\end{eqnarray*}
The modified post-order word is $\sigma_{0}=\mathrm{post}(L_{0})=21\underline{3}$.
We have twelve signed permutations $\sigma$'s satisfying $\sigma\ge_{L}\sigma_{0}$.
They are 
\begin{eqnarray*}
\begin{array}{c|cccccccccccc}
\sigma & 21\underline{3} & 2\underline{3}1 & \underline{3}21 & 2\underline{31} 
& \underline{3}2\underline{1} & 2\underline{13} & \underline{1}2\underline{3} 
& \underline{31}2 & \underline{13}2 & \underline{312} & \underline{132} 
& \underline{123} \\ \hline
\mathrm{Inv}(\sigma) & 2 & 3 & 4 & 4 & 5 & 5 & 6 & 6 & 7 & 7 & 8 & 9 
\end{array}
\end{eqnarray*}
The generating functions is 
\begin{eqnarray*}
P_{\lambda,\ast}^{I}=(1+q+2q^2+2q^3+2q^4+2q^5+q^6+q^7)=[4][6]/[2].
\end{eqnarray*}
\end{example}

The following Theorem is an analogue of Theorem \ref{thrm-A1}.
Let $N$ be the number of edges in $T$.
In $T$, we say an edge $e'$ is a child of an edge $e$ if there exists 
a unique sequence of edges 
$e=e_{0}\rightarrow e_{1}\rightarrow\cdots\rightarrow e_{n}=e'$ 
such that $e_{i+1}$ is below $e_{i}$ in $T$.
Given an edge $e$ in $T$, we define 
\begin{eqnarray*}
l(e):=
\begin{cases}
\#\{e' \ |\  \text{$e'$ is a child of $e$}\}+1, & \text{$e$ does not have a $\bullet$}, \\
2\#\{e' \ |\  \text{$e'$ is a child of $e$}\}+2, & \text{$e$ has a $\bullet$}, \\
\end{cases}
\end{eqnarray*}

\begin{theorem}
\label{thrm:BI-tree-fac}
We have 
\begin{eqnarray}
\label{P-fac2}
P^{I}_{\lambda,\ast}=
\frac{[2N]!!}{\prod_{e\in T}[l(e)]}.
\end{eqnarray}
\end{theorem}
\begin{proof}
We use the same notation as Lemma \ref{lemma:tree2}.
The root of $T$ has $p$ edges without $\bullet$ and at most 
one edge with $\bullet$. We call the edge with $\bullet$ $r_{p+1}$.
From Lemma \ref{lemma:tree2}, we have 
\begin{eqnarray*}
\frac{\mathrm{tree}_2(T,r_i)}{T}&=&\frac{[|T_{i}|+1]}{[2N]}, \\
\frac{\mathrm{tree}_2(T,r_{p+1})}{T}&=&
\frac{2(|T_{p+1}|+1)}{[2N]}. 
\end{eqnarray*}
We apply $\mathrm{tree}_2$ successively to $T$ and obtain 
Eqn.(\ref{P-fac2}).
\end{proof}

Let $T^{\mathrm{rev}}$ be a tree with following capacities.
A capacity on a leaf $l$ is the number of edges strictly right to 
the edge just above the leaf $l$.
Since we can put integers on $T$ satisfying (LS1) and (LS2), 
we denote by $L^{\mathrm{rev}}$ a labelling of $T$. 
We denote by $\mathcal{L}^{\mathrm{rev}}(\lambda)$ the set of possible
labellings $L^{\mathrm{rev}}$'s.
For a edge $e$ in $T$, we denote by $T_{e}$ a partial tree connected 
to the edge $e$ (the root of $T_{e}$ is connected to $e$ from bottom).
Let $c_{\mathrm{min}}$ be the smallest capacity appeared in $T_{e}$.
We define 
\begin{eqnarray*}
S^{\mathrm{rev}}(e):=c^{\mathrm{min}}_{e}+|T_{e}|+1.
\end{eqnarray*}
where $c^{\mathrm{min}}_{e}$ is the smallest capacity in a partial tree $T_{e}$.

\begin{theorem}
\label{thrm:rev3}
We have 
\begin{eqnarray}
\label{eqn:rev3}
T=\left(\sum_{L\in\mathcal{L}^{\mathrm{rev}}(\lambda)}q^{\mathrm{deg}_{9}(L)}\right)
\cdot
\prod_{e\in \mathcal{E}\setminus\mathcal{E}_{\bullet}}(1+q^{S^{\mathrm{rev}}(e)}).
\end{eqnarray}
\end{theorem}
\begin{proof}
Since $T$ does not have arrows, Theorem \ref{thrm:GFLS} can be reduced to 
\begin{eqnarray*}
T=\left(\sum_{L^{\mathrm{LS}}}q^{\mathrm{deg}_{9}(L^{\mathrm{LS}})}\right)
\cdot \prod_{e\in\mathcal{E}\setminus\mathcal{E}_{\bullet}}(1+q^{S'(e)}).
\end{eqnarray*}
By comparing the post-order from right in a tree $T$ with a capacity of $T^{\mathrm{rev}}$, 
we have $S'(e)=S^{\mathrm{rev}}(e)$.
Together with Lemma \ref{lemma:LSrev}, we obtain Eqn.(\ref{eqn:rev3}).
\end{proof}

By a straightforward calculation, we have the following lemma.
\begin{lemma}
\label{lemma:ram2}
Let $M=\sum_{i=1}^{r}M_i$.
We have 
\begin{eqnarray*}
\genfrac{[}{]}{0pt}{}{M+c}{M_1,M_2,\ldots,M_r,c}
\prod_{i=c+1}^{M+c}(1+q^{i})
=\genfrac{[}{]}{0pt}{}{M+c}{M_1,M_2,\ldots,M_r,c}_{q^2}
\prod_{i=1}^{r}\prod_{j=1}^{M_i}(1+q^{i}).
\end{eqnarray*}
\end{lemma}

Let $\lambda$ be a ballot path and $\lambda_0$ be a ballot path consisting 
of only $U$'s. 
Let $L$ be a labelling of type LS on the tree $A(\lambda/\lambda_{0})$ and 
$\mathcal{L}(\lambda)$ be the set of labellings of type LS.
Given a labelling $L$, we denote by $n_{e}$ an integer on an edge $e$.
We define $N:=|\mathcal{E}|-|\mathcal{E}_{\bullet}|$.

For an edge $e$ in a tree $A(\lambda)$, we define $S(e)$ as one plus
the number of children of $e$. In other words, $S(e)$ is one plus 
the sum of distances from $e$ to leaves.
\begin{theorem}
\label{thrm:2deg}
We have 
\begin{eqnarray*}
T=\left(\sum_{L\in\mathcal{L}(\lambda)}q^{2\mathrm{deg}_{9}(L)}\right)
\cdot \prod_{e\in\mathcal{E}\setminus\mathcal{E}_{\bullet}}(1+q^{S(e)}).
\end{eqnarray*}
\end{theorem}
\begin{proof}
Starting with the expression Eqn.(\ref{eqn:rev3}) in Theorem \ref{thrm:rev3}, 
we apply Lemma \ref{lemma:ram2} to ramification points in the tree $T$.
Then we apply Lemma \ref{lemma:LSrev} to the obtained expression, which implies 
Theorem holds true.
\end{proof}

Given a tree $T$, $T$ is a concatenation of trees $T_{1},\ldots,T_{M}$ at their
roots.
We call $T_{i}$, $1\le i\le M$, the $i$-th block.
Note that only the $M$-th block contains edges with a dot.
Fix a natural labelling $L$ of a tree $T$ and let 
$n_e$ be an integer on an edge $e$.
We define $d_{e}$ and $\mathrm{deg}_{10}(L)$ as 
\begin{eqnarray*}
d_{e}&:=&2\#\{e'\ |\  n_{e'}<n_{e}, \text{$e'$ is strictly right to $e$ and in a different block}\} \\
&&\quad+\#\{e' \ |\ n_{e'}<n_{e}, \text{$e'$ is strictly right to $e$ and in the same block as $e$} \},
\end{eqnarray*}
and 
\begin{eqnarray*}
\mathrm{deg}_{10}(L):=\sum_{e}d_{e}.
\end{eqnarray*}
Let $N_{i}$, $1\le i\le M-1$, be the number of edges in the $i$-th block.

\begin{theorem}
\label{thrm:BI-hybrid}
We have 
\begin{eqnarray*}
T=\left(\sum_{L}q^{\mathrm{deg}_{10}(L)}\right)\prod_{i=1}^{M-1}\prod_{j=1}^{N_{i}}(1+q^{j})
\prod_{e\in T_{M}\cap(\mathcal{E}\setminus\mathcal{E}_{\bullet})}(1+q^{S^{\mathrm{rev}}(e)}).
\end{eqnarray*}
\end{theorem}

\begin{proof}
We start with the expression in Theorem \ref{thrm:rev3}.
Let $c_i$ be the smallest capacity in the partial tree $T_{i}$ for 
$1\le i\le M-1$.
From the definition of $d_e$, if $n_{e'}$ and $n_e$ are in the same 
block and $e'$ is strictly right to $e$, 
the weight for $n_{e'}<n_e$ is one.
We can resolve ramification points by using Lemma \ref{lemma:LSDyck}. 
By this operation, $T_{i}$ is a product of $P^{\mathrm{Dyck}}(T_{i})$
and a tree $T'_{i}$ with $|T_{i}|$ edges and without a ramification point.
We apply Lemma \ref{lemma:ram2} to each $T'_{i}$. The right hand side 
of Lemma \ref{lemma:ram2} implies that we have a weight two if 
$n_{e'}<n_e$ and $e'$ is strictly right to $e$ and in a different block.
This completes the proof.
\end{proof}

\subsection{Relation to BTS}
\label{sec:twoa-BTS}
In this subsection, we consider a tree $T$ without arrows.
We put an additional condition on $T$: an edge $e$ without a dot 
does not have a parent edge with a dot.

We denote by $\mathcal{E}$ the set of edges in $A(\lambda)$ and by 
$\mathcal{E}_{\bullet}$ 
the set of edges with a $\bullet$.
We define:
\begin{eqnarray*}
\mathrm{deg}_{11}(L)
&:=&\sum_{e\in(\mathcal{E}\setminus\mathcal{E}_{\bullet})}n_e
+2\sum_{e\in\mathcal{E}_{\bullet}}n_e.
\end{eqnarray*}
Recall that $\mathcal{L}(\lambda)$ is the set of labellings of 
type LS.
\begin{theorem}
\label{thrm:BallotBTS}
We have 
\begin{eqnarray*}
T=\left(\sum_{L\in\mathcal{L}(\lambda)}q^{\mathrm{deg}_{11}(L)}\right)
\cdot \prod_{i=1}^{N}(1+q^{i}).
\end{eqnarray*}
where $N$ is the number of edges without a dot in $T$.
\end{theorem}
\begin{proof}
We assume that $T$ is written as a concatenation of two trees $T_{1}$ and 
$T_{2}$ such that $T_{2}$ consists of edges with a dot and without a ramification.
Let $\mathcal{L}_{1}$ (resp. $\mathcal{L}_{2}$) be the set of labellings satisfying 
(LS1) and (LS2) on $T_{1}$ (resp. $T_{2}$).
We denote by $\mathrm{deg}_{11}(L_i)$, $i=1,2$ the sum of labels 
in $L_i$.
By applying Section \ref{sec:Fac-trees} to the generating function $T$,
we have 
\begin{eqnarray*}
T&=&P^{\mathrm{Dyck}}(T_{1})\cdot\genfrac{[}{]}{0pt}{}{|T_{1}|+|T_{2}|}{|T_{1}|}_{q^{2}}
\cdot\prod_{i=1}^{|T_{1}|}(1+q^{i}) \\
&=&\left(\sum_{L\in\mathcal{L}_1}q^{\mathrm{deg}_{9}(L_1)}\right)
\left(\sum_{L\in\mathcal{L}_2}q^{2\mathrm{deg}_{9}(L_2)}\right)
\cdot\prod_{i=1}^{|T_{1}|}(1+q^{i}) \\
&=&\left(\sum_{L\in\mathcal{L}(\lambda)}q^{\mathrm{deg}_{11}(L)}\right)
\cdot \prod_{i=1}^{N}(1+q^{i})
\end{eqnarray*}
where we have used Lemma \ref{lemma:LSDyck} in the second equality.
This completes the proof.
\end{proof}

Let $L$ be a natural labelling of a tree $A(\lambda)$ and  
$n_e$ be a label of an edge $e$.
$L^{\mathrm{max}}$ be the unique maximum labelling of LS type 
and $n^{\mathrm{max}}_{e}$ be a label of an edge $e$.
We construct a labelling $L'$ by replacing $n_e$ with 
$n'_{e}$ where 
\begin{eqnarray*}
n'_{e}:=\#\{e'\ |\  n_{e'}>n_e, \text{$e'$ is strictly left to $e$}\}.
\end{eqnarray*}
We obtain a labelling of LS type by $L^{\mathrm{LS}}:=L^{\mathrm{max}}-L'$, 
{\it i.e.} $n^{\mathrm{LS}}_{e}:=n^{\mathrm{max}}_{e}-n'_{e}$.
We denote by $\psi$ the map $\psi:L\rightarrow L^{\mathrm{LS}}$.
The following lemma is obvious from the definition of $\psi$. 
Recall the definition of the inversion for an inverse pre-order word 
introduced in Section \ref{sec:ppt}.
\begin{lemma}
\label{lemma:invL}
Let $L$ be a natural labelling of a tree $T$ and $\sigma$ be the inverse 
pre-order word for $L$.
We have 
\begin{eqnarray*}
\mathrm{inv}(\sigma)=\mathrm{deg}_{11}(L^{\mathrm{max}})-\mathrm{deg}_{11}(\psi(L)).
\end{eqnarray*}
\end{lemma}

\begin{cor}
\label{cor:BTS}
Let $\lambda$ is a ballot path for a tree $T$ and $\sigma$ be an inverse pre-order 
word of the tree $T$.
Then, we have 
\begin{eqnarray*}
T&=&\left(\sum_{\sigma}q^{\mathrm{inv}(\sigma)}\right)\cdot 
\prod_{i=1}^{N}(1+q^{i}) \\
&=&\left(\sum_{\sigma}q^{\mathrm{art}(\mathrm{BTS}(\lambda,\sigma))}\right)\cdot 
\prod_{i=1}^{N}(1+q^{i}),
\end{eqnarray*}
where $N$ is the number of edges without a dot in $T$.
\end{cor}
\begin{proof}
The first equality in Corollary follows from Theorem \ref{thrm:BallotBTS} and 
Lemma \ref{lemma:invL}.
The second equality follows from Theorem \ref{thrm:BTSinpo}.
\end{proof}

\bibliographystyle{amsplainhyper} 
\bibliography{biblio}

\end{document}